\PassOptionsToPackage{colorlinks,linkcolor={blue},citecolor={blue},urlcolor={red},breaklinks=true}{hyperref}
\documentclass[acmsmall,screen,final,nonacm]{acmart}
\setcopyright{none}

\usepackage{ifthen}
\newboolean{arxiv}
\setboolean{arxiv}{true} %
\newcommand{\ifarx}[2]{\ifthenelse{\boolean{arxiv}}{#1}{#2}}

\usepackage{booktabs}   %
\usepackage{subcaption} %

\usepackage{afterpage}

\usepackage{stel-common}

\usepackage[capitalize,noabbrev,nameinlink]{cleveref} %

\providecommand{\catname}{\mathbf} 
\providecommand{\clsname}{\mathcal}
\providecommand{\oname}[1]{{\operatorname{\mathsf{#1}}}}

\def\defcatname#1{\expandafter\def\csname B#1\endcsname{\catname{#1}}}
\def\defcatnames#1{\ifx#1\defcatnames\else\defcatname#1\expandafter\defcatnames\fi}
\defcatnames ABCDEFGHIJKLMNOPQRSTUVWXYZ\defcatnames

\def\defclsname#1{\expandafter\def\csname C#1\endcsname{\clsname{#1}}}
\def\defclsnames#1{\ifx#1\defclsnames\else\defclsname#1\expandafter\defclsnames\fi}
\defclsnames ABCDEFGHIJKLMNOPQRSTUVWXYZ\defclsnames

\def\defbbname#1{\expandafter\def\csname BB#1\endcsname{{\bm{\mathsf{#1}}}}}
\def\defbbnames#1{\ifx#1\defbbnames\else\defbbname#1\expandafter\defbbnames\fi}
\defbbnames ABCDEFGHIJKLMNOPQRSTUVWXYZ\defbbnames

\def\Set{\catname{Set}}

\DeclareOldFontCommand{\bf}{\normalfont\bfseries}{\mathbf}

\providecommand{\id}{\mathsf{id}}

\providecommand{\xto}[1]{\,\xrightarrow{#1}\,}

\providecommand{\To}{\mathrel{\Rightarrow}}			           %

\providecommand{\dar}{\kern-1.2pt\operatorname{\downarrow}}	
\providecommand{\uar}{\kern-1.2pt\operatorname{\uparrow}}

\providecommand{\inl}{\oname{inl}}
\providecommand{\inr}{\oname{inr}}

\providecommand{\ev}{\oname{ev}}

\usepackage{stmaryrd}

\providecommand{\pacman}[1]{}					                     %

\newcommand{\undefine}[1]{\let #1\relax}					                       %

\providecommand{\mone}{{\text{\kern.5pt\rmfamily-}\mathsf{\kern-.5pt1}}}

\makeatletter
\def\mfix#1{\oname{#1}\@ifnextchar\bgroup\@mfix{}}	       %
\def\@mfix#1{#1\@ifnextchar\bgroup\mfix{}}			           %
\makeatother

\providecommand{\case}[3]{\mfix{case}{\mathbin{}#1}{of}{#2}{\kern-1pt;}{\mathbin{}#3}}

\DeclareMathSymbol{\mathinvertedexclamationmark}{\mathord}{operators}{'074}
\DeclareMathSymbol{\mathexclamationmark}{\mathord}{operators}{'041}
\makeatletter
\newcommand{\raisedmathinvertedexclamationmark}{%
  \mathord{\mathpalette\raised@mathinvertedexclamationmark\relax}%
}
\newcommand{\raised@mathinvertedexclamationmark}[2]{%
  \raisebox{\depth}{$\m@th#1\mathinvertedexclamationmark$}%
}
\makeatother

\newcommand{\appp}{\mathsf{app}}

\newcommand{\mybar}[3]{%
  \mathrlap{\hspace{#2}\overline{\scalebox{#1}[1]{\phantom{\ensuremath{#3}}}}}\ensuremath{#3}
}

\newcommand{\thickcdot}{\boldsymbol{\cdot}}

\newcommand{\rp}{{\mathbf{rp}}}
\newcommand{\jrp}{{\mathbf{jrp}}}
\renewcommand{\ne}{{\mathbf{ne}}}
\newcommand{\jne}{{\mathbf{jne}}}

\newcommand{\barf}{\mybar{0.6}{2pt}{f}}
\newcommand{\barh}{\mybar{0.6}{.5pt}{h}}
\newcommand{\barB}{\mybar{0.6}{1.65pt}{B}}
\newcommand{\barD}{\mybar{0.6}{1.65pt}{\D}}
\newcommand{\barR}{\mybar{0.6}{1.65pt}{R}}

\newcommand{\barSigmas}{\mybar{0.9}{0pt}{\Sigmas}}

\newcommand{\barF}{\mybar{0.6}{2pt}{F}}
\newcommand{\bard}{\mybar{0.5}{3pt}{d}}
\newcommand{\barH}{\mybar{0.6}{2pt}{H}}
\newcommand{\barSigma}{\mybar{0.9}{0pt}{\Sigma}}

\newcommand{\TP}{\mathsf{TP}}

\newcommand{\FRel}{\mathbf{FRel}}

\newcommand{\sqleq}{\sqsubseteq}

\newcommand{\ctx}{{\mathrm{ctx}}}

\newcommand{\st}{\mathsf{st}}

\newcommand{\hatP}{\hat{P}}

\newcommand{\qand}{\quad\text{and}\quad}
\newcommand{\qqand}{\qquad\text{and}\qquad}

\newcommand{\SKI}{\ensuremath{\mathbf{SKI}}\xspace}
\newcommand{\BCK}{\ensuremath{\mathbf{BCKI}}\xspace}
\newcommand{\pBCK}{\ensuremath{\mathbf{pBCKI}}\xspace}
\newcommand{\pSKI}{\ensuremath{\mathbf{pSKI}}\xspace}
\renewcommand{\L}{\mathcal{L}}

\newcommand{\B}{\mathbb{B}}
\newcommand{\CLat}{\mathbf{CLat}}

\newcommand{\Sigmas}{\Sigma^{\star}}

\newcommand{\ar}{\mathsf{ar}}

\renewcommand{\epsilon}{\varepsilon}

\renewcommand{\S}{{\mathcal{S}}}

\newcommand{\seq}{\subseteq}
\newcommand{\ol}{\overline}

\newcommand{\outl}{\mathsf{outl}}
\newcommand{\outr}{\mathsf{outr}}

\newcommand{\Qtl}{\mathbf{Qtl}}
\newcommand{\InQtl}{\mathbf{InQtl}}
\providecommand{\C}{}
\providecommand{\D}{}
\providecommand{\E}{}
\renewcommand{\C}{{\mathbb{C}}}
\renewcommand{\D}{{\mathcal{D}}}
\renewcommand{\E}{{\mathbb{E}}}
\renewcommand{\id}{{\mathsf{id}}}
\renewcommand{\Nat}{\mathds{N}}

\renewcommand{\phi}{\varphi}

\newcommand{\f}{\oname{f}}

\newcommand{\takeout}[1]{\empty}

\newcommand{\ini}{\iota}

\DeclareMathOperator{\Alg}{\mathbf{Alg}}

\renewcommand{\rho}{\varrho}

\newcommand{\opp}{\mathsf{op}}

\newcommand{\pullbackangle}[2][]{\arrow[phantom,to path={
                     -- ($ (\tikztostart)!1cm!#2:([xshift=8cm]\tikztostart) $)
                        node[anchor=west,pos=0.0,rotate=#2,
                        inner xsep = 0]
                        {\begin{tikzpicture}[minimum
                        height=1mm,baseline=0,#1]
    \draw[-] (0,0) -- (.5em,.5em) -- (0,1em);
                        \end{tikzpicture}}}]{}}

\usepackage{proof}
\renewcommand{\inference}[2]{\infer{~#2~}{~#1~}}

\usepackage{accents}

\usepackage{enumitem}
\setlist[enumerate,1]{label=(\arabic*),font=\normalfont,align=left,leftmargin=0pt,labelindent=0pt,listparindent=\parindent,labelwidth=0pt,itemindent=!,topsep=2pt,parsep=0pt,itemsep=2pt,start=1}
\setlist[enumerate,2]{label=(\alph*),font=\normalfont,labelindent=*,leftmargin=*,start=1}
\setlist[itemize]{labelindent=*,leftmargin=*}
\setlist[description]{labelindent=*,leftmargin=*,itemindent=-1 em}

\usepackage{ifdraft}

\renewcommand{\c}{\colon}

\usepackage{seqsplit}
\usepackage{xstring}
\usepackage{xcolor}
\usepackage{wasysym}

\usepackage{wrapfig}

\newcommand{\HOCoalg}{\mathbf{HOCoalg}}


\makeatletter
\newcommand{\pushright}[1]{\ifmeasuring@#1\else\omit\hfill$\displaystyle#1$\fi\ignorespaces}
\newcommand{\pushleft}[1]{\ifmeasuring@#1\else\omit$\displaystyle#1$\hfill\fi\ignorespaces}
\makeatother

\tikzstyle{shiftarr}=[
        rounded corners,%
        to path={--([#1]\tikztostart.center)
                     -- ([#1]\tikztotarget.center) \tikztonodes
                     -- (\tikztotarget)},
]

\tikzset{
    commutative diagrams/.cd,
    arrow style=tikz,
    diagrams={>=stealth},
    row sep=large,
    column sep = huge
}

\usetikzlibrary{decorations.pathmorphing}
\ifdraft{%
  \usepackage{hypcap}
  \setcounter{tocdepth}{2}
  \usepackage{showlabels} 
  
}{}

\usepackage[footnote,marginclue,nomargin,draft]{fixme}
\FXRegisterAuthor{hu}{ahu}{HU} %
\FXRegisterAuthor{sm}{asm}{SM} %
\FXRegisterAuthor{st}{ast}{ST} %
\FXRegisterAuthor{ls}{als}{LS} %
\FXRegisterAuthor{sg}{asg}{SG} %

\usepackage{xspace}

\theoremstyle{definition}

\newtheorem{defn}[theorem]{Definition} %
\newtheorem{rem}[theorem]{Remark} %

\usepackage{stackengine}
\stackMath
\newcommand\tsup[2][2]{%
 \def\useanchorwidth{T}%
  \ifnum#1>1%
    \stackon[-1.3ex]{\tsup[\numexpr#1-1\relax]{#2}}{\scalebox{2}[1]{$\mathchar"307E$}\kern-.5pt}%
  \else%
    \stackon[-1ex]{#2}{\scalebox{2}[1]{$\mathchar"307E$}\kern-.5pt}%
  \fi%
}

\usepackage{graphicx,scalerel}
\renewcommand\hat[1]{\hstretch{.71}{\widehat{\hstretch{1.4}{#1}}}}

\numberwithin{equation}{section}

\renewcommand{\xto}[1]{\mathrel{\raisebox{-.75pt}{$\xrightarrow{\;\smash{\raisebox{-1.5pt}{{\scriptsize $#1$}}\;}}$}}}

\let\xmpsto=\xmapsto
\renewcommand{\xmapsto}[1]{\xmpsto{~#1~}}

\DeclareFontFamily{U}{mathc}{}
\DeclareFontShape{U}{mathc}{m}{it}{<->s*[1.03] mathc10}{}
\DeclareMathAlphabet{\morph}{U}{mathc}{m}{it}

\renewcommand{\id}{{\mathsf{id}}}
\renewcommand{\inl}{{\mathsf{inl}}}
\renewcommand{\inr}{{\mathsf{inr}}}

\renewcommand{\outl}{{\mathsf{outl}}}
\renewcommand{\outr}{{\mathsf{outr}}}

\def\monoto{\rightarrowtail}
\renewcommand{\Nat}{\mathbb{N}}

\newcommand{\fset}{\mathbb{F}}

\newcommand{\gcat}{\mathbb{C}}

\tikzcdset{scale cd/.style={every label/.append style={scale=#1},
    cells={nodes={scale=#1}}}}

\newcommand{\wSigma}{
	\mathchoice
		{\boldsymbol{\Sigma}\kern-.61em{\boldsymbol\Sigma}}
		{\boldsymbol{\Sigma}\kern-.61em{\boldsymbol\Sigma}}
		{\boldsymbol{\Sigma}\kern-.45em{\boldsymbol\Sigma}}
		{\boldsymbol{\Sigma}\kern-.35em{\boldsymbol\Sigma}}
}

\newcommand{\app}{\,}

\theoremstyle{definition}
\newtheorem{notation}[theorem]{Notation}

\begin{document}\allowdisplaybreaks

\title[Higher-Order Behavioural Conformances via Fibrations]{Higher-Order Behavioural Conformances via Fibrations}         %

\author{Henning Urbat}
\authornote{Funded by the Deutsche Forschungsgemeinschaft (DFG, German
  Research Foundation) -- project numbers 470467389 and 569130867}  %
\orcid{0000-0002-3265-7168}             %
\affiliation{
  \institution{Friedrich-Alexander-Universität Erlangen-Nürnberg}            %
  \city{Erlangen}
  \country{Germany}                    %
}
\email{henning.urbat@fau.de}          %

\begin{abstract}
Coinduction is a widely used technique for establishing behavioural equivalence of programs in higher-order languages. In recent years, the rise of languages with quantitative (e.g.~probabilistic) features has led to extensions of coinductive methods to more refined types of behavioural conformances, most notably notions of behavioural distance. To guarantee soundness of coinductive reasoning, one needs to show that the behavioural conformance at hand forms a program congruence, i.e.~it is suitably compatible with the operations of the language. This is usually achieved by a complex proof technique known as \emph{Howe's method}, which needs to be carefully adapted to both the specific language and the targeted notion of behavioural conformance. We develop a uniform categorical approach to Howe's method that features two orthogonal dimensions of abstraction: (1) the underlying higher-order language is modelled by an \emph{abstract higher-order specification} (\emph{AHOS}), a novel and very general categorical account of operational semantics, and (2) notions of behavioural conformance (such as relations or metrics) are modelled  via \emph{fibrations} over the base category of an AHOS. Our main result is a fundamental congruence theorem at this level of generality: Under natural conditions on the categorical ingredients and the operational rules of a language modelled by an AHOS, the greatest behavioural (bi)conformance on its operational model forms a congruence. We illustrate our theory by deriving congruence of bisimilarity and behavioural pseudometrics for probabilistic higher-order languages.
\end{abstract}

\begin{CCSXML}
  <ccs2012>
  <concept>
  <concept_id>10003752.10010124.10010131.10010137</concept_id>
  <concept_desc>Theory of computation~Categorical semantics</concept_desc>
  <concept_significance>500</concept_significance>
  </concept>
  <concept>
  <concept_id>10003752.10010124.10010131.10010134</concept_id>
  <concept_desc>Theory of computation~Operational semantics</concept_desc>
  <concept_significance>500</concept_significance>
  </concept>
  </ccs2012>
\end{CCSXML}

\ccsdesc[500]{Theory of computation~Categorical semantics}
\ccsdesc[500]{Theory of computation~Operational semantics}

\keywords{Higher-Order Languages, Behavioural Distances, Howe's Method}

\maketitle

\section{Introduction}
One of the most fundamental, and most challenging, questions in the theory of higher-order languages concerns the \emph{behavioural equivalence} of programs: what precisely does it mean for two programs to `behave in the same way', and how does one actually prove it? In the world of operational semantics, the ubiquitous notion of program equivalence is \emph{contextual equivalence}~\cite{morris}. Two programs $p$ and $q$ are contextually equivalent if, whichever context (program with a hole) $C[\cdot]$ they are plugged into, the resulting programs $C[p]$ and $C[q]$ expose the same observable behaviour, e.g.\ they co-terminate, terminate with the same probability, or output the same values.

While the definition of contextual equivalence is simple and natural, it is somewhat difficult to work with: to prove that  $p$ and $q$ are contextually equivalent, one would need to analyse the observable behaviour of $C[p]$ and $C[q]$ for all (infinitely many) contexts at the same time. Therefore, efficient proof techniques are needed. One fundamental approach is given by \emph{coinduction}. Here, the idea is to view the set of programs as an \emph{applicative transition system}~\cite{Abramsky:lazylambda}, that is, a labelled transition system (LTS) whose states and labels are both given by programs. For instance, in the case of the untyped call-by-name $\lambda$-calculus, the term $\lambda x. t$ has the outgoing transitions $\lambda x. t \xto{e} t[e/x]$ for all $\lambda$-terms~$e$, reflecting the intuition that $\lambda x. t$ behaves like a function mapping the input $e$ to the output $t[e/x]$. Like every LTS, the applicative transition system on $\lambda$-terms comes with a notion of \emph{bisimilarity}~\cite{sangiorgi11} expressing behavioural equivalence of states. The key observation, due to \citet{Abramsky:lazylambda}, is that bisimilarity implies contextual equivalence. Thus, coinduction yields a sound proof technique: to show that $p$ and $q$ are contextually equivalent, it suffices to find a bisimulation containing them. Usually, this is much easier than proving contextual equivalence directly. 

The soundness of coinduction essentially amounts to the fact that the bisimilarity relation is a \emph{congruence}, that is, compatible with the operations of the language. The standard approach to establishing that congruence property is \emph{Howe's method}~\cite{DBLP:conf/lics/Howe89,DBLP:journals/iandc/Howe96}. It has proven extremely robust and has been adapted to derive soundness results for a wide range of languages, including languages with different forms of evaluation (call-by-name~\cite{pitts_2011}, call-by-value~\cite{pitts_2011}, call-by-push-value~\cite{gtu25}), typed languages~\cite{pitts1997operationally}, probabilistic languages (both discrete~\cite{dsa14} and continuous~\cite{dg19}), stateful languages~\cite{jr99,gmstu25}, languages with monadic computational effects~\cite{dgl17}, higher-order process calculi~\cite{ls15}, and many more. Additionally, the emergence of quantitative (e.g.~probabilistic) higher-order languages has led to quantitative extensions of Howe's method, where contextual equivalence is replaced with the more fine-grained notion of \emph{contextual distance} (measuring how much contexts can discriminate programs) and bisimilarity is generalized to \emph{behavioural pseudometrics}~\cite{cdl15,gavazzo18}.

In practice, the use of Howe's method is highly non-trivial. This is mainly for three reasons:

\begin{enumerate}
\item\label{issue1} Howe's method is typically applied on a per case basis, that is, its principles are reinvented from scratch for every individual language and carefully adapted to its features. Every modification of the language at hand (e.g.\ extending its type system, adding states, adding probabilistic features) essentially requires starting anew, or at least laboriously reworking existing proofs.
\item\label{issue2} Orthogonally, there is a wide spectrum of \emph{behavioural conformances}~\cite{bgkmfsw24} of interest. Is the targeted notion plain behavioural equivalence, or some kind of behavioural distance? In the latter case, are distances supposed to be rational numbers, or real numbers, or more complex objects, such as functions~\cite{dgy19}, or elements of a quantale~\cite{gavazzo18}? Does `congruence' mean non-expansivity of the language operations w.r.t.~some distance, or something else, like Lipschitz-continuity? The suitable type of conformance depends on the nature of the given language and the kind of reasoning tasks one aims to address, and again every choice calls for its own incarnation of Howe's method.
\item\label{issue3} Even for a fixed language and conformance type, congruence proofs based on Howe's method tend to be complex and subtle, requiring tedious case distinctions over the syntax and operational rules of the language, while at the same containing plenty of `bookkeeping' tasks that essentially work in the same way for all languages and bloat the argument. The high-level intuition \emph{why} the method works is not always easy to grasp. This difficulty can be partly attributed to the fact that coinduction is an inherently first-order concept that is now applied to a higher-order setting.
\end{enumerate}

\paragraph*{Contribution.} 
We present a novel \emph{categorical} approach to coinductive reasoning on higher-order languages that addresses all three issues discussed above at the same time.
\begin{enumerate}
\item One key reason for the need to reinvent Howe's method for every language is that there is simply no sufficiently universal (categorical) notion of \emph{higher-order language}. We introduce \emph{abstract higher-order specifications} ($\emph{AHOS}$), a general and flexible categorical formalism that presents the operational semantics of a higher-order language as a \emph{rule morphism} in a category~$\B$. Roughly, a rule morphism specifies the operational rules of a language by distributing its syntax (given by an endofunctor $\Sigma\colon \B\to \B$) over its behaviour (given by a mixed variance bifunctor $B\colon \B^\opp\times \B\to \B$). Every AHOS naturally induces an \emph{operational model}, a transition system ({higher-order coalgebra}) $\gamma\colon \Lambda\to B(\Lambda,\Lambda)$ on the object $\Lambda$ of program terms that runs programs according to the given rules. A related idea appears in the \emph{higher-order abstract GSOS} framework recently introduced by \citet{gmstu23}. The latter captures operational rules whose model is definable via \emph{structural induction}, typical for small-step operational semantics. However, it is not capable of handling the more powerful \emph{rule induction} needed for big-step semantics, which is, e.g., the standard semantics for coinductive reasoning on probabilistic languages. AHOS overcome this limitation by combining the principles of higher-order abstract GSOS with Rot's \emph{(first-order) monotone biGSOS}~\cite{rot19} framework.
\item To get a handle on the plethora of different types of behavioural conformance, we put a \emph{fibration} above the base category of an AHOS. Fibrations allow to model notions of behavioural equivalence and behavioural distance for coalgebras in a flexible and uniform way~\cite{hj98,bbkk18}. In fact, in the fibrational framework, moving from qualitative to quantitative reasoning just amounts to switching from a fibration of relations to a fibration of suitable fuzzy relations. 
\item We develop Howe's method at the generality of higher-order languages modelled by AHOS and fibrational behavioural conformances. This enables us to establish a strikingly general congruence result (\Cref{thm:app-struct-cong-bisim}), which informally states:

\smallskip\noindent
\emph{If the operational rules specifying a higher-order language are sufficiently monotone and continuous, then the greatest behavioural (bi)conformance on its operational model forms a congruence.}
\smallskip\noindent 

Here the \emph{greatest behavioural (bi)conformance} is the fibrational abstraction of (bi)similarity and its quantitative analogues. The importance of this abstract result is that it systematically reduces the congruence property of a language to a number of simple conditions on the categorical ingredients and the operational rules of the language, while suppressing the `generic', language-independent aspects of ad hoc congruence proofs. In this way, the complexity of deriving congruence results is dramatically reduced.  We will demonstrate this point by giving short and simple congruence proofs for probabilistic higher-order languages corresponding to probabilistic $\lambda$-calculi~\cite{dsa14,cdl15}.
\end{enumerate}

\paragraph*{Related Work} The need for categorical tools to address the complexities of Howe's method has been identified by several researchers, and systematic categorical accounts of Howe's method have been an active topic of investigation in recent years. However, unlike our present paper, all approaches so far were limited to reasoning about behavioural \emph{equivalence}, rather than quantitative notions of behavioural distance or (fibrational) conformances in general. Moreover, they model languages in categorical frameworks that are more specific than our very general AHOS setup.

\citet{UrbatTsampasEtAl23} develop Howe's method at the generality of the higher-order abstract GSOS framework~\cite{gmstu23}, which models languages with a small-step operational semantics by suitable dinatural transformations. They prove a congruence result for the {similarity} relation on operational models. The non-trivial step to \emph{bi}similarity is missing. In the present paper, we achieve this step by developing Howe's \emph{transitive closure trick} at fibrational generality.

\citet{DBLP:conf/lics/BorthelleHL20} and \citet{DBLP:journals/lmcs/HirschowitzL22} model operational rules 
as endofunctors on a specific presheaf category of \emph{transition systems} over models
of a signature endofunctor, and the initial algebra for the rule
endofunctor represents the induced transition system for the given
semantics. In this setting, they achieve a congruence result for bisimilarity. Their categorical setup is substantially different, and technically more complex, than AHOS. 

Dal Lago et al.~\cite{DBLP:conf/lics/LagoGL17} develop Howe's method and prove a congruence result for bisimilarity in a setting of call-by-value $\lambda$-calculi
with algebraic effects, based on the theory of relators. Their notion
of a \emph{computational} $\lambda$-calculus is parametrized over a
signature $\Sigma$ and a monad $T$ on the category of sets, representing syntax and
effects of the language. 

While fibrations have not appeared in the literature on coinduction for higher-order languages, their general usefulness for operational reasoning has been recognized. For instance, in recent work, \citet{dg24} model notions of operational logical relations in a fibrational setting.

\section{Preliminaries}\label{sec:category-theory}

Our abstract approach to behavioural conformances for higher-order languages uses some basic concepts from category theory~\cite{mac2013categories}, in particular, algebras for a functor (which capture sets of program terms) and coalgebras for a functor (which capture applicative transition systems).  In the following we review some categorical terminology used throughout our paper.

\paragraph*{Notation}
For objects
$X_1, X_2$ of a category $\C$, we denote their product by $X_1\times X_2$, the projections by $\outl\colon X_1\times X_2\to X_1$ and $\outr\colon X_1\times X_2\to X_2$, and
the pairing of morphisms $f_i\c X\to X_i$, $i=1,2$, by $\langle f_1, f_2\rangle\c
X\to X_1\times X_2$. We write
$X_1+X_2$ for the coproduct, $\inl\c X_1\to X_1+X_2$ and
$\inr\c X_2\to X_1+X_2$ for its injections, $[g_1,g_2]\c X_1+X_2\to X$ for the copairing of morphisms $g_i\colon X_i\to X$,
$i=1,2$, and $\nabla=[\id_X,\id_X]\colon X+X\to X$ for the codiagonal. For every category $\C$, we let $\C^{(2)}$ denote the category that has the same objects as $\C$, and as morphisms $(f,g)\colon X\to Y$ all pairs of morphisms from $X$ to $Y$ in $\C$. Composition and identities are formed componentwise. 
 Note that $(\B\times \C)^{(2)}\cong \B^{(2)}\times \C^{(2)}$ and $(\C^\opp)^{(2)}\cong(\C^{(2)})^\opp$. Every functor $F\colon \B\to \C$ induces a functor $F^{(2)}\colon \B^{(2)}\to \C^{(2)}$ given by 
$F^{(2)}X=FX$ and $F^{(2)}(f,g)=(Ff,Fg)$. The category $\C$ forms a non-full subcategory of $\C^{(2)}$ via the embedding functor
$I\colon \C\monoto \C^{2}$ given by $X\mapsto X$ and $f\mapsto (f,f)$.

\paragraph*{Algebras}
Let $\Sigma$ be an endofunctor on a category $\gcat$. A \emph{$\Sigma$-algebra} $(A,a)$ consists of an object~$A$ (the \emph{carrier} of the
algebra) and a morphism $a\colon \Sigma A\to A$ (its \emph{structure}). A
\emph{morphism} from $(A,a)$ to a $\Sigma$-algebra $(B,b)$ is a morphism
$h\colon A\to B$ of~$\gcat$ such that $h\circ a = b\circ \Sigma h$. Algebras
for $\Sigma$ and their morphisms form a category, which we denote by $\Alg(\Sigma)$. An \emph{initial algebra} is an initial object of $\Alg(\Sigma)$; if it exists, we denote its carrier by $\mu \Sigma$ 
and its structure by ${\ini\colon \Sigma(\mu \Sigma) \to \mu \Sigma}$. The structure $\ini$ is an isomorphism in $\C$.
A \emph{free $\Sigma$-algebra} generated by an object $X$ of $\gcat$ is a
$\Sigma$-algebra $(\Sigma^{\star}X,\iota_X)$ together with a morphism
$\eta_X\c X\to \Sigma^{\star}X$ of~$\gcat$ such that for every algebra $(A,a)$
and every morphism $h\colon X\to A$ in $\gcat$, there exists a unique
$\Sigma$-algebra morphism $h^\#\colon (\Sigma^{\star}X,\iota_X)\to (A,a)$
such that $h=h^\#\circ \eta_X$; the morphism $h^\#$ is called the \emph{free
  extension} of $h$. If a free $\Sigma$-algebra exists on every object $X\in \C$, then their formation extends to a monad $\Sigma^\star\colon \C\to \C$, the \emph{free monad} generated by $\Sigma$. For every $\Sigma$-algebra $(A,a)$ we write
$\hat{a} \colon \Sigma^{\star} A \to A$ for the free extension of $\id_A\c A\to A$.

A prime example of functor algebras are algebras for a
signature. An \emph{(algebraic) signature} consists of a set~$\Sigma$
of \emph{operation symbols} and a map $\ar\colon \Sigma\to \Nat$
associating to every $\f\in \Sigma$ its \emph{arity}. Operation symbols of arity $0$ are called \emph{constants}. Every
signature~$\Sigma$ induces an endofunctor on the category $\Set$ of sets and functions, denoted by the
same letter $\Sigma$ and defined by $\Sigma X = \coprod_{\f\in \Sigma} X^{\ar(\f)}$. (Such endofunctors are called \emph{polynomial functors}.) An algebra for the functor $\Sigma$ is
precisely an algebra for the signature~$\Sigma$: a set $A$ with an operation $\f^A\colon A^n\to A$ for every $n$-ary operation symbol $\f\in \Sigma$. Morphisms of $\Sigma$-algebras are
maps respecting the algebraic structure. Given a set $X$ of
variables, the free algebra $\Sigmas X$ is the $\Sigma$-algebra of
$\Sigma$-terms with variables from~$X$; more precisely, $\Sigmas X$ is inductively defined by $X\seq \Sigmas X$ and $\f(t_1,\ldots,t_n)\in \Sigmas X$ for all $\f\in \Sigma$ of arity $n$ and $t_1,\ldots,t_n\in \Sigmas X$.
The free
algebra on the empty set is the initial algebra~$\mu \Sigma$; it is
formed by all \emph{closed terms} of the signature.  For every
$\Sigma$-algebra $(A,a)$, the induced map
$\hat{a}\colon \Sigmas A \to A$ evaluates terms in~$A$.

\paragraph*{Coalgebras}
Dually to the notion of algebra, a \emph{coalgebra} for an endofunctor $B$ on $\C$ is a pair $(C,c)$ consisting of an object $C$ (the
\emph{state space}) and a morphism $c\colon C\to BC$ (its
\emph{structure}). A \emph{cofree $B$-coalgebra} generated by an object $X$ of $\gcat$ is a
$B$-coalgebra $(B^{\infty}X,\tau_X)$ together with a morphism
$\epsilon_X\c B^{\infty}X\to X$ of~$\gcat$ such that for every coalgebra $(C,c)$
and every morphism $h\colon C\to X$ in $\gcat$, there exists a unique
$B$-coalgebra morphism $\hat h\colon (C,c)\to (B^\infty X,\tau_X)$
such that $h=\epsilon_X\circ {\hat{h}}$; the morphism $\hat h$ is called the \emph{cofree
  extension} of $h$. (We use the same notation $\widehat{(-)}$ for free and cofree extensions for symmetry; the kind of extension will always be clear from the context.)
If a cofree coalgebra exists on every $X\in \C$, their formation yields a comonad $B^\infty\colon \C\to \C$, the \emph{cofree comonad} generated by $B$. 

Informally, coalgebras are abstractions of transition systems, and $B^\infty X$ consists of all possible unravelled behaviours of $B$-coalgebras whose states are colored using colors from $X$. For example, for the polynomial functor $BC = C\times C + 1$ on $\Set$, a coalgebra $c\colon C\to C\times C+1$ corresponds to a deterministic labelled transition system with labels $\{l,r\}$ where every state either has a single outgoing $l$-transition and a single outgoing $r$-transition, or is terminating. The cofree coalgebra $B^\infty X$ is given by the set of all (possibly infinite) ordered binary trees with $X$-colored nodes.

A \emph{higher-order coalgebra} for a mixed variance bifunctor $B\colon \C^\opp\times \C\to \C$ is a pair $(C,c)$ of an object $C$ and a morphism $c\colon C\to B(C,C)$, that is, a coalgebra for the endofunctor $B(C,-)$ with state space $C$. We write $\HOCoalg_C(B)$ for the set of all higher-order coalgebras with state space $C$. Higher-order coalgebras can be regarded as an abstract capturing of \emph{applicative transition systems}~\cite{Abramsky:lazylambda}, i.e.~labelled transition systems where both the states and the transition labels are given by programs of some higher-order language; see the next section for examples. Bisimulations and behavioural metrics of higher-order programs are studied in terms of such transition systems.

\paragraph*{Probability Distributions}
A \emph{countable probability distribution} on a set $X$ is a map $\phi\colon X\to [0,1]$ whose support $\{ x\in X\mid \phi(x)\neq 0 \}$ is a countable set and that satisfies $\sum_{x\in X} \phi(x) = 1$. We represent a distribution $\phi$ as a formal sum $\phi=\sum_{i\in I} p_i\cdot x_i$ where $I$ is a countable index set, $x_i\in X$ for all $i\in I$, and $\phi(x)=\sum_{x_i=x} p_i$ for all $x\in X$. For $\phi,\psi\in \D X$ and $p\in [0,1]$ we write $p\cdot \phi + (1-p)\cdot \psi\in \D X$ for the distribution  $x\mapsto p\cdot \phi(x)+(1-p)\cdot \psi(x)$. The \emph{(countable) distribution functor} $\D\colon \Set\to\Set$ associates to each set $X$ the set $\D X$ of countable probability distributions on $X$, and to each map $f\colon X\to Y$ the map $\D f \colon \D X\to \D X$ given by $\D f(\sum_{i\in I} p_i\cdot x_i) =  \sum_{i\in I} p_i\cdot f(x_i)$. The functor $\D$ extends to a monad with unit $\eta\colon X\to \D X$ and multiplication $\mu\colon \D\D X\to \D X$ given by
\[
\eta(x)=1\cdot x\qqand
 \mu(\sum_{i\in I} p_i\cdot \phi_i)=\sum_{i\in I}\sum_{j\in J_i} p_i\cdot p_{i,j}\cdot x_{i,j} \quad\text{where}\quad \phi_i=\sum_{j\in J_i} p_{i,j}\cdot x_{i,j}. 
\]
Finally, we let $\st\colon \D X\times Y\to \D(X\times Y)$ denote the (right) strength of the monad $\D$, given by
\[ \st(\sum_{i\in I} p_i\cdot x_i, y) = \sum_{i\in I} p_i\cdot (x_i,y). \]

\section{Probabilistic Combinatory Logic}\label{sec:prob-comb-logic}
To illustrate the developments of this paper, we introduce as running examples two simple probabilistic higher-order languages that take the form of combinatory logics~\cite{hindley2008lambda}. Recall that combinatory logics serve as alternative presentations of $\lambda$-calculi that avoid explicit notions of variables, binding, $\alpha$-renaming, and substitution. This technical simplification allows us to model the syntax and operational semantics of our languages within the category of sets and focus on the key aspects of our work. More complex higher-order languages, featuring e.g.\ variables, types, store, can be modelled by moving froms sets to suitable presheaf categories (cf.\ \Cref{sec:conclusion} for a discussion). 

Our two example languages, $\pSKI$ and $\pBCK$, are probabilistic extensions of the well-known combinatory logics $\SKI$ and $\BCK$~\cite{hindley2008lambda}, respectively. $\SKI$ is computationally equivalent to the untyped call-by-name $\lambda$-calculus, while $\BCK$ is computationally equivalent to the \emph{affine} untyped call-by-name $\lambda$-calculus. The latter is the restriction of the $\lambda$-calculus to \emph{affine} $\lambda$-terms, where in no subterm a variable is allowed to occur freely more than once. For instance, $\lambda x.\lambda y. x\app y$ is an affine $\lambda$-term, but $\lambda x.\, x\app x$ is not, because the subterm $x\, x$ has two free occurrences of $x$.
 
\subsection{\SKI}
As a warm-up, we first recall the deterministic \SKI calculus. Its terms are formed over the signature
\begin{equation}\label{eq:sigma-SKI} \Sigma_\SKI = \{\, S:0,\, S':1,\, S'':2,\, K:0,\, K':1,\, I:0,\, \Omega:0,\,\appp:2\,\}, \end{equation}
with arities as indicated. We let $\Lambda_\SKI$ denote the set of closed $\Sigma_\SKI$-terms (i.e.\ $\Lambda_\SKI=\mu\Sigma_\SKI$). The intended meaning of the above operations is as follows. The binary operator $\appp$ corresponds to function application; accordingly, we write $t\, s$ for $\mathsf{app}(t,s)$. 
The constant $\Omega$ represents a non-terminating computation. The constants $S$, $K$, $I$ represent the $\lambda$-terms 
$S=\lambda x.\,\lambda y.\,\lambda z.\, (x\, z)\, (y\, z)$,  $K=\lambda x.\,\lambda y.\, x$, and $I=\lambda x.\,x$.
Finally, the primed versions of $S$ and $K$ correspond to applications of the respective $\lambda$-terms: the terms $S'(t)$, $S''(t,s)$, and $K'(t)$ behave like $S\, t$, $S\, t\, s$, and $K\, t$, respectively. Our presentation of combinatory logics using the auxiliary primed operators follows \citet{GianantonioRPO} and leads to a consistent operational semantics where all combinators accept exactly one argument, which is well-suited for reasoning about behavioural conformances. In particular, one cannot tell combinators apart just based on the number of their arguments.

The (big-step) operational semantics of \SKI consists of judgements of the form \[t\Downarrow f \qquad \text{where}\qquad \text{$t\in {\Lambda_{\SKI}}$}\qand \text{$f\in{\Lambda_{\SKI}}^{\Lambda_{\SKI}}$},\] specified by the rules of \Cref{fig:skirules-untyped1} with $s,t,u$ ranging over $\Lambda_\SKI$.
\begin{figure*}[t]
  \begin{gather*}
    \inference{}{S\Downarrow (t\mapsto S'(t))}
    \qquad \inference{}{S'(t)\Downarrow (s\mapsto S''(t,s))} \qquad \inference{}{S''(t,s)\Downarrow  (u\mapsto (t\, u)\, (s\, u))}  
  \\[1ex]
\inference{}{K\Downarrow  (t\mapsto K'(t))} \qquad
 \inference{}{K'(t)\Downarrow (s\mapsto t) } \qquad
  \inference{}{I\Downarrow  (t\mapsto t)} \qquad
   \inference{t\Downarrow f\quad f(s)\Downarrow g}{t\, s \Downarrow g}
  \end{gather*}
  \caption{Operational semantics of \SKI.}
  \label{fig:skirules-untyped0}
\end{figure*}
 Informally, $t\Downarrow f$ means that the program $t$ normalizes and computes the function $f\in {\Lambda_{\SKI}}^{\Lambda_{\SKI}}$. If there is no such $f$ (e.g.\ for $t=\Omega$), then $t$ diverges. To make this precise, we consider the behaviour bifunctor
\begin{equation}\label{eq:beh-pBCK-untyped0} B_0\colon \Set^\opp \times \Set \to \Set \qquad \text{given by}\qquad  B_0(X,Y)=\{\bot\}+Y^{X}.\end{equation}
We regard $B_0({\Lambda_{\SKI}},{\Lambda_{\SKI}})$ as a flat DCPO with bottom element $\bot$. Then the set $\HOCoalg_{\Lambda_\SKI}(B_0)$ of higher-order coalgebras on ${\Lambda_{\SKI}}$ is also a DCPO with the pointwise partial order $\gamma\leq \gamma'$ iff $\gamma(t)\leq \gamma'(t)$ for all $t\in {\Lambda_{\SKI}}$. Given a higher-order coalgebra $\gamma\colon \Lambda_\SKI \to B_0(\Lambda_\SKI,\Lambda_\SKI)$ we write $t\Downarrow f$ if $\gamma(t)=f$, and we say that $\gamma$ \emph{satisfies} the rules of \Cref{fig:skirules-untyped1} if whenever there are $\gamma$-transitions matching the premises of some rule, then there is a $\gamma$-transition given by its conclusion.
The \emph{operational model} of \SKI is the \emph{least} higher-order coalgebra
\begin{equation}\label{eq:op-model-pSKI-untyped0}
\gamma_\SKI\colon \Lambda_{\SKI} \to B_0(\Lambda_\SKI,\Lambda_\SKI)
\end{equation}
satisfying the rules of \Cref{fig:skirules-untyped1}. Note that the least coalgebra $\gamma_{\SKI}$ exists by the Knaster-Tarski theorem\footnote{Every monotone endomap on a DCPO with $\bot$ has a least fixed point.}: it is the least fixed point of the monotone map 
\begin{equation}\label{eq:Gamma-SKI} \Gamma\colon \HOCoalg_{\Lambda_{\SKI}}(B_0)\to\HOCoalg_{\Lambda_{\SKI}}(B_0) \end{equation}
that sends a coalgebra $\gamma$ to the coalgebra $\Gamma(\gamma)$ induced by the rules of \Cref{fig:skirules-untyped1}; for instance, $\Gamma(\gamma)(S)= (t\mapsto S'(t))$ and if $\gamma(t)=f$ and $\gamma(f(s))=g$ where $f,g\in \Lambda_\SKI^{\Lambda_\SKI}$, then $\Gamma(\gamma)(t\, s)=g$.

\subsection{\pSKI}
The probabilistic extension of $\SKI$ adds a binary operator $\oplus$ corresponding to probabilistic choice: the term $t\oplus s$ behaves like $t$ or $s$ with probability $\frac{1}{2}$ each. Thus, the signature of $\pSKI$ is given by
\begin{equation}\label{eq:sigma-pSKI} \Sigma_\pSKI = \{\, S:0,\, S':1,\, S'':2,\, K:0,\, K':1,\, I:0,\, \Omega:0,\,\appp:2,\, \oplus:2 \,\}. \end{equation}

The (big-step) operational semantics of \pSKI consists of judgements of the form \[t\Downarrow \phi \qquad \text{where}\qquad \text{$t\in {\Lambda_{\pSKI}}$}\qand \text{$\phi\in \D(\{\bot\}+{\Lambda_{\pSKI}}^{\Lambda_{\pSKI}})$},\] specified by the rules of \Cref{fig:skirules-untyped1} with $s,t,u$ ranging over $\Lambda_\pSKI$.
\begin{figure*}[t]
  \begin{gather*}
    \inference{}{S\Downarrow 1\cdot(t\mapsto S'(t))}
    \qquad \inference{}{S'(t)\Downarrow 1\cdot(s\mapsto S''(t,s))} \qquad \inference{}{S''(t,s)\Downarrow 1\cdot (u\mapsto (t\, u)\, (s\, u))}  
  \\[1ex]
\inference{}{K\Downarrow 1\cdot (t\mapsto K'(t))} \qquad
 \inference{}{K'(t)\Downarrow 1\cdot(s\mapsto t) } \qquad
  \inference{}{I\Downarrow 1\cdot (t\mapsto t)} 
 \\[1ex]
\inference{}{\Omega\Downarrow  1\cdot \bot}
\qquad
     \inference{t\Downarrow \phi\quad s\Downarrow \psi}{t \oplus s\Downarrow  \frac{1}{2}\cdot \phi + \frac{1}{2}\cdot \psi}  
\qquad
   \inference{t\Downarrow p\cdot \bot + \sum_i p_i\cdot f_i\quad f_i(s)\Downarrow \sum_{j\in J_i} p_{i,j}\cdot b_{i,j}  \, (i\in I)}{t\, s \Downarrow p\cdot \bot + \sum_{i\in I}\sum_{j\in J_i} p_i\cdot p_{i,j}\cdot b_{i,j} }
  \end{gather*}
  \caption{Operational semantics of \pSKI.}
  \label{fig:skirules-untyped1}
\end{figure*}
 Informally, $t\Downarrow p\cdot \bot + \sum_i p_i\cdot f_i$ means that the program $t$ diverges with probability $p$, and computes the function $f_i\in {\Lambda_{\pSKI}}^{\Lambda_{\pSKI}}$ with probability $p_i$. To make this precise, we extend the behaviour bifunctor \eqref{eq:beh-pBCK-untyped0} to
\begin{equation}\label{eq:beh-pBCK-untyped} B\colon \Set^\opp \times \Set \to \Set \qquad \text{given by}\qquad  B(X,Y)=\D(\{\bot\}+Y^{X}).\end{equation}
The set $B({\Lambda_{\pSKI}},{\Lambda_{\pSKI}})$ again carries a DCPO structure given by $\phi\leq\psi$ iff $\phi(f)\leq \psi(f)$ for all $f\in {\Lambda_{\pSKI}}^{\Lambda_{\pSKI}}$, and this yields a DCPO structure on the set $\HOCoalg_{\Lambda_{\pSKI}}(B)$ by pointwise extension. Given a higher-order coalgebra $\gamma\colon \Lambda_\pSKI \to B(\Lambda_\pSKI,\Lambda_\pSKI)$ we write $t\Downarrow \phi$ if $\gamma(t)=\phi$, and we say that $\gamma$ \emph{satisfies} the rules of \Cref{fig:skirules-untyped1} if whenever there are $\gamma$-transitions matching the premises of some rule, then there is a $\gamma$-transition given by its conclusion.
The \emph{operational model} of \pSKI is the \emph{least} higher-order coalgebra
\begin{equation}\label{eq:op-model-pSKI-untyped}
\gamma_\pSKI\colon \Lambda_{p\SKI} \to B(\Lambda_\pSKI,\Lambda_\pSKI)
\end{equation}
satisfying the rules of \Cref{fig:skirules-untyped1}. This corresponds to the standard approximation semantics for probabilistic $\lambda$-calculi~\cite{dallago_zorzi12}. This is the least fixed point of the monotone map 
\begin{equation}\label{eq:Gamma-pSKI} \Gamma\colon \HOCoalg_{\Lambda_{\pSKI}}(B)\to  \HOCoalg_{\Lambda_{\pSKI}}(B) \end{equation}
that sends a coalgebra $\gamma$ to the coalgebra $\Gamma(\gamma)$ induced by the rules of \Cref{fig:skirules-untyped1}; for instance, if $\gamma(t)=\phi$ and $\gamma(s)=\psi$, then $\Gamma(\gamma)(t\oplus s)=\frac{1}{2}\cdot \phi + \frac{1}{2}\cdot \psi$. Monotonicity of $\Gamma$ amounts to the observation that the rules of \Cref{fig:skirules-untyped1} are monotone: whenever the premises of a rule are grown according to the order $\leq$, then the conclusion also grows. For instance, in the case of the rule for $\oplus$, if $\phi\leq \phi'$ and $\psi\leq\psi'$, then $\frac{1}{2}\cdot \phi + \frac{1}{2}\cdot \psi\leq \frac{1}{2}\cdot \phi' + \frac{1}{2}\cdot \psi'$.

A natural notion of program equivalence between $\pSKI$-terms is given by \emph{contextual equivalence}~\cite{dsa14}. A \emph{context} is a $\Sigma_\pSKI$-term in the single variable `$\cdot$' (the `hole'), which appears at most once in the term. We denote a context by $C[\cdot]$ and write $C[t]$ for the outcome of substituting a term $t\in \Lambda$ for the hole. 
\emph{Contextual equivalence} is the equivalence relation on $\Lambda_\pSKI$ defined by
\[ t\approx^\pSKI_\ctx s \iff \forall C[\cdot].\, \gamma_\pSKI(C[t])(\bot)=\gamma_\pSKI(C[s])(\bot). \]
Thus, two programs are equivalent if in every context they terminate with the same probability. Our theory of congruence developed below will provide coinductive methods to reason about $\approx_\ctx^\pSKI$.

\subsection{\pBCK}
The terms of $\pBCK$ are formed over the signature
\begin{equation}\label{eq:sigma-pBCK} \Sigma_\pBCK = \{\, B:0,\, B':1,\, B'':2,\, C:0,\, C':1,\, C'':2,\, K:0,\, K':1,\, I:0,\, \Omega:0,\,\appp:2,\, \oplus:2 \,\}. \end{equation}
We let $\Lambda_\pBCK$ denote the set of closed $\Sigma_\pBCK$-terms. In comparison to $\pSKI$, we have dropped the non-affine combinator $S$ and added the combinators $B$ and $C$, which represent the affine $\lambda$-terms
\[ B=\lambda x.\,\lambda y.\,\lambda z.\, x\, (y\, z), \qquad C = \lambda x.\,\lambda y.\,\lambda z.\, (x\, z)\, y.\]
The operational rules for $B$ and $C$ (and their respective primed versions) are given in \Cref{fig:skirules-untyped}, and the rules for all remaining operations are the same as for $\pSKI$ (\Cref{fig:skirules-untyped1}). As before, the operational model of $\pBCK$ is the least higher-order coalgebra
\begin{equation}\label{eq:op-model-pBCK-untyped} \gamma_\pBCK\colon \Lambda_\pBCK\to B(\Lambda_\pBCK,\Lambda_\pBCK)  \end{equation}
satisfying all operational rules of $\pBCK$. This time, rather than contextual equivalence, we consider the \emph{contextual pseudometric}~\cite{cdl15} for $\pBCK$, which is the pseudometric on $\Lambda_\pBCK$ defined by
\[ d_\ctx^\pBCK(t,s) = \sup_{C[\cdot] \text{ context}} |\gamma(C[t])(\bot)-\gamma(C[s])(\bot)|. \]
Thus $d_\ctx^\pBCK(t,s)$ measures to what extent contexts are able to distinguish $t$ from $s$ by observing termination probabilities. At this stage, the contextual pseudometric is easily defined but hard to use (the reader may try to prove the `obvious' statement $d_\ctx^\pBCK(I,I\oplus \Omega)=\frac{1}{2}$). Again, our categorical congruence results will lead to efficient coinductive proof methods that simplify the reasoning.

 \begin{figure*}[t]
  \begin{gather*}
    \inference{}{B\Downarrow 1\cdot(t\mapsto B'(t))}
    \qquad \inference{}{B'(t)\Downarrow 1\cdot(s\mapsto B''(t,s))} \qquad \inference{}{B''(t,s)\Downarrow 1\cdot (u\mapsto t\, (s\, u))}  
 \\[1ex]
\inference{}{C\Downarrow 1\cdot (t\mapsto C'(t))}
\qquad
  \inference{}{C'(t)\Downarrow 1\cdot (s\mapsto C''(t,s))}   \qquad
  \inference{}{C''(t,s)\Downarrow 1\cdot (u\mapsto (t\, u)\, s)}
  \end{gather*}
  \caption{Operational semantics of \pBCK.}
  \label{fig:skirules-untyped}
\end{figure*}

\section{Abstract Higher-Order Specifications}\label{sec:ahos}
The languages \SKI, \pSKI, and \pBCK exemplify the way how higher-order languages are typically specified: their syntax is given by a signature of operation symbols, and their operational semantics determines some form of transition system (higher-order coalgebra) on the set of program terms whose transitions are specified by a set of operational rules. This process is captured in categorical terms via the notion of \emph{abstract higher-order specification} (\emph{AHOS}), introduced next. AHOS are related to both \emph{biGSOS}~\cite{rot19} and \emph{higher-order abstract GSOS}~\cite{gmstu23} specifications; see \Cref{rem:ho-gsos}.

\begin{definition}[AHOS]
An \emph{abstract higher-order specification} (\emph{AHOS}) $\S=(\B,\Sigma,B,\rho)$ is given by:
\begin{enumerate}
\item a category $\B$;
\item a functor $\Sigma\colon \B\to \B$ with an initial algebra $(\Lambda,\ini)$ and a free algebra $\Sigmas\Lambda$ generated by $\Lambda$.\footnote{The elements of $\Sigmas \Lambda$ are thought of as terms whose variables are terms themselves.}
\item a bifunctor $B\colon \B^\opp\times \B\to \B$ with a cofree $B(\Lambda,-)$-coalgebra generated by $\Lambda$, denoted $B^\infty(\Lambda,\Lambda)$.
\item a morphism $\rho\colon \Sigma B^\infty(\Lambda,\Lambda)\to B(\Lambda,\Sigmas\Lambda)$, the \emph{rule morphism} of $\S$.
\end{enumerate}
A \emph{model} of $\S$ is a higher-order coalgebra $\gamma\colon \Lambda\to B(\Lambda,\Lambda)$ such that the following diagram commutes:
\begin{equation}\label{eq:ahos-op-model}
\begin{tikzcd}
\Sigma \Lambda \ar{r}{\ini} \ar{d}[swap]{\Sigma \hat\gamma} & \Lambda \ar{r}{\gamma} & B(\Lambda,\Lambda)\\
\Sigma B^\infty(\Lambda,\Lambda) \ar{rr}{\rho} && B(\Lambda,\Sigmas\Lambda) \ar{u}[swap]{B(\id,\hat\ini)}
\end{tikzcd}
\end{equation}
Here $\hat\gamma$ and  $\hat\ini$ are the cofree extension and free extension of $\id\colon \Lambda\to\Lambda$, respectively.
\end{definition}

Informally, $\Sigma$ and $B$ specify the syntax and the behaviour type of a higher-order language. The initial $\Sigma$-algebra $\Lambda$ is the object of program terms. The rule morphism $\rho$ describes the set of operational rules of the language and encodes them into a function. For each operation symbol $\f$ in the signature of the language, $\rho$ specifies the behaviour of a program $\f(t_1,\ldots,t_n)$ in terms of the unravelled (infinite-depth) behaviours of its subprograms $t_1,\ldots,t_n$, which are elements of the cofree coalgebra $B^\infty(\Lambda,\Lambda)$. A model of an AHOS is a transition system on the object $\Lambda$ of program terms that `runs' programs according to the operational rules specified by the map $\rho$.

\begin{rem} 
The rule morphism $\rho$ of an AHOS typically extends to a dinatural transformation
\[ \rho_X\colon \Sigma B^\infty(X,X)\to B(X,\Sigmas X) \qquad (X\in \B).\]
Dinaturality means that the operational rules encoded by $\rho$ are parametrically polymorphic, i.e.~they do not inspect the structure of their arguments. For instance, this is the case for the rules of \SKI, \pSKI, and \pBCK. In our congruence results for behavioural conformances (\Cref{sec:congruence-ahos}), parametric polymorphism is not expressed via dinaturality, but in terms of a lifting property of $\rho$.
\end{rem}

AHOS generally allow to specify the behaviour of a program $\f(t_1,\ldots,t_n)$  in terms of the infinite-depth behaviours of its subprograms $t_i$. In applications, it is typically sufficient to restrict to behaviours of some finite depth. For instance, in \pSKI and \pBCK, the operational rule for an application $t\, s$ (\Cref{fig:skirules-untyped1}) requires inspecting the depth-$2$ behaviour of $t$, and the rules for all remaining operations inspect at most the depth-$1$ behaviour of the operands. This is captured by the notion of \emph{depth-$n$ AHOS} defined below. We first introduce some required notation.

\begin{notation}
Given an AHOS $\S=(\B,\Sigma,B,\rho)$ we denote the counit and the $B(\Lambda,-)$-coalgebra structure of the cofree coalgebra $B^\infty(\Lambda,\Lambda)$ by
\[\epsilon\colon B^\infty(\Lambda,\Lambda)\to \Lambda\qqand\tau\colon B^\infty(\Lambda,\Lambda)\to B(\Lambda,B^\infty(\Lambda,\Lambda)).\]
For each $n\geq 0$ we write $B^{(n)}(\Lambda,\Lambda)$ for the $n$-fold application of $B(\Lambda,-)$ to $\Lambda$, that is,
\[ B^{(0)}(\Lambda,\Lambda)=\Lambda\qqand B^{(n+1)}(\Lambda,\Lambda)=B(\Lambda,B^{(n)}(\Lambda,\Lambda)). \]
We define the projections
$c^{(n)}\colon B^\infty(\Lambda,\Lambda)\to B^{(n)}(\Lambda,\Lambda)$
inductively by 
\[c^{(0)}=\epsilon\;\,\text{and}\;\, c^{(n+1)} = (\, B^\infty(\Lambda,\Lambda)\xto{\tau} B(\Lambda,B^\infty(\Lambda,\Lambda)) \xto{B(\Lambda,c^{(n)})} B(\Lambda,B^{(n)}(\Lambda,\Lambda))=B^{(n+1)}(\Lambda,\Lambda) \,).\]
Finally, we define the $n$-fold iteration $\gamma^{(n)}\colon \Lambda\to B^{(n)}(\Lambda,\Lambda)$ of $\gamma\colon \Lambda\to B(\Lambda,\Lambda)$ inductively by
\begin{align*} \gamma^{(0)} &=(\,\Lambda\xto{\id}\Lambda=B^{(0)}(\Lambda,\Lambda)\,) \\
 \gamma^{(n+1)} &=  
(\,\Lambda \xto{\gamma} B(\Lambda,\Lambda) \xto{B(\Lambda,\gamma^{(n)})} B(\Lambda,B^{(n)}(\Lambda,\Lambda)) = B^{(n+1)}(\Lambda,\Lambda)\,).
\end{align*}
\end{notation} 
An easy induction on $n$ shows that:
\begin{lemma}\label{lem:c-gamma}
For every $n\in \Nat$, one has $\gamma^{(n)} = (\, \Lambda \xto{\hat{\gamma}} B^\infty(\Lambda,\Lambda) \xto{c^{(n)}} B^{(n)}(\Lambda,\Lambda) \,)$.
\end{lemma}

\begin{definition}[Depth-$n$ AHOS]
Let $n\in \Nat$. A \emph{depth-$n$ AHOS} is a quadruple $\S=(\B,\Sigma,B,\rho)$ where $\B$, $\Sigma$, $B$ are given as for AHOS, the category $\B$ has finite products, and $\rho$ is a morphism
\begin{equation}\label{eq:abstract-rule-depth-n} \rho\colon \Sigma \textstyle\prod_{k=0}^n B^{(n)}(\Lambda,\Lambda)\to B(\Lambda,\Sigmas\Lambda). \end{equation}
A \emph{model} of $\S$ is a higher-order bialgebra $(\Lambda,\gamma)$ such that the following diagram commutes:
\begin{equation}\label{eq:rho-model-depth-n}
\begin{tikzcd}
\Sigma \Lambda \ar{r}{\ini} \ar{d}[swap]{\Sigma \langle \gamma^{(k)}\rangle_{k=0}^n} & \Lambda \ar{r}{\gamma} & B(\Lambda,\Lambda)\\
\Sigma \prod_{k=0}^n B^{(k)}(\Lambda,\Lambda) \ar{rr}{\rho} && B(\Lambda,\Sigmas\Lambda) \ar{u}[swap]{B(\id,\hat\ini)}
\end{tikzcd}
\end{equation}
\end{definition}

\begin{rem}
 Every depth-$n$ AHOS $\S=(\B,\Sigma,B,\rho)$ induces the AHOS $\S'=(\B,\Sigma,B,\rho')$ given by
\[ \rho' = (\, \begin{tikzcd}[column sep=50] \Sigma B^\infty(\Lambda,\Lambda) \ar{r}{\Sigma\langle c^{(k)}\rangle_{k=0}^n} & \Sigma \textstyle\prod_{k=0}^n B^{(k)}(\Lambda,\Lambda)\ar{r}{\rho} & B(\Lambda,\Sigmas\Lambda) \end{tikzcd}\,). \]
By \Cref{lem:c-gamma}, every model of $\S$ is a model of $\S'$, and vice versa. Therefore, depth-$n$ AHOS can be regarded as a special case of AHOS. Depth-$n$ AHOS are easier to work with as they do not involve the cofree coalgebra $B^\infty(\Lambda,\Lambda)$, which may be hard to compute or even fail to exist.
\end{rem}

\begin{example}[AHOS for \SKI]\label{ex:ahos-SKI} The language \SKI is modelled by the depth-$2$ AHOS  \[\S_\SKI=(\Set,\Sigma_\SKI,B_0,\rho_\SKI)\] where $\Sigma=\Sigma_\SKI$ is the polynomial functor corresponding to the signature \eqref{eq:sigma-SKI} of $\SKI$, and $B_0$ is the behaviour bifunctor \eqref{eq:beh-pBCK-untyped0}. To define the rule map $\rho_\SKI$, we write $\Lambda=\Lambda_\SKI$ for the set of closed $\SKI$-terms, and note that
\[ B^{(0)}(\Lambda,\Lambda) = \Lambda,\qquad B_0^{(1)}(\Lambda,\Lambda)= \{\bot\}+\Lambda^\Lambda,\qquad B_0^{(2)}(\Lambda,\Lambda)= \{\bot\}+(\{\bot\}+\Lambda^\Lambda)^\Lambda.  \]
Thus $\rho_\SKI$ is a map of type
\begin{equation*} \rho_\SKI\colon \Sigma (\Lambda\times (\{\bot\}+\Lambda^\Lambda)\times (\{\bot\}+(\{\bot\}+\Lambda^\Lambda)^\Lambda)) \to \{\bot\}+(\Sigmas\Lambda)^\Lambda \end{equation*}
that simply presents the operational rules of \Cref{fig:skirules-untyped0}, e.g.\ for $t,s\in \Lambda$ and $F\in (\{\bot\}+\Lambda^\Lambda)^\Lambda$, 
\begingroup
\allowdisplaybreaks
\begin{alignat*}{9}
&S && \mapsto && \;\;  (t\mapsto S'(t)) 
&& \Omega && \mapsto && \;\;  \bot \;\; \\
& S'((t,-,-)) &&\mapsto && \;\;  (s\mapsto S''(t,s))
&&\appp((-,-,F),(s,-,-))\;\; &&\mapsto && \;\; F(s) \\
& S''((t,-,-),(s,-,-)) && \mapsto && \;\;  (u\mapsto (t\, u)\, (s\, u)) \qquad
&& \appp((-,-,\bot),(s,-,-))\;\; &&\mapsto && \;\; \bot 
\end{alignat*}
\endgroup
\end{example}

\begin{example}[AHOS for \pSKI]\label{ex:ahos-pSKI} The language \pSKI is modelled by the depth-$2$ AHOS  \[\S_\pSKI=(\Set,\Sigma_\pSKI,B,\rho_\pSKI)\] with $\Sigma=\Sigma_\pSKI$ given by \eqref{eq:sigma-pSKI} and $B$ given by \eqref{eq:beh-pBCK-untyped}. Putting $\Lambda=\Lambda_\pSKI$, we have
\[ B^{(0)}(\Lambda,\Lambda) = \Lambda,\qquad B^{(1)}(\Lambda,\Lambda)= \D(\{\bot\}+\Lambda^\Lambda),\qquad B^{(2)}(\Lambda,\Lambda)= \D(\{\bot\}+(\D(\{\bot\}+\Lambda^\Lambda))^\Lambda).  \]
Thus $\rho_\pSKI$ is a map of type
\begin{equation*} \rho_\pSKI\colon \Sigma (\Lambda\times \D(\{\bot\}+\Lambda^\Lambda)\times \D(\{\bot\}+(\D(\{\bot\}+\Lambda^\Lambda))^\Lambda)) \to \D(\{\bot\}+(\Sigmas\Lambda)^\Lambda) \end{equation*}
that presents the operational rules of \Cref{fig:skirules-untyped1}:
\begingroup
\allowdisplaybreaks
\begin{alignat*}{9}
&S && \mapsto && \;\; 1\cdot (t\mapsto S'(t)) 
&& K && \mapsto && \;\; 1\cdot (t\mapsto K'(t)) \\
& S'((t,-,-)) &&\mapsto && \;\; 1\cdot (s\mapsto S''(t,s))
&& K'((t,-,-)) && \mapsto && \;\; 1\cdot (s\mapsto t) \\
& S''((t,-,-),(s,-,-)) && \mapsto && \;\; 1\cdot (u\mapsto (t\, u)\, (s\, u))  
&& I && \mapsto && \;\; 1\cdot (t\mapsto t)  \\
&\Omega && \mapsto && \;\; 1\cdot \bot 
&&(-,\phi,-) \oplus (-,\psi,-) && \mapsto && \;\; \frac{1}{2}\cdot \phi + \frac{1}{2}\cdot \psi \\
&\appp((-,-,\Phi),(s,-,-))\;\; &&\mapsto && \;\; p\cdot \bot + \sum_{i\in I}\sum_{j\in J_i} p_i\cdot p_{i,j}\cdot b_{i,j},
\end{alignat*}
\endgroup
where $t,s\in \Lambda$ and $\phi,\psi \in \D(\{\bot\} + \Lambda^\Lambda)$, and in the last clause, 
\[
\Phi = p\cdot \bot + \sum_{i\in I} p_i\cdot F_i \in \D(\{\bot\}+(\D(\{\bot\}+\Lambda^\Lambda))^{\Lambda})\qand 
F_i(s)=\sum_{j\in J_i} p_{i,j}\cdot b_{i,j}\in \D(\{\bot\}+\Lambda^\Lambda).
\]
\end{example}

\begin{remark}
 In the clause for $S$, the expression $S'(t)$ appearing in the map $(t\mapsto S'(t))\in (\Sigmas\Lambda)^\Lambda$ is to be read as a $\Sigma$-term in the variable $t\in \Lambda$, \emph{not} as its interpretation in the initial $\Sigma$-algebra $\Lambda$, i.e.\ an element of $\Lambda\hookrightarrow \Sigmas\Lambda$. Analogously for the terms $S''(t,s)$, $(t\, u)\, (s\, u)$,  $K'(t)$ in the other clauses. 
\end{remark}

\begin{rem}\label{rem:app-mu}
The action of $\rho_\pSKI$ on $\appp$-terms can be expressed in terms of the strength $\st$, the unit $\eta$ and the multiplication $\mu$ of the distribution monad $\D$ (\Cref{sec:category-theory}). Indeed, it is given by
\[ \rho_\pSKI(\appp((-,-,\Phi),(s,-,-)) = j(\Phi,s)\]
where $j\colon B(\Lambda,B(\Lambda,\Lambda))\times \Lambda \to B(\Lambda,\Lambda)$ is the composite shown below:
\allowdisplaybreaks
\[
\begin{tikzcd}
 \D(\{\bot\}+(\D(\{\bot\}+\Lambda^\Lambda))^{\Lambda})\times \Lambda \ar{r}{\st} & 
 \D((\{\bot\}+(\D(\{\bot\}+\Lambda^\Lambda))^{\Lambda})\times \Lambda) \ar{dl}[swap]{\cong} \\
 \D(\{\bot\}\times \Lambda+(\D(\{\bot\}\times \Lambda+\Lambda^\Lambda))^{\Lambda}\times \Lambda) \ar{r}{\D(\outl+\ev)} &
 \D(\{\bot\}+ \D(\{\bot\} + \Lambda^\Lambda)) \ar{dl}[description]{\D(\eta\circ\inl + \id)} \\
 \D(\D(\{\bot\} + \Lambda^\Lambda)+\D(\{\bot\} + \Lambda^\Lambda)) \ar{r}{\D\nabla} &
 \D\D(\{\bot\} + \Lambda^\Lambda) \xto{\mu}
 \D(\{\bot\} + \Lambda^\Lambda)
\end{tikzcd}
\]
\end{rem}

\begin{example}[AHOS for \pBCK]\label{ex:ahos-pBCK} The language \pBCK is modelled by the depth-$2$ AHOS  \[\S_\pBCK=(\Set,\Sigma_\pBCK,B,\rho_\pBCK)\] where $\Sigma=\Sigma_\pBCK$ is the polynomial functor corresponding to the signature \eqref{eq:sigma-pBCK} of $\pBCK$, and $B$ is again the behaviour bifunctor \eqref{eq:beh-pBCK-untyped}. 
Putting $\Lambda=\Lambda_\pBCK$, the rule map
\begin{equation*} \rho_\pBCK\colon \Sigma (\Lambda\times \D(\{\bot\}+\Lambda^\Lambda)\times \D(\{\bot\}+(\D(\{\bot\}+\Lambda^\Lambda))^\Lambda)) \to \D(\{\bot\}+(\Sigmas\Lambda)^\Lambda) \end{equation*}
is given like in \Cref{ex:ahos-pSKI} for $\Omega$, $K$, $K'$, $I$, $\oplus$, $\mathsf{app}$, and for $B$, $B'$, $B''$, $C$, $C'$, $C''$ by
\begingroup
\allowdisplaybreaks
\begin{alignat*}{9}
&B && \mapsto && \;\; 1\cdot (t\mapsto B'(t)) 
&& C && \mapsto && \;\; 1\cdot (t\mapsto C'(t)) \\
& B'((t,-,-)) &&\mapsto&& \;\; 1\cdot (s\mapsto B''(t,s)) \;\;\;\;\;\;\;\;\;\;\;
&&
C'((t,-,-)) &&\mapsto && \;\; 1\cdot (s\mapsto C''(t,s))
 \\
& B''((t,-,-),(s,-,-))\;\; && \mapsto && \;\; 1\cdot (u\mapsto t\,(s\, u)) 
&&C''((t,-,-),(s,-,-)) && \mapsto && \;\; 1\cdot (u\mapsto (t\, u)\, s)  
\end{alignat*}
\endgroup
\end{example}

It remains to explain how to construct a model $\gamma\colon \Lambda\to B(\Lambda,\Lambda)$ of an AHOS. Similar to the construction of the operational models of the languages \SKI, \pSKI, and \pBCK in \Cref{sec:prob-comb-logic}, the idea is that the \emph{least} model is the one to be aimed for, since it carries precisely the information given by the operational rules of the AHOS, and nothing more.
This idea can be formalized by imposing a suitable order structure on the set $\HOCoalg_\Lambda(B)$ of higher-order coalgebras on program terms.

\begin{definition}[Complete AHOS, Canonical Model]\label{def:complete-ahos-can-model}
An \emph{ordered AHOS} $\S=(\B,\Sigma,B,\leq,\rho)$ is an AHOS $(\B,\Sigma,B,\rho)$ together with a partial order $\leq$ on the set $\HOCoalg_\Lambda(B)$. The AHOS $\S$ is \emph{complete} if the order $\leq$ is a DCPO with a least element $\bot$, and moreover the map
\[ \Gamma\colon\HOCoalg_\Lambda(B)\to \HOCoalg_\Lambda(B),\qquad \gamma\mapsto B(\id,\hat\ini)\circ \rho\circ \Sigma\hat\gamma\circ \ini^{-1},  \]
is monotone. The \emph{canonical model} of a complete AHOS $\S$ is the $\leq$-least model of $\S$.
\end{definition}

\begin{rem}\label{rem:complete-ahos}
\begin{enumerate}
\item Since models of $\S$ are precisely fixed points of the map $\Gamma$, the canonical model exists by the Knaster-Tarski fixed point theorem.
\item The canonical model can be constructed iteratively (which corresponds to the iterative proof of the Knaster-Tarski theorem). Consider the ordinal-indexed family of higher-order coalgebras $(\gamma_\alpha\colon \Lambda\to B(\Lambda,\Lambda))_{\alpha}$ defined via transfinite recursion as follows:
\[\gamma_0=\bot,\qquad \gamma_{\alpha+1} = B(\id,\hat\ini) \circ \rho  \circ \Sigma\hat\gamma_\alpha\circ \ini^{-1},\qqand \gamma_\alpha =\bigvee_{\beta<\alpha} \gamma_\beta \;\;\;(\text{$\alpha$ limit ordinal}).
\]
These coalgebras form an ascending chain ($\alpha\leq \beta$ implies $\gamma_\alpha\leq \gamma_\beta$). Since the hom-set $\B(\Lambda,B(\Lambda,\Lambda))$ is small, the chain eventually stabilizes, that is, there exists an ordinal $\alpha_0$ such that $\gamma_\alpha=\gamma_{\alpha_0}$ for all $\alpha\geq \alpha_0$. Then $\gamma=\gamma_{\alpha_0}$ is the least fixed point of $\Gamma$, i.e.\ the canonical model of $\S$.

Intuitively, this iteration computes the canonical model $\gamma$ step by step from below: One starts with a trivial coalgebra $\gamma_0\colon \Lambda\to B(\Lambda,\Lambda)$ that carries no information, and then constructs a sequence $\gamma_1,\gamma_2,\gamma_3,\cdots$ of increasingly more detailed approximations of $\gamma$, until a fixed point is reached.
\end{enumerate}
\end{rem}

\begin{example}[\SKI, \pSKI, \pBCK]
The AHOS $\S_\SKI$ for \SKI (\Cref{ex:ahos-SKI}) is complete w.r.t.\ to the pointwise DCPO structure on $\HOCoalg_{\Lambda_\SKI}(B_0)$. The map $\Gamma$ of \Cref{def:complete-ahos-can-model} is precisely the map \eqref{eq:Gamma-SKI}; thus, the canonical model of $\S_\pSKI$ is the higher-order coalgebra $\gamma_\pSKI$ defined in \eqref{eq:op-model-pSKI-untyped0}. Analogously, the canonical models for $\S_\pSKI$ and $\S_\pBCK$ are  the coalgebras \eqref{eq:op-model-pSKI-untyped} and \eqref{eq:op-model-pBCK-untyped}.
\end{example}
We conclude this section with a discussion of existing categorical frameworks related to AHOS.
\begin{rem}\label{rem:ho-gsos}
AHOS form a common generalization of two operational specification frameworks known in the literature, namely \emph{biGSOS}~\cite{rot19} and \emph{higher-order abstract GSOS}~\cite{gmstu23}. In more detail:
\begin{enumerate}
\item  AHOS are a higher-order version of \emph{monotone biGSOS specifications}~\cite{rot19}. The latter are  natural transformations of type $\Sigma B^\infty \to B\Sigmas$ for an endofunctor (rather than mixed variance bifunctor) $B$. Our construction of the operational model of a complete AHOS as a least fixed point extends a corresponding construction from \emph{op.~cit.} to the higher-order setting.
\item AHOS also relate to the recently introduced \emph{higher-order abstract GSOS} framework~\cite{gmstu23}. The core ingredient of the latter is the notion of \emph{higher-order GSOS law}, which is essentially the same as a depth-$1$ AHOS (subject to additional dinaturality conditions). To construct the operational model of a higher-order GSOS law, no order on $B$ is necessary: there is a \emph{unique} higher-order coalgebra~$\gamma$  making \eqref{eq:rho-model-depth-n} with $n=1$ commute, and it can be defined via structural induction (using initiality of the algebra $\Lambda=\mu\Sigma$) rather than as an order-theoretic least fixed point. This is no longer true for depth-$n$ AHOS where $n\geq 2$. The present notion of AHOS can thus be regarded as a substantially more general and flexible version of higher-order abstract GSOS. In particular, due to restriction to operational models definable by structural induction, the latter framework is inherently tailored to small-step operational semantics. Operational rules corresponding to a big-step semantics can only be captured indirectly in higher-order abstract GSOS, namely by generating them from a small-step specification. This requires syntactic restrictions on the format, such as an explicit syntactic distinction between values and computations, and a complex technical condition on higher-order GSOS laws called \emph{strong separatedness}~\cite{gpt25}. In contrast, our present AHOS framework captures big-step operational semantics directly, and with a remarkably simple categorical setup. The operational model of an AHOS is, in contrast to the approach of higher-order abstract GSOS,  constructed via a form of rule induction (corresponding to the order-theoretic construction of a least fixed point), not via the more restrictive structural induction. It is the use of rule induction that allows modelling rules that require unravelling more than one layer of behaviours (like the rule for application in \Cref{fig:skirules-untyped1}), which is typical for big-step specifications.
\end{enumerate}
\end{rem}

With the abstract categorical account of higher-order languages and their operational semantics provided by AHOS at hand, our primary goal is to reason about notions of \emph{behavioural conformance} (e.g.~behavioural equivalence or behavioural distance) on canonical models of AHOS. This is achieved by putting suitable fibrations over the base category of an AHOS.

\section{Fibrations and Behavioural Conformances}\label{sec:fibrations}
Behavioural conformances for transition systems, and coalgebras in general, can be conveniently modelled by considering liftings of behaviour functors along fibrations~\cite{hj98,bbkk18}. The level of abstraction achieved in this way gives rise to a uniform  framework for reasoning about similarity of programs, bisimilarity of programs, distances of programs, and related notions. In this section, we set up the fibrational foundations for our theory of behavioural conformances on canonical models of AHOS. In particular, we develop a fibrational version of \emph{Howe's closure}, the key construction for proving abstract congruence results for higher-order languages.

\subsection{Complete Lattice Fibrations}
We first recall some terminology and basic results from the theory of fibrations; see~\citet{jacobs99} for a comprehensive introduction to the subject. In the context of the coalgebraic theory of behavioural conformances, a fibration $p\colon \E\to \B$ is thought of as the forgetful functor of a category $\E$ of $\B$-objects carrying additional (typically relational) structure, together with a notion of pulling back that structure along maps in the base category $\B$.

\begin{defn}[Fibration]\label{def:fibration}
Let $p\colon \E\to\B$ be a functor.
\begin{enumerate}
\item An object $P\in \E$ is \emph{above} the object $X\in \B$ if $pP=X$; similarly, a morphism $u\in \E$ is \emph{above} the morphism $f\in \B$ if $pu=f$. The \emph{fiber} at $X\in \B$ is the non-full subcategory $\E_X\hookrightarrow \E$ given by all objects above $X$ and all morphisms above $\id_X$.
\end{enumerate}
\vspace{-.2cm}
\begin{minipage}[t]{0.75\textwidth}
\begin{enumerate}
 \setcounter{enumi}{1} \item Given $X\in \B$ and $P\in \E$, a \emph{cartesian lift} of a morphism $f\colon X\to pP$ is a morphism $\barf\colon f^{*}P \to P$ above $f$ such that for every morphism $g\colon Q\to P$ in~$\E$ and $h\colon pQ\to X$ such that $pg=f\circ h$, there exists a unique $\barh\colon Q\to f^{*}P$ above $h$ such that $g=\barf\circ \barh$. The functor $p$ is a \emph{fibration} if every $f\colon X\to pP$ has a cartesian lift. Every $f\colon X\to Y$ in $\B$ induces the \emph{reindexing} functor
\[f^{*}\colon \E_Y\to \E_X, \qquad P\mapsto f^{*} P.\] 
\end{enumerate}
\end{minipage}
\begin{minipage}[t]{0.23\textwidth}
\vspace{-.4cm}
\[\begin{tikzcd}[column sep=10]
      Q \ar[dashed]{r}{\barh} \ar[shiftarr={yshift=15}]{rr}{g} & f^{*}P \ar{r}{\barf} & P \\
      pQ \ar{r}{h} \ar[shiftarr={yshift=15}]{rr}{pg} & X \ar{r}{f} & pP 
    \end{tikzcd}
\]
\end{minipage} 
\\[.2cm]
\begin{minipage}[t]{0.75\textwidth}
\begin{enumerate}
 \setcounter{enumi}{2} \item  Dually, an \emph{opcartesian lift} of $f\colon pP\to Y$ is a morphism $\barf\colon P \to f_{*}P$ above $f$ such that for every morphism $g\colon P\to Q$ in $\E$ and $h\colon Y\to pQ$ such that $pg=h\circ f$, there exists a unique $\barh\colon f_{*}P\to Q$ above $h$ such that $g=\barh\circ \barf$. The functor $p$ is an \emph{opfibration} if every $f\colon pP\to Y$ has an opcartesian lift. Every $f\colon X\to Y$ in $\B$ induces the \emph{opreindexing} functor
 \[f_{*}\colon \E_X\to \E_Y, \qquad P\mapsto f_{*} P.\]
\end{enumerate}
\end{minipage}
\begin{minipage}[t]{0.23\textwidth}
\vspace{-.4cm}
\[\begin{tikzcd}[column sep=10]
      P \ar{r}{\barf} \ar[shiftarr={yshift=15}]{rr}{g} & f_{*}P \ar[dashed]{r}{\barh} & Q \\
      pP \ar{r}{f} \ar[shiftarr={yshift=15}]{rr}{pg} & Y \ar{r}{h} & pQ 
    \end{tikzcd}
\]
\end{minipage} 
\vspace{.2cm}
\begin{enumerate}
\setcounter{enumi}{3}
\item The functor $p$ is a \emph{bifibration} if it is both a fibration and an opfibration.
\item A fibration $p$ is a \emph{$\CLat_\sqcap$-fibration} if each fiber is a complete lattice, i.e.\ a small poset (viewed as a thin category) where each subset has a join (= least upper bound) and a meet (= greatest lower bound), and moreover the reindexing map $f^{*}\colon \E_Y\to \E_X$ is meet-preserving for every $f\colon X\to Y$. We denote the partial order on $\E_X$ by $\sqleq_X$, or simply $\sqleq$, and joins and meets in $\E_X$ by $\bigsqcup$ and $\bigsqcap$. 
\end{enumerate}
\end{defn}
The next lemma collects some standard facts about $\CLat_\sqcap$-fibrations:

\begin{lemma}\label{lem:clat-fib-props}
For every $\CLat_\sqcap$-fibration $p\colon \E\to \B$, the following statements hold: 
\begin{enumerate}
\item\label{lem:clat-fib-props-faithful} $p$ is faithful, that is, for all $u,v\colon P\to Q$ in $\E$, if $pu=pv$ then $u=v$.
\item\label{lem:clat-fib-props-bifib} $p$ is a bifibration. 
\item\label{lem:clat-fib-props-adjunction} For each $f\colon X\to Y$ in $\B$ there is a Galois connection $f_{*} \dashv f^{*}\colon \E_Y\to \E_X$, that is,
\[ f_{*} P \sqleq Q \iff P\sqleq f^{*} Q\qquad \text{for all $P\in \E_X$ and $Q\in \E_Y$}. \]
In particular, the map $f_{*}$ preserves all joins.
\item\label{lem:clat-fib-probs-split} One has $(g\circ f)_{*}=g_{*}\circ f_{*}$ and $(g\circ f)^{*} = f^{*}\circ g^{*}$ for all $f\colon X\to Y$ and $g\colon Y\to Z$ in $\B$.
\end{enumerate}
\end{lemma}

\begin{rem}
$\CLat_\sqcap$-fibrations are essentially equivalent to another notion from category theory, namely \emph{topological functors}~\cite{ahs90}; in fact, a $\CLat_\sqcap$-fibration is precisely a topological functor with small fibers. This observation is folklore and has been noted by several authors without a proof~\cite{kkhkh19,fsw24}. We give a proof in \ifarx{the appendix (\Cref{app:fibrations-vs-topological})}{\cite[App.~A]{u25_arxiv}} for the convenience of the reader.
\end{rem}

\begin{example}[$\CLat_\sqcap$-Fibrations]\label{ex:fibrations}
\begin{enumerate}
\item \label{ex:fibrations-rel} To capture (bi)simulations in the fibrational setting, we consider a fibration of relations. A \emph{relation} is a pair $(X,R)$ of a set $X$ and a subset $R\seq X\times X$. Given two relations $(X,R)$ and $(Y,S)$, a pair $(f,g)$ of maps $f,g\colon X\to Y$ is \emph{jointly relation-preserving} if $R(x,x')$ implies $S(f(x),g(x'))$ for all $x,x'\in X$. We let $\Rel_\jrp$ denote the category of relations and jointly relation-preserving pairs of maps. The forgetful functor $p_\jrp\colon \Rel_\jrp \to \Set^{(2)}$ (cf.\ \Cref{sec:category-theory}, \emph{Notation}) given by $(X,d)\mapsto X$ and $(f,g)\mapsto (f,g)$ is a $\CLat_\sqcap$-fibration.
We denote the fiber at $X\in \Set^{(2)}$ by $\Rel_{\jrp,X}$. The (op)reindexing maps for $(f,g)\colon X\to Y$ in $\Set^{(2)}$ are given by
\begin{align*}
(f,g)^{*}\colon  \Rel_{\jrp,Y}\to \Rel_{\jrp,X},&\quad (Y,S) \mapsto (X,(f,g)^{*}S),\\
(f,g)_{*}\colon \Rel_{\jrp,X}\to \Rel_{\jrp,Y},&\quad (X,R)\mapsto (Y,(f,g)_{*}R),
\end{align*}
where $(f,g)^{*}S$ and $(f,g)_{*}R$ are formed by taking the (pre)image under the map $f\times g\colon X\times X\to Y\times Y$:
\[ (f,g)^{*}S = (f\times g)^{-1}[S] \qqand  (f,g)_{*}R = (f\times g)[R]. \]
The order $\sqleq$ on a fiber $\Rel_{\jrp,X}$ is given by inclusion of relations: $(X,R)\sqleq (X,S)$ iff $R\seq S$. 
\item  Given two relations $(X,R)$ and $(Y,S)$, a map $f\colon X\to Y$ is \emph{relation-preserving} if $R(x,x')$ implies $S(f(x),f(x'))$ for all $x,x'\in X$, that is, if $(f,f)$ is jointly relation-preserving. We let $\Rel_\rp$ denote the category of relations and relation-preserving maps. The forgetful functor $p_\rp\colon \Rel_\rp\to \Set$ given by $(X,R)\mapsto X$ and $f\mapsto f$ is a $\CLat_\sqcap$-fibration. The (op)reindexing maps for $f\colon X\to Y$ are given in terms of those for $p_\jrp$ by $f^{*}=(f,f)^{*}$ and $f_{*}=(f,f)_{*}$.  
\item\label{ex:fibrations-frel2} Quantitative notions of behavioural distance are captured by moving from relations to fuzzy relations. A \emph{fuzzy relation} is a pair $(X,d)$ of a set $X$ and a map $d\colon X\times X\to [0,1]$. Given two fuzzy relations $(X,d_X)$ and $(Y,d_Y)$, a pair $(f,g)$ of maps $f,g\colon X\to Y$ is \emph{jointly non-expansive} if  $d_Y(f(x),g(x'))\leq d_X(x,x')$ for all $x,x'\in X$. We let $\FRel_\jne$ denote the category of fuzzy relations and jointly non-expansive pairs of maps. The forgetful functor $p_\jne\colon \FRel_\jne \to \Set^{(2)}$ given by $(X,d)\mapsto X$ and $(f,g)\mapsto (f,g)$ is a $\CLat_\sqcap$-fibration. The (op)reindexing maps for $(f,g)\colon X\to Y$ in $\Set^{(2)}$ are given by
\begin{align*}
(f,g)^{*}\colon  \FRel_{\jne,Y}\to \FRel_{\jne,X},&\quad (Y,d) \mapsto (X,(f,g)^{*}d),\\
(f,g)_{*}\colon \FRel_{\jne,X}\to \FRel_{\jne,Y},&\quad (X,d)\mapsto (Y,(f,g)_{*}d),
\end{align*}
with the fuzzy relations $(f,g)^{*}d$ and $(f,g)_{*}d$ defined as follows:
\[ (f,g)^{*}d(x,x')=d(f(x),g(x')) \qand  (f,g)_{*}d(y,y')=\inf \{\, d(x,x') \mid f(x)=y,\, g(x')=y'\,\}. \]
The order $\sqleq$ on a fiber $\FRel_{\jne,X}$ is given by the \emph{reversed} pointwise order of fuzzy relations, that is, $(X,d)\sqleq (X,d')$ iff $d(x,x)\geq d'(x,x')$ for all $x,x'\in X$. 
\item  Given fuzzy relations $(X,d_X)$ and $(Y,d_Y)$, a map $f\colon X\to Y$ is \emph{non-expansive} if $d_Y(f(x),f(x'))\leq d_X(x,x')$ for all $x,x'\in X$, that is, if the pair $(f,f)$ is jointly non-expansive.  The forgetful functor $p_\ne\colon \FRel_\ne\to \Set$ is a $\CLat_\sqcap$-fibration. The (op)reindexing maps for $f\colon X\to Y$ are given in terms of those for $p_\jne$ by $f^{*}=(f,f)^{*}$ and $f_{*}=(f,f)_{*}$.  
\end{enumerate}
\end{example}

\begin{rem}\label{rem:rel-frel}
\begin{enumerate}
\item\label{rem:rel-frel-embedding} A relation $R\seq X\times X$ is identified with the fuzzy relation $d\colon X\times X\to [0,1]$ given by $d(x,x')=0$ if $R(x,x')$ and $d(x,x')=1$ otherwise. Under this identification, jointly relation-preserving pairs of maps correspond precisely to jointly non-expansive pairs of maps, so $\Rel_\jrp$ forms a full subcategory of $\FRel_\jne$. Similarly, $\Rel_\rp$ forms a full subcategory of $\FRel_\ne$.
\item\label{rem:rel-frel-cob} In fibrational terminology, $p_\rp$ and $p_\ne$ arise from $p_\jrp$ and $p_\jne$ via \emph{change of base}~\cite[Lem.~1.5.1]{jacobs99}, that is, the diagrams below are pullbacks in the (superlarge) category of categories and functors. Here $I$ is the inclusion functor (\Cref{sec:category-theory}, \emph{Notation}).
\[
\begin{tikzcd}
\Rel_\rp \ar[tail]{r} \ar{d}[swap]{p_\rp} \pullbackangle{-45} \ar{d} & \Rel_\jrp \ar{d} \ar{d}{p_\jrp} \\
\Set \ar[tail]{r}{I} & \Set^{(2)}
\end{tikzcd}
\qquad
\begin{tikzcd}
\FRel_\ne \ar[tail]{r} \ar{d}[swap]{p_\ne} \pullbackangle{-45} \ar{d} & \FRel_\jne \ar{d} \ar{d}{p_\jne} \\
\Set \ar[tail]{r}{I} & \Set^{(2)}
\end{tikzcd}
\]
\item\label{rem:rel-frel-jointly} In most accounts of behavioural conformances for coalgebras, the authors work with fibrations like $p_\rp$ and $p_\ne$. The slight extension to \emph{jointly} relation-preserving or non-expansive pairs of maps is required for our fibrational development of Howe's method.
\end{enumerate}
\end{rem}

\subsection{Involutive Quantale Fibrations}
Apart from being a complete lattice, the fiber $\Rel_{\jrp,X}$ ($=\Rel_{\rp,X}$) naturally carries an algebraic structure whose operations are given by composition and reversal of relations:
\[ R\cdot S = \{\, (x,y) \mid \exists z\in Z.\, R(x,z)\wedge S(z,y) \, \}\qqand R^\circ = \{ \, (y,x) \mid R(x,y) \, \}. \]  
Similarly for the fiber $\FRel_{\jne,X}$ ($=\FRel_{\ne,X}$), where the fuzzy generalizations of composition and reversal are given as follows (we write $a\boxplus b = \min\{a+b,1\}$ for truncated addition in $[0,1]$):
\begin{align*}
d\cdot e\colon X\times X\to [0,1],&&(d\cdot e)(x,y) & = \inf_{z\in Z} \{\, d(x,z) \boxplus e(z,y)\,\}, \\
d^\circ\colon X\times X\to [0,1],&& d^\circ(x,y) &= d(y,x).
\end{align*}
In both cases, composition and reversal are compatible with the complete lattice structure of the fibers and with reindexing. This is captured abstractly by the following definitions:
\begin{definition}\label{def:quantale}
 A \emph{monoid} $(M,\cdot,1)$ is a set $M$ equipped with a\\
\begin{minipage}[t]{0.66\textwidth}
 binary operation $\cdot\colon M\times M\to M$ and a constant $1\in M$ satisfying \eqref{eq:mon-1} and \eqref{eq:mon-2}. A \emph{involutive monoid} (or \emph{$\circ$-monoid}) is a monoid~$M$ equipped with an additional unary operation $(-)^\circ\colon M\to M$ satisfying \eqref{eq:inv-mon1}--\eqref{eq:inv-mon3}. A \emph{quantale} is a monoid $M$ equipped with complete lattice structure satisfying \eqref{eq:quant-1} and \eqref{eq:quant-2} for all $x\in M$ and $S\seq M$. An \emph{involutive quantale} is a monoid that is both involutive and a quantale, and additionally satisfies \eqref{eq:inv-quant}. A \emph{morphism} $h\colon M\to N$ of ($\circ$-)monoids is a map satisfying, for all $x,y\in M$,
\[ 1=h(1)\qand h(x)\cdot h(y)=h(x\cdot y)\quad (\text{and}\quad  (h(x))^\circ = h(x^\circ)). \]
If $N$ is equipped with a partial order $\sqleq$ (e.g.\ if $N$ is a quantale), a \emph{lax ($\circ$-)morphism} is a map $h\colon M\to N$ satisfying
\[ 1\sqleq h(1)\qand h(x)\cdot h(y)\sqleq h(x\cdot y)\quad (\text{and}\quad  (h(x))^\circ \sqleq h(x^\circ)). \]
\vspace{-.2cm}
\end{minipage}
\begin{minipage}[t]{0.33\textwidth}
\vspace{-1cm}
\begin{align} 
(x\cdot y)\cdot z &= x\cdot (y\cdot z) \label{eq:mon-1} \\
1\cdot x = x &= x\cdot 1 \label{eq:mon-2} \\
(x^\circ)^\circ &= x \label{eq:inv-mon1}\\
1^\circ&=1, \label{eq:inv-mon2}\\
(x\cdot y)^\circ &= y^\circ\cdot x^\circ \label{eq:inv-mon3}\\
x\cdot \bigsqcup S &= \bigsqcup_{s\in S} x\cdot s \label{eq:quant-1}\\
(\bigsqcup S)\cdot x &= \bigsqcup_{s\in S} s\cdot x \label{eq:quant-2}\\
(\bigsqcup S)^\circ &= \bigsqcup_{s\in S} s^\circ \label{eq:inv-quant}
  \end{align}
\end{minipage}
\end{definition}

\begin{definition}\label{def:qtl-fib}
A $\CLat_\sqcap$-fibration $p\colon \E\to \B$ is a \emph{$\Qtl$-fibration} if \\
\begin{minipage}[t]{0.66\textwidth}
each fiber $\E_X$ ($X\in \B$) carries a quantale structure $(\E_X,\cdot,1_X)$ satisfying \eqref{eq:qtl-fibration1} and \eqref{eq:qtl-fibration2} for each $f\colon X\to Y$ and $P,Q\in \E_Y$. It is an \emph{$\InQtl$}-fibration if each fiber carries an involutive quantale structure satisfying \eqref{eq:qtl-fibration1}--\eqref{eq:iqtl-fibration}.
\end{minipage}
\begin{minipage}[t]{0.33\textwidth}
\vspace{-.9cm}
\begin{align} 1_X &\sqleq f^{*}1_Y \label{eq:qtl-fibration1}\\
f^{*} P \cdot f^{*} Q &\sqleq f^{*} (P\cdot Q) \label{eq:qtl-fibration2} \\
(f^{*} P)^\circ & \sqleq f^{*}(P^\circ) \label{eq:iqtl-fibration}
  \end{align}
\end{minipage}
\end{definition}

\begin{definition}\label{def:heterogeneous-iqtl-fib}
A $\CLat_\sqcap$-fibration $p\colon \E\to \B^{(2)}$ is a \emph{heterogeneous} $(\mathbf{In})\Qtl$-fibration if each fiber carries an (involutive) quantale structure $(\E_X,\cdot,1_X)$ such that, for $f,g,h\colon X\to Y$ and $P,Q\in \E_Y$,
\begin{equation}\label{eq:het-qtl-fibration-axioms} 1_X\sqleq (f,f)^{*}1_Y \qand (f,g)^{*} P \cdot (g,h)^{*} Q \sqleq (f,h)^{*} (P\cdot Q) \quad(\text{and}\quad ((f,f)^{*}P)^{\circ} = (f,f)^{*}(P^\circ)).  \end{equation}
\end{definition}
Heterogeneous $(\mathbf{In})\Qtl$-fibrations provide exactly enough structure to make Howe's method work. Hence, this is the level of generality at which our congruence results for AHOS will emerge.

\begin{example}\label{ex:iqtl-fibration}
\begin{enumerate}
\item  The fibration $p_\jrp\colon \Rel_\jrp\to \Set^{(2)}$ of \Cref{ex:fibrations}\ref{ex:fibrations-rel} is a heterogeneous $\InQtl$-fibration. For each set $X$ the involutive quantale structure $(\Rel_{\jrp,X},\cdot,1_X,(-)^\circ)$ is given by composition of relations, the identity relation $1_X\seq X\times X$, and reversal of relations.
\item The fibration $p_\jne\colon \Rel_\jne\to \Set^{(2)}$ of \Cref{ex:fibrations}\ref{ex:fibrations-frel2} is a heterogeneous $\InQtl$-fibration. For each set $X$ the involutive quantale structure $(\FRel_{\jne,X},\cdot,1_X,(-)^\circ)$ is given by composition of fuzzy relations, the identity relation $1_X\seq X\times X$ viewed as a fuzzy relation, and the reversal operator. 
\end{enumerate} 
\end{example}

\begin{definition}
An element $x$ of a quantale is \emph{reflexive} if $1\sqleq x$, and \emph{transitive} if $x\cdot x\sqleq x$. If the quantale is involutive, $x$ is \emph{symmetric} if $x^\circ = x$ (equivalently $x^\circ\sqleq x$). The \emph{transitive closure} of $x$ is given by $x^{+} = \bigsqcup_{n\geq 1} x^n$, where $x^n= x\cdot \cdots \cdot x$ denotes the $n$-fold product.
\end{definition}

Note that a reflexive, symmetric, and transitive element of $\Rel_{\jrp,X}$ and $\FRel_{\jne,X}$ is, respectively, an equivalence relation and a {pseudometric} on $X$.

\subsection{Lifting Functors Along Fibrations}
The key parameter to the fibrational notion of behavioural conformance for coalgebras is the choice of a lifting of the behaviour functor along a fibration. Dually, different notions of \emph{congruence} for algebras are modelled in terms of liftings of the syntax functor. 

\begin{defn}[Lifting]\label{def:lifting}
A \emph{lifting} of an endofunctor $F\colon \B\to \B$ (bifunctor $F\colon \B\times \B\to \B$, mixed variance bifunctor $F\colon \B^\opp\times \B\to \B$) along the fibration $p\colon \E\to \B$ is an endofunctor (bifunctor, mixed variance bifunctor) $\barF$ on $\E$ making the respective diagram below commute:
\begin{equation*}
	\label{eq:liftingRel}
	\begin{tikzcd}
		\E \ar{d}[swap]{p}  \ar{r}{\barF} & \E \ar{d}{p}  \\
		\B \ar{r}{F}  & \B
	\end{tikzcd}
\qquad\qquad
\begin{tikzcd}
		\E\times \E
		\ar{d}[swap]{p\times p} \ar{r}{\barF} & \E \ar{d}{p}  \\
		\B \times \B \ar{r}{F}  & \B
	\end{tikzcd}
\quad\qquad
\begin{tikzcd}
		\E^\opp\times \E
		\ar{d}[swap]{p^\opp\times p} \ar{r}{\barF} & \E \ar{d}{p}  \\
		\B^\opp \times \B \ar{r}{F}  & \B
	\end{tikzcd}
\end{equation*}
If each fiber is a quantale, the lifting $\barF$ is \emph{(laxly) monoidal} if for each $X,Y\in \B$ the induced map 
\[
\barF\colon \E_X\to \E_{FX} \qquad / \qquad \barF\colon \E_X\times \E_Y\to \E_{F(X,Y)} \qquad / \qquad  \barF\colon \E_X^\opp\times \E_Y\to \E_{F(X,Y)}
 \] between fibers is a (lax) monoid morphism. Similarly, if each fiber is an an involutive quantale, a lifting is \emph{(laxly) $\circ$-monoidal} if the above induced maps between fibers are (lax) $\circ$-morphisms. 
\end{defn}

\begin{remark}
A lifting of a functor along a $\CLat_\sqcap$-fibration is uniquely determined by its action on objects since the fibration is faithful (\Cref{lem:clat-fib-props}\ref{lem:clat-fib-props-faithful}).
\end{remark}

\begin{example}[Fuzzy Relation Liftings]
We present a number of liftings of functors on $\Set^{(2)}$ along the fibration $p_\jne\colon \FRel_\jne\to \Set^{(2)}$ (\Cref{ex:fibrations}\ref{ex:fibrations-frel2}). These liftings will feature in the categorical modelling of behavioural conformances for our example languages $\pSKI$ and $\pBCK$. The reader may want to consider \Cref{ex:liftings-rel} below in parallel, where the restrictions of the present liftings from fuzzy relations to relations are discussed.

\label{ex:liftings}
\begin{enumerate}
\item\label{ex:lifting:bot} The endofunctor $\{\bot\} + (-)\colon \Set^{(2)}\to \Set^{(2)}$
lifts to the endofunctor
\[(-)_\bot\colon \FRel_\jne\to \FRel_\jne,\qquad (X,d)_\bot = (\{\bot\}+X,d_\bot),\]
where $d_\bot(x,x')=d(x,x')$ for $x,x'\in X$, $d_\bot(\bot,\ol{x})=0$ for $\ol{x}\in \{\bot\}+X$, and $d_\bot(x,\bot)=1$ for $x\in X$.
\item\label{ex:liftings:prod-coprod} The product and coproduct functors $\times,+\colon \Set^{(2)}\times \Set^{(2)}\to \Set^{(2)}$ (which are formed like in $\Set$) lift to the product and coproduct on $\FRel_\jne$. Indeed, the product of fuzzy relations is given by \[(X,d_X)\times (Y,d_Y)=(X\times Y, d_{X\times Y})\quad \text{where}\quad  d_{X\times Y}((x,y),(x',y'))=\max \{\,d_X(x,x'),d_Y(y,y')\,\},\]
and the coproduct by
\[(X,d_X)+ (Y,d_Y)=(X+Y, d_{X+Y})\]
 where $d_{X+Y}(z,z')$ is $d_X(z,z')$ or $d_Y(z,z')$ if $z,z'$ both lie in $X$ or $Y$, respectively, and $1$ otherwise.

\item\label{ex:lifting:prod-non-can} An alternative lifting of the product functor  $\times$ from $\Set^{(2)}$ to $\FRel_\jne$ is given by
\[ (X,d_X)\,\ol{\times}\,(Y,d_Y) = (X\times Y, d_{+}) \quad\text{where}\quad d_{+}((x,y),(x',y'))=d_X(x,x')\boxplus d_Y(y,y'), \]
Recall that $\boxplus$ denotes truncated addition.
\item\label{ex:liftings:hom-functor} The hom bifunctor 
\[H\colon (\Set^{(2)})^\opp\times \Set^{(2)} \to \Set^{(2)},\qquad H(X,Y)=Y^X,\qquad H((f,g),(h,k))=(h\circ - \circ f, k\circ - \circ g),\]
has a lifting given by
\[\barH\colon \FRel_\jne^\opp\times \FRel_\jne \to \FRel_\jne, \qquad
\barH((X,d_X),(Y,d_Y))=(Y,d_Y)^{(X,d_X)}=(Y^X,d_{Y^X}), \] where the fuzzy relation $d_{Y^X}$ on the function space $Y^X$ is given by \[d_{Y^X}(f,g)=\inf \{\, \epsilon\in [0,1] \mid \forall x,x'\in X.\, d_Y(f(x),g(x'))\leq d_X(x,x')+\epsilon \,\}.\] 
Thus $d_{Y^X}(f,g)$ measures the `defect' of joint non-expansivity of the pair $(f,g)$. Note that the infimum is attained, i.e.\ one has \[d_Y(f(x),g(x'))\leq d_X(x,x')+d_{Y^X}(f,g)\quad \text{for all $x,x'\in X$},\] which means precisely that the evaluation map
\[ \ev\colon (Y^X,d_{Y^X})\,\ol{\times}\, (X,d_X)\to (Y,d_Y),\quad (f,x)\mapsto f(x),  \]
is non-expansive. Indeed, $\ol{\times}$ is a closed monoidal structure with internal homs given by $(Y,d_Y)^{(X,d_X)}$.
\item\label{ex:liftings:dist} Let $\D\colon \Set\to\Set$ be the countable distribution functor, and consider its corresponding functor $\D^{(2)}\colon \Set^{(2)}\to\Set^{(2)}$. The \emph{Wasserstein lifting}~\cite{bbkk18} $\barD\colon \FRel_\jne\to \FRel_\jne$ of $\D^{(2)}$ is given by
\[ \barD(X,d)=(\D X, \barD d) \]
where $\barD d$ is the \emph{Wasserstein distance} on $\D X$: given $\phi=\sum_{i\in I} p_i\cdot x_i$ and $\phi'=\sum_{j\in J} p_j'\cdot x_j'$ in $\D X$, the distance $\barD d(\phi,\phi')$ is the value of a minimum solution of the (countably infinite) linear program

\noindent\begin{minipage}[t]{.45\textwidth}
shown on the right. Note that a minimum solution always exists~\cite[Thm.~2.1]{kortanek_yamasaki95}. We refer to this linear program as the \emph{transportation problem from $\phi$ to $\phi'$}. One interprets $p_i$ as the supply of some commodity available at location $x_i$, and $p_j'$ as the demand for that commodity at location $x_j'$, and $d(x_i,x_j')$ as the cost of transporting one unit of the commodity from $x_i$ to $x_j'$. A solution $(t_{i,j})_{i\in I,j\in J}$ 
\end{minipage}
\begin{minipage}[t]{.5\textwidth}
\begin{equation*}
\begin{array}{ll@{}ll}
\text{minimize}  & \displaystyle\sum\limits_{i,j}& \,d(x_i,x_j') \cdot t_{i,j}&\\
\text{subject to}& \displaystyle\sum\limits_{j}&  \, t_{i,j} = p_i  & (i\in I)\\
& \displaystyle\sum\limits_{i}&   \,t_{i,j} = p_j'  & (j\in J) \\
&& \,t_{i,j}\geq 0 & (i\in I,\, j\in J)
\end{array}
\end{equation*}
\end{minipage}

\noindent of the problem is a transportation plan that determines the amount $t_{i,j}$ of the commodity to be transported from $x_i$ to $x_j'$ for each $i,j$ and requires that all supplies and demands are matched. The minimum value of the transportation problem is the cost of a cheapest transportation plan.
\end{enumerate}
\end{example}

Let us mention a few useful properties of the Wasserstein distance:

\begin{lemma}\label{lem:wasserstein-props}
Let $(X,d),(Y,d')$ be fuzzy relations.
\begin{enumerate}
\item\label{lem:wasserstein-props-convex-combinations} For all $\phi,\phi',\psi,\psi'\in \D X$ and $p\in [0,1]$ we have 
\begin{equation}\label{eq:convex-wasserstein}\barD d(p\cdot \phi+(1-p)\cdot \psi,\, p\cdot \phi'+(1-p)\cdot \psi') \leq p\cdot \barD d(\phi,\phi')+(1-p)\cdot \barD d(\psi,\psi').  \end{equation}
\item\label{lem:wasserstein-props-st-non-expansive} The map $\st\colon \barD(X,d)\,\ol{\times}\, (Y,d')\to \barD((X,d)\,{\ol\times}\, (Y,d'))$ is non-expansive.
\item\label{lem:wasserstein-props-eta-non-expansive} The map $\eta\colon (X,d)\to \barD(X,d)$ is non-expansive.
\item\label{lem:wasserstein-props-mu-non-expansive} The map $\mu\colon \barD\,\barD(X,d)\to \barD(X,d)$ is non-expansive.
\end{enumerate}
\end{lemma}
Parts \ref{lem:wasserstein-props-eta-non-expansive} and \ref{lem:wasserstein-props-mu-non-expansive} imply that $\barD$ not only lifts $\D^{(2)}$ as a functor, but as a \emph{monad}.

\begin{example}[Relation Liftings]\label{ex:liftings-rel}
All liftings of \Cref{ex:liftings} restrict to liftings along the fibration $p_\jrp\colon \Rel_\jrp\to \Set^{(2)}$, where $\Rel_\jrp$ is regarded as the full subcategory of $\FRel_\jne$ given by $\{0,1\}$-valued fuzzy relations. We denote the restricted liftings again by $(-)_\bot$, $\barH$, $\barD$ etc. More explicitly:
\begin{enumerate}
\item The relation lifting $(-)_\bot\colon \Rel_\jrp\to \Rel_\jrp$ of $\{\bot\}+(-)$ maps the relation $(X,R)$ to $(\{\bot\}+X,R_\bot)$ where $R_\bot$ is given by $R_\bot(x,x')$ iff either $x,x'\in X$ and $R(x,x')$, or $x=\bot$. 
\item The relation lifting $+\colon \Rel_\jrp\times \Rel_\jrp \to \Rel_\jrp$ of the coproduct maps the relations $(X,R)$ and $(Y,S)$ to $(X+Y,R+S)$. Similarly, the relation lifting $\times \colon \Rel_\jrp\times \Rel_\jrp \to \Rel_\jrp$ of the product maps $(X,R)$ and $(Y,S)$ to $(X\times Y,R\times S)$ where $R\times S((x,y),(x',y'))$ iff $R(x,x')$ and $S(y,y')$.
\item The alternative lifting $\ol{\times}$ of the product coincides with $\times$ in the case of relations.
\item The lifting $\barH\colon \Rel_\jrp^\opp \times \Rel_\jrp \to \Rel_\jrp$ of the hom bifunctor $H$ maps two relations $(X,R)$ and $(Y,S)$ to the relation $(Y^X, S^R)$ where $S^R(f,g)$ iff for all $x,x'\in X$, if $R(x,x')$ then $S(f(x),g(x'))$.
\item The lifting $\barD\colon \Rel_\jrp\to \Rel_\jrp$ of $\D^{(2)}$ maps $(X,R)$ to $(\D X,\barD R)$ where $\barD R(\sum_i p_i\cdot x_i, \sum_j p_j'\cdot x_j')$ iff there exist non-negative reals $t_{i,j}\geq 0$ ($i\in I$, $j\in J$) such that, for all $i\in I$ and $j\in J$,
\[\sum_{j\in J:\, R(x_i,x_j')} t_{i,j} = p_i \qqand \sum_{i\in I:\, R(x_i,x_j')} t_{i,j} = p_j'.\]
This says that there exists a transportation plan sending goods only along connections given by $R$.
\end{enumerate}
\end{example}

We conclude this subsection by noting an important property of $\CLat_\sqcap$-fibrations: free algebras and cofree coalgebras of lifted endofunctors can both be lifted from the base category. To the best of our knowledge, this result is new. In the case where $\B$ has binary products and coproducts, it follows from a corresponding result for initial algebras and final coalgebras~\cite[Lem.~III.5]{fsw24}.

\begin{proposition}[Lifting of Free Algebras]\label{prop:initial-alg-cofree-coalg-lift}
Let $p\colon \E\to\B$ be a $\CLat_\sqcap$-fibration, and suppose that $F\colon \B\to \B$ is an endofunctor with a lifting $\barF\colon \E\to \E$ along $p$. For every $P\in \E_X$, if $F$ has a free algebra $F^{\star}X$ generated by $X$, then $\barF$ has a free algebra $\barF^{\star} P$ generated by $P$. Moreover,
\begin{enumerate}
\item the structure and unit of $\barF^\star P$ are above the structure and unit of $F^{\star}X$, respectively;
\item for every $\barF$-algebra $(A,a)$ and $f\colon P\to A$ in $\E$, the free extension $f^\#\colon \barF^\star P\to (A,a)$ of $f$ is above the free extension $(pf)^\#\colon F^{\star} X\to (pA,pa)$ of $pf\colon X\to pA$.
\end{enumerate}
\end{proposition}
A dual statement applies to cofree coalgebras $F^\infty X$.

\subsection{Behavioural Conformances and Biconformances}
Liftings of behaviour functors along fibrations give rise to an abstract notion of (bi)simulation and its quantitative counterpart, behavioural distances~\cite{hj98}. 

\begin{notation}\label{not:decent}
Let $p\colon \E\to \B$ be a $\CLat_\sqcap$-fibration. We denote an object $P\in \E_X$ by $(X,P)$, alluding to the intuition that $\E$-objects are $\B$-objects with extra structure, and call $P$ a \emph{conformance} over the object $X$~\cite{kkhkh19,bgkmfsw24}. For $(X,P),(Y,Q)$ in $\E$ and $f\colon X\to Y$ in $\B$, we write $f\colon (X,P)\xto{\thickcdot} (Y,Q)$ if there exists a (necessarily unique) morphism $\dot{f}\colon (X,P)\to (Y,Q)$ in $\E$ above $f$. Note that
\[ f\colon (X,P)\xto{\thickcdot} (Y,Q) \;\;\iff\;\; P\sqleq f^{*}Q \;\;\iff\;\; f_{*}P\sqleq Q. \]
\end{notation}

\begin{defn}[Behavioural Conformance]\label{def:simulation}
Let $F\colon \B\to\B$ be an endofunctor with a lifting $\barF\colon \E\to\E$ along the $\CLat_\sqcap$-fibration $p\colon \E\to \B$. An \emph{$\barF$-behavioural conformance} on an $F$-coalgebra $(C,c)$ is a conformance $P\in \E_C$ such that the following equivalent conditions hold:
\[ c\colon (C,P)\xto{\thickcdot} (FC,\barF P) \;\;\iff\;\; P\sqleq c^{*}\barF P \;\;\iff\;\; c_{*}P\sqleq \barF P. \]  
If a $\sqleq$-greatest $\barF$-behavioural conformance on $(C,c)$ exists, we denote it by $P_c^\to$.
\end{defn}

For $\Qtl$-fibrations and under mild assumptions on the lifting $\barF$, the greatest $\barF$-behavioural conformance $P_c^\to$ always exists:

\begin{lemma}\label{lem:sim-props}
Let $F\colon \B\to \B$ be an endofunctor with a laxly monoidal lifting $\barF\colon \E\to \E$ along the $\Qtl$-fibration $p\colon \E\to \B$, and let $(C,c)$ be an $F$-coalgebra.
\begin{enumerate}
\item The transitive closure of a $\barF$-behavioural conformance on $(C,c)$ is a $\barF$-behavioural conformance.
\item $P_c^\to$ is given by the join in $\E_C$ of all $\barF$-behavioural conformances on $(C,c)$. 
\item\label{lem:sim-props-refl-trans} $P_c^\to$ is reflexive and transitive.
\end{enumerate}
\end{lemma}

\begin{example}[Simulations]\label{ex:beh-conformance}
Consider the behaviour bifunctor $B(X,Y)=\D(\{\bot\}+Y^X)$ for $\pSKI$ and $\pBCK$, and the associated \emph{endo}functor $F=B(L,-)\colon \Set\to \Set$ that fixes the contravariant component of $B$ to some set $L$ of labels. Given a coalgebra $c\colon C\to FC$, we write $x\Downarrow \phi$ if $c(x)=\phi$. 
\begin{enumerate}
\item\label{ex:beh-conformance-rel} Let $\barF$ be the following lifting of $F$ along the fibration $p_\rp\colon \Rel_\rp \to \Set$:
\[  \barF(X,R) = (\D(\{\bot\}+ X^L), \barD((R^L)_\bot)) \]
where $(-)_\bot$ is the lifting of $\{\bot\}+(-)$ from \Cref{ex:liftings-rel} and $R^L$ is the $L$-fold product of the relation $(X,R)$. Then a $\barF$-behavioural conformance on a coalgebra $(C,c)$ corresponds to a \emph{probabilistic simulation}, i.e.\ a relation $R\seq C\times C$ such that, for all $x,y\in C$,
\[  R(x,y) \wedge x\Downarrow \phi \wedge y\Downarrow \psi \implies \barD((R^L)_\bot)(\phi,\psi).\] 
The conclusion $\barD((R^L)_\bot)(\phi,\psi)$ says that there exists a transportion plan from $\phi$ to $\psi$ along $(R^L)_\bot$.
\item\label{ex:beh-conformance-frel} We now extend the above above lifting of $F$ to a lifting along the fibration $p_\ne\colon \FRel_\ne \to \Set$:
\[  \barF(X,d) = (\D(\{\bot\}+ X^L), \barD((d^L)_\bot)) \]
where $(-)_\bot$ is the fuzzy relation lifting of $\{\bot\}+(-)$ and $(X,d)^L = (X, d^L)$ is the $L$-fold product of $(X,d)$. Then a $\barF$-behavioural conformance on a coalgebra $(C,c)$ is a \emph{probabilistic fuzzy simulation}, i.e.\ a fuzzy relation $d\colon C\times C\to [0,1]$ such that, for all $x,y\in C$,
\begin{equation}\label{eq:fuzzy-sim}  x\Downarrow \phi \wedge y\Downarrow \psi \implies \barD((d^L)_\bot)(\phi,\psi)\leq d(x,y).\end{equation}
\end{enumerate} 
\end{example}
The notions of probabilistic (fuzzy) simulation apply, in particular, to the operational models of \pSKI and \pBCK, which are $F$-coalgebras $\gamma\colon \Lambda \to B(\Lambda,\Lambda)$ for $\Lambda\in \{ \Lambda_\pSKI, \Lambda_\pBCK \}$. Here $L=\Lambda$.

As indicated by the above examples, behavioural conformances are to be thought of as an abstract notion of simulation (or quantitative versions thereof). This is supported by the fact that the greatest behavioural conformance on a coalgebra is always a preorder (\Cref{lem:sim-props}\ref{lem:sim-props-refl-trans}), but generally not symmetric. To achieve symmetry, we consider behavioural \emph{bi}conformances, which are fibrational abstractions of \emph{bi}simulations and naturally emerge in the setting of $\InQtl$-fibrations:

\begin{defn}[Behavioural Biconformance]\label{def:biconformance}
Let $F\colon \B\to\B$ be an endofunctor with a lifting $\barF\colon \E\to\E$ along the $\InQtl$-fibration $p\colon \E\to \B$. An \emph{$\barF$-behavioural biconformance} on an $F$-coalgebra $(C,c)$ is a conformance $P\in \E_C$ such that both $P$ and $P^\circ$ are $\barF$-behavioural conformances on $(C,c)$. If a $\sqleq$-greatest $\barF$-behavioural biconformance on $(C,c)$ exists, we denote it by $P_c^\leftrightarrow$.
\end{defn}

Much analogous to \Cref{lem:sim-props}, we have:
\begin{lemma}\label{lem:bisim-props}
Let $F\colon \B\to \B$ be an endofunctor with a laxly monoidal lifting $\barF\colon \E\to \E$ along the $\InQtl$-fibration $p\colon \E\to \B$, and let $(C,c)$ be an $F$-coalgebra.
\begin{enumerate}
\item $P_c^\leftrightarrow$ is given by the join in $\E_C$ of all $\barF$-behavioural biconformances on $(C,c)$. 
\item $P_c^\leftrightarrow$ is reflexive, symmetric, and transitive.
\end{enumerate}
\end{lemma}

\begin{example}[Bisimulations]\label{ex:bisim}
\begin{enumerate} 
\item\label{ex:bisim-bisim} In the setting of \Cref{ex:beh-conformance}\ref{ex:beh-conformance-rel}, a
$\barF$-behavioural biconformance on a coalgebra $(C,c)$ corresponds to a \emph{probabilistic bisimulation} $R\seq C\times C$, where 
\[  R(x,y) \wedge x\Downarrow \phi \wedge y\Downarrow \psi \implies \barD((R^L)_\bot)(\phi,\psi) \wedge \barD((R^\circ)^L)_\bot)(\psi,\phi).\] 
 This notion essentially goes back to \citet{ls91}, except that we do not insist on bisimulations to be equivalence relations. However, by \Cref{lem:bisim-props}, \emph{probabilistic bisimilarity} (the greatest probabilistic bisimulation) is always an equivalence relation. 
\item\label{ex:bisim-pseudomet} Similarly, in \Cref{ex:beh-conformance}\ref{ex:beh-conformance-frel}, a $\barF$-behavioural biconformance on $(C,c)$ is precisely a \emph{probabilistic fuzzy bisimulation} $d\colon C\times C\to [0,1]$, where in the conclusion of \eqref{eq:fuzzy-sim} we additionally require $\barD(((d^\circ)^L)_\bot)(\psi,\phi)\leq d(x,y)$. The greatest probabilistic fuzzy bisimulation on  $(C,c)$ is the \emph{behavioural pseudometric}~\cite{bbkk18} on $(C,c)$; note that it is indeed a pseudometric by \Cref{lem:bisim-props}.
\end{enumerate}
\end{example}
The notions of probabilistic bisimilarity and behavioural pseudometric apply to the canonical models of \pSKI and \pBCK, and correspond to analogous notions for probabilistic $\lambda$-calculi~\cite{dsa14,cdl15}.

\subsection{Congruences}

Dually to behavioural conformances on coalgebras, notions of congruence for algebras are modelled by liftings of their syntax functor:
\begin{defn}[Congruence]\label{def:congruence}
Let $\Sigma\colon \B\to\B$ be an endofunctor with a lifting $\barSigma\colon \E\to\E$ along the $\CLat_\sqcap$-fibration $p\colon \E\to \B$. A \emph{$\barSigma$-congruence} on a $\Sigma$-algebra $(A,a)$ is a conformance $P\in \E_A$ such that the following equivalent conditions hold:
\[ a\colon (\Sigma A,\barSigma P)\xto{\thickcdot} (A,P) \;\;\iff\;\; \barSigma P\sqleq a^{*}P \;\;\iff\;\; a_{*}\barSigma P\sqleq P. \]  
\end{defn}

\begin{example}\label{ex:congruence}
Let $\Sigma$ be a signature with its associated polynomial functor $\Sigma\colon \Set\to \Set$. Depending on the choice of fibration and lifting, we recover various notions of congruence on $\Sigma$-algebras:
\begin{enumerate}
\item\label{ex:congruence-rel} For the fibration $p_\rp\colon \Rel_\rp\to \Set$ and the lifting
\[ \barSigma (X,R) = \coprod_{\sigma\in\Sigma} \underbrace{(X,R)\times \cdots \times (X,R)}_{\text{$\ar(\sigma)$ factors}}, \]
a $\barSigma$-congruence on an algebra $A$ is the usual notion from universal algebra: a relation $R\seq A\times A$ such that if $\f\in \Sigma$ is $n$-ary and $R(a_i,a_i')$ for $i=1,\ldots,n$, then $R(\f^A(a_1,\ldots, a_n),\f^A(a_1',\ldots,a_n'))$.
\item\label{ex:congruence-fuzzy} For the fibration $p_\ne\colon \FRel_\ne\to \Set$ and the liftings
\[ \widetilde{\Sigma} (X,d) = \coprod_{\sigma\in\Sigma} \underbrace{(X,d)\times \cdots \times (X,d)}_{\text{$\ar(\sigma)$ factors}} \qqand \barSigma (X,d) = \coprod_{\sigma\in\Sigma} \underbrace{(X,d) \mathbin{\ol{\times}} \cdots \mathbin{\ol{\times}} (X,d)}_{\text{$\ar(\sigma)$ factors}}, \]
a $\widetilde{\Sigma}$-congruence on an algebra $A$ is fuzzy relation $d\colon A\times A\to [0,1]$ satisfying the inequality $d(\f^A(a_1,\ldots,a_n),\f^A(a_1',\ldots,a_n'))\leq \max_{i=1,\ldots,n} d(a_i,a_i')$ for all $a_i,a_i'$ and $n$-ary $\f\in \Sigma$. A $\barSigma$-congruence requires the weaker property $d(\f^A(a_1,\ldots,a_n'),\f^A(a_1',\ldots,a_n'))\leq \sum_{i=1}^n d(a_i,a_i')$ instead.
\end{enumerate}
\end{example}

The next lemma gives some properties of congruences; recall that $(-)^{+}$ is the transitive closure.

\begin{lemma}\label{lem:cong-props} Let $\Sigma\colon \B\to\B$ be an endofunctor with a monoidal lifting $\barSigma\colon \E\to\E$ along the $\Qtl$-fibration $p\colon \E\to \B$, and let $(A,a)$ be a $\Sigma$-algebra. 
\begin{enumerate}
\item\label{lem:cong-props-comp} If $P$, $Q$ are $\barSigma$-congruences on $(A,a)$, then so is $P\cdot Q$.
\item\label{lem:cong-props-transhull} If $\barSigma\colon \E_A\to \E_{\Sigma A}$ preserves directed joins and $P$ is a reflexive $\barSigma$-congruence on $(A,a)$, then so is $P^{+}$.
\end{enumerate}
\end{lemma}

In order to establish congruence properties of behavioural (bi)conformances on models of AHOS, we will make use of a fibrational generalization of \emph{Howe's method}~\cite{DBLP:conf/lics/Howe89,DBLP:journals/iandc/Howe96}. Its key ingredient is a purely algebraic concept, the \emph{Howe closure} of a conformance:

\begin{definition}[Howe Closure]\label{def:howe-closure}
Let $\Sigma\colon \B\to\B$ be an endofunctor with a lifting $\barSigma\colon \E\to\E$ along the $\Qtl$-fibration $p\colon \E\to \B$, and let $(A,a)$ be a $\Sigma$-algebra. The \emph{Howe closure} of a conformance $P\in \E_A$ w.r.t.\ the algebra $(A,a)$ is the $\sqleq$-least conformance $\hat P\in \E_A$ such that
\[ \hat P = P\sqcup (a_{*} \barSigma\hatP)\cdot P,\] 
i.e.\ the least fixed point of the monotone endomap $Q\mapsto P\sqcup (a_{*} \barSigma Q)\cdot P$ on the complete lattice $\E_A$.
\end{definition}

The key properties of the Howe closure are the following:

\begin{lemma}\label{lem:howe-closure-props}
Let $\Sigma\colon \B\to\B$ be an endofunctor with a lifting $\barSigma\colon \E\to\E$ along the $\Qtl$-fibration $p\colon \E\to \B$, let $(A,a)$ be a $\Sigma$-algebra, and let $\hatP$ be the Howe closure of $P\in \E_A$ w.r.t.\ $(A,a)$.  
\begin{enumerate}
\item\label{lem:howe-closure-props-reflexive} If $P$ is reflexive, then $\hat P$ is reflexive and a $\barSigma$-congruence on $(A,a)$.
\item\label{lem:howe-closure-props-weakly-transitive} If $P$ is transitive, then $\hat P\cdot P \sqleq \hat P$.
\item\label{lem:howe-closure-props-symmetric} Suppose that $p$ is an $\InQtl$-fibration, the lifting $\barSigma$ is $\circ$-monoidal, and $\barSigma\colon \E_A\to \E_{\Sigma A}$ preserves directed joins. If $P$ is reflexive and symmetric, then $(\hatP)^{+}$ is symmetric.
\end{enumerate}
\end{lemma}

\section{Congruence Theorems for Abstract Higher-Order Specifications}\label{sec:congruence-ahos}
We now put the categorical machinery of \Cref{sec:ahos,sec:fibrations} together and develop a theory of congruence for behavioural (bi)conformances on canonical models of AHOS, which generalizes corresponding results for applicative (bi)simulations and their quantitative extensions and brings them onto a common foundation. The data for our congruence theorems is given by a complete AHOS (specifiying the syntax and operational semantics of the higher-order language at hand) together with a choice of liftings of its syntax and behaviour functors along a fibration, specifying the targeted forms of congruence and behavioural (bi)conformance.

\begin{definition}[Lifting Situation]\label{def:lifting-situation}
Let $\S=(\B,\Sigma,B,\leq,\rho)$ be a complete AHOS. A \emph{lifting situation} for $\S$ is a triple $\L=(p,\barSigma,\barB)$ given by a $\CLat_\sqcap$-fibration $p\colon \E\to \B^{(2)}$ and liftings $\barSigma$, $\barB$ of the functors $\Sigma^{(2)}$, $B^{(2)}$ (cf.\ \Cref{sec:category-theory}, \emph{Notation}) along the fibration $p$.
\begin{equation*}
	\begin{tikzcd}
		\E \ar{d}[swap]{p}  \ar{r}{\barSigma} & \E \ar{d}{p}  \\
		\B^{(2)} \ar{r}{\Sigma^{(2)}}  & \B^{(2)}
	\end{tikzcd}
\qquad\qquad
\begin{tikzcd}
		\E^\opp\times \E
		\ar{d}[swap]{p^\opp\times p} \ar{r}{\barB} & \E \ar{d}{p}  \\
		(\B^{(2)})^\opp \times \B^{(2)} \ar{r}{B^{(2)}}  & \B^{(2)}
	\end{tikzcd}
\end{equation*}
\end{definition}
In our applications (\Cref{sec:applications}), the fibration $p$ is chosen as $p_\jrp\colon \Rel_\jrp\to \Set^{(2)}$ (corresponding to bisimulations) or $p_\jne\colon \FRel_\jrp\to \Set^{(2)}$ (corresponding to behavioural distances). The requirement for the base category to be $\B^{(2)}$ instead of $\B$ is necessary for our approach to Howe's method.

\begin{notation}\label{not:restricted-liftings}
We let $p_0\colon \E_0\to \B$ denote the fibration arising via the change of base shown in the first diagram below. More explicitly, $\E_0$ is the non-full subcategory of $\E$ that has the same objects as~$\E$, and all morphisms of $\E$ that lie above some $\B^{(2)}$-morphism of the form $(f,f)$. The functor $p_0$ is a restriction of $p$. Note that $\barSigma$ and $\barB$ restrict to liftings of $\Sigma$ and $B$ along $p_0$. We denote the restricted liftings by the same letters $\barSigma$ and $\barB$, as shown in the second and third diagram below:
\[
\begin{tikzcd}
\E_0 \ar[tail]{r} \ar{d}[swap]{p_0} \pullbackangle{-45} \ar{d} & \E \ar{d} \ar{d}{p} \\
\B \ar[tail]{r}{I} & \B^{(2)}
\end{tikzcd}
\qquad\qquad
	\begin{tikzcd}
		\E_0 \ar{d}[swap]{p_0}  \ar{r}{\barSigma} & \E_0 \ar{d}{p_0}  \\
		\B \ar{r}{\Sigma}  & \B
	\end{tikzcd}
\qquad\qquad
\begin{tikzcd}
		\E_0^\opp\times \E_0
		\ar{d}[swap]{p_0^\opp\times p_0} \ar{r}{\barB} & \E_0 \ar{d}{p_0}  \\
		\B^\opp \times \B \ar{r}{B}  & \B
	\end{tikzcd}
\]
\end{notation}
For example, for the fibrations $p_\jrp\colon \Rel_\jrp\to \Set^{(2)}$ and $p_\jne\colon \FRel_\jrp\to \Set^{(2)}$, the change of base yields the respective restrictions $p_\rp\colon \Rel_\rp\to \Set$ and $p_\jne\colon \FRel_\jne\to \Set$ (\Cref{rem:rel-frel}\ref{rem:rel-frel-cob}).

In the following we isolate the two core conditions on a lifting situation that make our abstract congruence results work. The first condition states that the DCPO structure on $\HOCoalg_\Lambda(B)$, which is used to construct the canonical model, interacts well with the order on fibers and reindexing:
\begin{definition}[Compatibility]\label{def:compatibility}
A complete AHOS $\S=(\B,\Sigma,B,\leq,\rho)$ is \emph{compatible} with the lifting situation $\L=(p,\barSigma,\barB)$ if for every $\leq$-directed family $c_i\colon \Lambda\to B(\Lambda,\Lambda)$ ($i\in I$)  and every $d\colon \Lambda\to B(\Lambda,\Lambda)$, one has
 \[
(\bigvee_{i\in I} c_i, d)^{*}\barB(P,Q) = \bigsqcap_{i\in I}\, (c_i,d)^{*}\barB(P,Q)\qquad \text{for all $P,Q\in \E_\Lambda$}.\]  
\end{definition}

The second condition says that the rule morphism should lift along the given fibration:
\begin{definition}[Liftability]\label{def:liftability}
Let $\S=(\B,\Sigma,B,\leq,\rho)$ be an AHOS and let $\L=(p,\barSigma,\barB)$ be a lifting situation for $\S$. We say that the rule morphism $\rho$ is \emph{liftable w.r.t.\ $\L$} if 
\[ \rho\colon \barSigma\barB^\infty(P,P)\xto{\thickcdot} \barB(P,\barSigma^\star P) \qquad \text{for every $P\in \E_\Lambda$.} \]
(Note that $\barSigma\barB^\infty(P,P)$ and $\barB(P,\barSigma^\star P)$ are above $\Sigma B^\infty(\Lambda,\Lambda)$ and $B(\Lambda,\Sigmas\Lambda)$, resp., by \Cref{prop:initial-alg-cofree-coalg-lift}.)
\end{definition}

Liftability turns out to be the key to our congruence results below. Informally, it corresponds to a form of `continuity' requirement on the operational rules represented by $\rho$: whichever conformance, e.g.\ (fuzzy) relation, $P$ is put above $\Lambda$, if the premises of a rule are related according to $P$, then the conclusions should be related according to $P$ as well. In particular, since $P$ is arbitrary, this requires that the operational rules should be parametrically polymorphic, i.e.\ they treat the program terms appearing in them as black-box variables and do not inspect their content. We will support these intuitions in our analysis of \pSKI and \pBCK in \Cref{sec:applications}.

We are prepared to state the first half of our main result, which establishes congruence of the greatest behavioural conformance on the canonical model of a complete AHOS $\S=(\B,\Sigma,B,\leq,\rho)$ w.r.t.~a  lifting situation $\L=(p,\barSigma,\barB)$.  We regard the canonical model $\gamma\colon \Lambda\to B(\Lambda,\Lambda)$ as a coalgebra for the endofunctor $B(\Lambda,-)\colon \B\to \B$, which has a lifting $\barB(1_\Lambda,-)$ along $p_0$ (\Cref{not:restricted-liftings}). The greatest $\barB(1_\Lambda,-)$-behavioural conformance on $(\Lambda,\gamma)$ is denoted by $P_{\gamma}^\to$ as in \Cref{def:simulation}.

\begin{theorem}[Congruence of Behavioural Conformances]\label{thm:app-struct-cong}
Suppose that the following is given:
\begin{itemize}
\item a complete AHOS $\S=(\B,\Sigma,B,\leq,\rho)$ with its canonical model $(\Lambda,\gamma)$;
\item a compatible lifting situation $\L=(p,\barSigma,\barB)$, where $p$ is a heterogeneous $\Qtl$-fibration and the liftings $\barSigma$, $\barB$ are laxly monoidal.
\end{itemize} If the rule morphism $\rho$ is liftable w.r.t.\ $\L$, then $P_{\gamma}^\to$ is a $\barSigma$-congruence on the initial algebra $(\Lambda,\ini)$.
\end{theorem}
The data of the congruence theorem is summarized by the three commutative diagrams below:
\[
\begin{tikzcd}[column sep=5, row sep=10, scale cd=.8]
\Sigma \Lambda \ar{r}{\ini} \ar{dd}[swap]{\Sigma \hat\gamma} & \Lambda \ar{r}{\gamma} & B(\Lambda,\Lambda)\\
&& \\
\Sigma B^\infty(\Lambda,\Lambda) \ar{rr}{\rho} && B(\Lambda,\Sigmas\Lambda) \ar{uu}[swap]{B(\id,\hat\ini)}
\end{tikzcd}
\quad\,
\begin{tikzcd}[column sep=10, row sep=11, scale cd=.8]
\E \ar{rrr}{\barSigma} \ar{ddd}[swap]{p} & & & \E \ar{ddd}{p} \\
& \E_0 \ar[tail]{ul} \ar{r}{\barSigma} \ar{d}[swap]{p_0} & \E_0 \ar{d}{p_0} \ar[tail]{ur} & \\
& \B \ar[tail]{dl} \ar{r}{\Sigma} & \B \ar[tail]{dr} & \\
\B^{(2)} \ar{rrr}{\Sigma^{(2)}} & & & \B^{(2)}  
\end{tikzcd}
\quad\,
\begin{tikzcd}[column sep=10, row sep=10, scale cd=.8]
\E^\opp\times \E \ar{rrr}{\barB} \ar{ddd}[swap]{p} & & & \E \ar{ddd}{p} \\
& \E_0^\opp\times \E_0 \ar[tail]{ul} \ar{r}{\barB} \ar{d}[swap]{p_0} & \E_0 \ar{d}{p_0} \ar[tail]{ur} & \\
& \B^\opp\times \B  \ar[tail]{dl} \ar{r}{B} & \B \ar[tail]{dr} & \\
(\B^{(2)})^\opp \times \B^{(2)} \ar{rrr}{B^{(2)}} & & & \B^{(2)}  
\end{tikzcd}
\]
Our proof of \Cref{thm:app-struct-cong} rests on a fibrational interpretation of Howe's method~\cite{DBLP:conf/lics/Howe89,DBLP:journals/iandc/Howe96}, which makes use of the abstract version of \emph{Howe closure} introduced in \Cref{def:howe-closure}. The heart of the proof lies in the next proposition, which says that Howe closures preserve behavioural conformances:
\begin{proposition}[Howe Closures of Behavioural Conformances]\label{prop:howe-closure-app-struc}
In the setting of \Cref{thm:app-struct-cong}, suppose that $P\in \E_\Lambda$ is a reflexive and transitive $\barB(1_\Lambda,-)$-behavioural conformance on $(\Lambda,\gamma)$. Then the Howe closure~$\hat P$ w.r.t.\ the initial algebra $(\Lambda,\ini)$ is a $\barB(1_\Lambda,-)$-behavioural conformance on $(\Lambda,\gamma)$.
\end{proposition}

Using this proposition, the congruence theorem now easily follows:

\begin{proof}[Proof of \Cref{thm:app-struct-cong}]
Let $P=P_{\gamma}^\to$ be the greatest $\barB(1_\Lambda,-)$-behavioural conformance on $(\Lambda,\gamma)$. It is reflexive and transitive by \Cref{lem:howe-closure-props}. Thus, by \Cref{prop:howe-closure-app-struc}, the Howe closure $\hatP$ is a $\barB(1_\Lambda,-)$-behavioural conformance on $(\Lambda,\gamma)$. It follows that $\hatP=P$;
indeed, $P\sqleq \hatP$ holds by definition of $\hatP$, and $\hatP\sqleq P$ holds because $P$ is the greatest $\barB(1_\Lambda,-)$-behavioural conformance on $(\Lambda,\gamma)$. 
Since $\hatP$ is a $\barSigma$-congruence on $(\Lambda,\ini)$ by \Cref{lem:howe-closure-props}\ref{lem:howe-closure-props-reflexive}, we conclude that $P$ is a $\bar\Sigma$-congruence on $(\Lambda,\ini)$.
\end{proof}

Under slightly stronger assumptions on the fibration and the lifted functors, the congruence theorem extends from behavioural conformances to \emph{bi}conformances. Following \Cref{def:biconformance}, we let $P_{\gamma}^\leftrightarrow$ denote the greatest $\barB(1_\Lambda,-)$-behavioural biconformance on the canonical model $(\Lambda,\gamma)$.

\enlargethispage*{1pt}
\begin{theorem}[Congruence of Behavioural Biconformances]\label{thm:app-struct-cong-bisim}
Suppose that the following is given:
\begin{itemize}
\item a complete AHOS $\S=(\B,\Sigma,B,\leq,\rho)$ with its canonical model $(\Lambda,\gamma)$;
\item a compatible lifting situation $\L=(p,\barSigma,\barB)$, where $p$ is a heterogeneous $\InQtl$-fibration, the lifting $\barSigma$ is $\circ$-monoidal, the lifting $\barB$ is laxly monoidal, and $\barSigma\colon \E_\Lambda\to \E_{\Sigma \Lambda}$ preserves directed joins.
\end{itemize} If the rule morphism $\rho$ is liftable w.r.t.\ $\L$, then $P_{\gamma}^\leftrightarrow$ is a $\barSigma$-congruence on the initial algebra $(\Lambda,\ini)$.
\end{theorem}

The proof uses a fibrational generalization of Howe's \emph{transitive closure trick}~\cite{DBLP:journals/iandc/Howe96,pitts_2011}, which allows to reduce the congruence of behavioural biconformances to that of behavioural conformances.

\begin{proof}[Proof of \Cref{thm:app-struct-cong-bisim}]
Let $P=P^{\leftrightarrow}_{\gamma}$ be the greatest $\barB(1_\Lambda,-)$-behavioural biconformance on $(\Lambda,\gamma)$. It is reflexive, symmetric, and transitive by \Cref{lem:bisim-props}. By \Cref{prop:howe-closure-app-struc}, the Howe closure~$\hatP$ w.r.t.\ $(\Lambda,\ini)$ is a $\barB(1_\Lambda,-)$-behavioural conformance. Then also the transitive closure $(\hatP)^{+}$ of $\hatP$ is a $\barB(1_\Lambda,-)$ behavioural conformance (\Cref{lem:sim-props}), and it is symmetric by \Cref{lem:howe-closure-props}\ref{lem:howe-closure-props-symmetric}. In particular, $(\hatP)^{+}$ is a $\barB(1_\Lambda,-)$-behavioural biconformance, and so
$ (\hatP)^{+}=P $ 
because $P$ is the greatest $\barB(1_\Lambda,-)$-behavioural biconformance. Since $\hatP$ is a $\barSigma$-congruence (\Cref{lem:howe-closure-props}), its transitive closure $(\hatP)^{+}$ is also a $\barSigma$-congruence (\Cref{lem:cong-props}\ref{lem:cong-props-transhull}), and so $P$ is a $\barSigma$-congruence.
\end{proof}

\begin{remark} We have stated \Cref{thm:app-struct-cong,thm:app-struct-cong-bisim} for general AHOS, whose rule morphisms assume existence of a cofree coalgebra $B^\infty(\Lambda,\Lambda)$. However, both theorems and their proofs can be straightforwardly restricted to depth-$n$ AHOS, where a cofree coalgebra is not required.
\end{remark}

Our general congruence theorems for AHOS (\Cref{thm:app-struct-cong,thm:app-struct-cong-bisim}) reduce congruence proofs for higher-order languages to the systematic verification of a number of elementary conditions on the underlying fibrations, liftings, and rules. In comparison to usual ad-hoc congruence proofs, this can greatly simplify the derivation of congruence results for concrete languages, since the generic, boiler-plate aspects and much of the complexity of such proofs are hidden from the user of our theorems. We will demonstrate this point in the next section.

\section{Applications}\label{sec:applications}
We now put our abstract theory to use derive congruence results for the probabilistic higher-order languages \pSKI and \pBCK by simple instantiation of \Cref{thm:app-struct-cong-bisim}.

\subsection{\pBCK}\label{sec:applications-pBCK}
As our first showcase, we study the behavioural pseudometric for the language \pBCK, which corresponds to an affine probabilistic untyped $\lambda$-calculus. For that purpose, we instantiate \Cref{thm:app-struct-cong-bisim} to the AHOS $\S_\pBCK=(\Set,\Sigma_\pBCK,B,\rho_\pBCK)$ for \pBCK (\Cref{ex:ahos-pSKI}) 
and the lifting situation $\L=(p,\barSigma,\barB)$ defined as follows:
\begin{enumerate}
\item $p=p_\jne\colon \Rel_\jne\to \Set^{(2)}$ is the heterogeneous $\InQtl$-fibration of \Cref{ex:fibrations}\ref{ex:fibrations-frel2}.
\item The lifting $\barSigma\colon \FRel_\jne\to \FRel_\jne$ of $\Sigma^{(2)}$ is given as in \Cref{ex:congruence}\ref{ex:congruence-fuzzy}; that is, we lift products via the non-standard lifting $\ol{\times}$ of \Cref{ex:liftings}\ref{ex:lifting:prod-non-can} that computes sums of distances.
\item The lifting $\barB\colon \FRel_\jne^\opp\times \FRel_\jne\to \FRel_\jne$ of $B^{(2)}$ is given by the composite
\[ \barB = (\,\FRel_\jne^\opp\times \FRel_\jne \xto{\barH} \FRel_\jne \xto{(-)_\bot} \FRel_\jne \xto{\barD}  \FRel_\jne\,), \]
where $\barH$ is the lifting of $H(X,Y)=Y^X$ of \Cref{ex:liftings}\ref{ex:liftings:hom-functor}, $(-)_\bot$ is the lifting of the functor $X\mapsto \{\bot\}+X$ of \Cref{ex:liftings}\ref{ex:lifting:bot}, and $\barD$ is the lifting of the distribution functor $\D^{(2)}$ of \Cref{ex:liftings}\ref{ex:liftings:dist}.
\end{enumerate}

For brevity, we drop subscripts and put $\Sigma=\Sigma_\pBCK$, $\Lambda=\Lambda_\pBCK$ and $\rho=\rho_\pBCK$. Moreover, we let $d^\pBCK$ denote the behavioural pseudometric on the operational model $\gamma=\gamma_\pBCK\colon \Lambda\to B(\Lambda,\Lambda)$, i.e.\ the greatest $\barB(1_\Lambda,-)$-fuzzy bisimulation (\Cref{ex:bisim}\ref{ex:bisim-pseudomet}). Our abstract congruence result for behavioural biconformances (\Cref{thm:app-struct-cong-bisim}) then instantiates as follows:
\begin{theorem}[Congruence of $d^{\pBCK}$]\label{thm:cong-bck}
For all $n$-ary operations $\f\in \Sigma_\pBCK$ and $t_i,s_i\in \Lambda$,
\[ d^{\pBCK}(\f(t_1,\ldots,t_n),\f(s_1,\ldots,s_n))\leq d^{\pBCK}(t_1,s_1)+\cdots+d^{\pBCK}(t_n,s_n).\] 
\end{theorem}
For the proof, all we need to do is to systematically check that the conditions of \Cref{thm:app-struct-cong-bisim} hold for $\S$ and $\L$. This is done in the next two lemmas. Let us start with the required properties of the lifted functors, which follow via a routine verification:

\begin{lemma}\label{lem:lifting-compatible}
The lifting situation $\L=(p,\barSigma,\barB)$ is compatible with $\S_\pBCK$, the lifting $\barSigma$ is $\circ$-monoidal, the lifting $\barB$ is laxly monoidal, and the map $\barSigma\colon \FRel_{\jne,\Lambda}\to \FRel_{\jne,\Sigma \Lambda}$ preserves directed joins.
\end{lemma} 

The key condition of \Cref{thm:app-struct-cong-bisim}, liftability of $\rho$ w.r.t.\ $\L$, corresponds to the following statement:

\begin{lemma}\label{lem:rho-pBCK-non-exp}
For every fuzzy relation $d$ on $\Lambda$, the rule map \eqref{eq:rho-pBCK-nonexp} is non-expansive:
\begin{equation}\label{eq:rho-pBCK-nonexp} \rho=\rho_\pBCK \c \barSigma((\Lambda,d)\times \barB((\Lambda,d),(\Lambda,d))\times \barB((\Lambda,d),\barB((\Lambda,d),(\Lambda,d)))) \to
    \barB((\Lambda,d), \barSigma^{\star} (\Lambda,d)).\end{equation}
\end{lemma}

\begin{proof}
For ease of notation, we use the letter $\bard$ to denote all the fuzzy relations induced by $d$ via the various liftings, e.g.\ the fuzzy relations on
 \[(\Lambda,d)^{(\Lambda,d)},\quad
 \barB((\Lambda,d),(\Lambda,d)),\quad \barB((\Lambda,d), \barB((\Lambda,d),(\Lambda,d))), \quad \barSigma^\star(\Lambda,d).\]
For example, we put $(\Lambda,d)^{(\Lambda,d)}=(\Lambda^\Lambda,\bard)$ etc. Note that $\times$ in the domain of $\rho$ is the categorical product in $\FRel_\ne$ given by the maximum of distances, while the products appearing in the polynomial functor~$\Sigma$ are lifted via the sum of distances. Thus, non-expansivity of \eqref{eq:rho-pBCK-nonexp} means that for every $\f$ in $\Sigma$, every $(t_i,\phi_i,\Phi_i),(t_i',\phi_i',\Phi_i')\in \Lambda\times {B}(\Lambda,\Lambda)\times  {B}(\Lambda, {B}(\Lambda,\Lambda))$ and every $\epsilon_i\in [0,1]$ ($i=1,\ldots,n$),
\begin{align*}
& d(t_i,t_i),\, \bard(\phi_i,\phi_i'),\, \bard(\Phi_i,\Phi_i')\leq \epsilon_i \quad (i=1,\ldots,n) \\
\implies~&
 \bard(\rho(\f((t_1,\phi_1,\Phi_1),\ldots, (t_n,\phi_n,\Phi_n))), \rho(\f((t_1',\phi_1',\Phi_1'),\ldots, (t_n',\phi_n',\Phi_n')))) 
\leq 
\textstyle \sum_{i=1}^n \epsilon_i. \end{align*}
We will now verify this property for every operator $\f$.\\

\noindent\fbox{$\Omega$}~\,We have $\bard(\rho(\Omega),\rho(\Omega)) = \bard(1\cdot \bot,1\cdot\bot)=0$, using that $\bard(\bot,\bot)=0$.\\

\noindent\fbox{$B,B',B'',C,C',C'',K,K',I$}~\,On each of these operations $\f$, the map $\rho$ is defined as
\[ \rho(\f((t_1,-,-),\ldots,(t_n,-,-))) = 1\cdot (t\mapsto T)    \]
where $T=T[t,t_1,\ldots,t_n]\in \Sigmas\Lambda$ is a term in the variables $t,t_1,\ldots,t_n\in \Lambda$ and each of the variables appears at most once in $T$, that is, $T$ is an \emph{affine} term.  Now let $t_i,t_i'\in \Lambda$ and suppose that $d(t_i,t_i')\leq \epsilon_i$ for $i=1,\ldots,n$. By affinity of $T$ and the definition of the fuzzy relation $\bard$ on the term algebra $\barSigma^\star(\Lambda,d)$,
\[  \bard(T[t,t_1,\dots,t_n], T[t',t_1',\ldots,t_n']) \leq d(t,t') + \sum_{i=1}^n \epsilon_i \qquad \text{for all $t,t'\in \Lambda$}, \]
where $T[t',t_1',\cdots,t_n']$ is obtained from $T$ by substituting the variables $t,t_1,\ldots,t_n$ for $t',t_1',\ldots,t_n'$.
This means that the functions $t\mapsto T[t,t_1,\ldots,t_n]$ and $t\mapsto T[t,t_1',\ldots,t_n']$ have distance at most $\sum_{i=1}^n \epsilon_i$ in the hom-set fuzzy relation from \Cref{ex:liftings}\ref{ex:liftings:hom-functor}, which proves that
\begin{align*}
&  \,\bard(\rho(\f((t_1,-,-),\ldots,(t_n,-,-)),\rho(\f((t_1',-,-),\ldots,(t_n',-,-))))\\ 
=& \,\bard(1\cdot (t\mapsto T[t,t_1,\ldots,t_n]),1\cdot (t\mapsto T[t,t_1'\ldots,t_n'])) \leq \textstyle\sum_{i=1}^n \epsilon_i. \end{align*} 

\noindent\fbox{$\oplus$}~\, For $k=1,2$ let $\phi_k,\phi_k'\in B(\Lambda,\Lambda)=\D(\{\top\} + \Lambda^\Lambda)$, and suppose
that $\bard(\phi_k,\phi_k')\leq \epsilon_k$. Then
\begin{align*}
&~\bard(\rho((-,\phi_1,-) \oplus (-,\phi_2,-)), \rho((-,\phi_1',-) \oplus (-,\phi_2',-)))\\
=&~ \bard(\frac{1}{2}\cdot \phi_1+\frac{1}{2}\cdot \phi_2,\frac{1}{2}\cdot \phi_1'+\frac{1}{2}\cdot \phi_2') \leq  \frac{1}{2}\cdot \bard(\phi_1,\phi_1')+\frac{1}{2}\cdot \bard(\phi_2,\phi_2') \leq \frac{1}{2}\cdot \epsilon_1+\frac{1}{2}\cdot \epsilon_2 \leq \epsilon_1+\epsilon_2
\end{align*}
where the inequality in the second step follows from \Cref{lem:wasserstein-props}\ref{lem:wasserstein-props-convex-combinations}.\\

\noindent\fbox{$\mathsf{app}$}~\, Let $\Phi,\Phi'\in \D(\{\bot\}+(\D(\{\bot\}+\Lambda^\Lambda))^{\Lambda})$ and $s,s'\in \Lambda$ and suppose that $\bard(\Phi,\Phi')\leq \epsilon_1$ and $d(s,s')\leq \epsilon_2$. Using the map $j$ from \Cref{rem:app-mu}, we have
\begin{align*} \bard(\app((-,-,\Phi),(s,-,-)), \app((-,-,\Phi'),(s',-,-)) 
= \bard(j(\Phi,s),j(\Phi',s')) 
\leq \bard(\Phi,\Phi')+d(s,s') 
\leq \epsilon_1+\epsilon_2. 
\end{align*}
In the second step we use that
$j\colon B((\Lambda,d),B((\Lambda,d),(\Lambda,d)))\, \ol{\times}\, (\Lambda,d) \to B((\Lambda,d),(\Lambda,d))$ is non-expansive. Indeed, all the maps appearing in the definition of $j$ are non-expansive: for $\st$, $\eta$, $\mu$, this follows from \Cref{lem:wasserstein-props}, for $\ev$ from \Cref{ex:liftings}\ref{ex:liftings:hom-functor}, and for  $\outl$, $\inl$, $\nabla$ it is clear. \qedhere
\end{proof}

The congruence property of $d^\pBCK$ has some useful consequences. In particular, it entails that $d^\pBCK$ relates well to the \emph{contextual pseudometric} $d_\ctx^\pBCK$ (\Cref{sec:prob-comb-logic}). Let us first note that the latter has a convenient alternative characterization.  We say that a pseudometric $d$ on $\Lambda$ is \emph{adequate (for termination)} if 
$ d(t,s) \geq |\gamma(t)(\bot)-\gamma(s)(\bot)|$ for all $t,s\in \Lambda$. Then the following holds:

\begin{proposition}\label{prop:ctx-greatest}
The pseudometric $d_\ctx^\pBCK$ is the $\sqleq$-greatest adequate $\barSigma$-congruence on $(\Lambda,\ini)$.
\end{proposition}

It follows that $d^\pBCK$ yields a sound method for proving upper bounds to contextual distances:
\begin{corollary}[Soundness of $d^\pBCK$]\label{cor:sound}
For all $t,s\in \Lambda $ one has $d_\ctx^\pBCK(t,s)\leq d^\pBCK(t,s)$.
\end{corollary}

\begin{proof}
By \Cref{prop:ctx-greatest} and since the pseudometric $d^\pBCK$ is a $\barSigma$-congruence (\Cref{thm:cong-bck}), we only need to prove that $d^\pBCK$ is adequate. This follows from the observation that given $\phi,\psi\in \D(\{\bot\}+\cdots)$,  every transportation plan from $\phi$ to $\psi$ must send at least $|\phi(\bot)-\psi(\bot)|$ units from locations distinct from $\bot$ to the location $\bot$, which incurs a cost of at least $|\phi(\bot)-\psi(\bot)|$.
\end{proof}

\begin{example}
We prove  $d_\ctx^\pBCK(I,I\oplus \Omega)=\frac{1}{2}$. For `$\geq$', consider the empty context. For `$\leq$', put $d=d^\pBCK$. By \Cref{cor:sound} it suffices to show $d(I,I\oplus \Omega)\leq\frac{1}{2}$. Since $d$ is a fuzzy simulation, for that it is enough to show that $\gamma(I)=1\cdot (t\mapsto t)$ and $\gamma(I\oplus \Omega)=\frac{1}{2}\cdot \bot + \frac{1}{2}\cdot (t\mapsto t)$ have distance at most $\frac{1}{2}$ w.r.t.~the fuzzy relation on $\D(\{\bot\}+\Lambda^{\Lambda})$ induced by $d$. This holds as $d$ is reflexive.
\end{example}

\subsection{\pSKI}\label{sec:applications-pSKI}
As our second application, we consider the language $\pSKI$, which corresponds to a (non-affine) probabilistic untyped $\lambda$-calculus. Unlike $d^\pBCK$, it turns out the behavioural pseudometric $d^\pSKI$ is not a congruence. This is due to the fact that non-affine contexts can amplify distances, leading to trivialization of the contextual pseudometric~\cite{cdl15}. In our categorical setup, the failure of congruence is conceptually reflected by the fact that the rule map $\rho_\pSKI$ of $\pSKI$ does not lift along the fibration $p_\jne\colon \FRel_\jrp\to\Set$. Indeed, the argument for the combinator case in \Cref{lem:rho-pBCK-non-exp} uses that the behaviour of all combinators is given by affine terms, which is not true for the combinator $S$. 

However, \pSKI still has a well-behaved qualitative notion of bisimulation, which emerges in our setting by restricting fuzzy relations to relations. Thus we instantiate \Cref{thm:app-struct-cong-bisim} to the AHOS $\S_\pSKI=(\Set,\Sigma_\pSKI,B,\rho_\pSKI)$
and the lifting situation $\L=(p,\barSigma,\barB)$ defined as follows:
\begin{enumerate}
\item $p=p_\jrp\colon \Rel_\jrp\to \Set^{(2)}$ is the heterogeneous $\InQtl$-fibration of \Cref{ex:fibrations}\ref{ex:fibrations-rel}.
\item The lifting $\barSigma\colon \Rel_\jrp\to \Rel_\jrp$ of $\Sigma^{(2)}\colon \Set^{(2)}\to \Set^{(2)}$ is given by \Cref{ex:congruence}\ref{ex:congruence-rel}.
\item The lifting $\barB\colon \Rel_\jrp^\opp\times \Rel_\jrp\to \Rel_\jrp$ of $B^{(2)}$ is given by the composite
\[ \barB = (\,\Rel_\jrp^\opp\times \Rel_\jrp \xto{\barH} \Rel_\jrp \xto{(-)_\bot} \Rel_\jrp \xto{\barD}  \Rel_\jrp\,), \]
using the liftings $\barH$, $(-)_\bot$, $\barD$ as defined in \Cref{ex:liftings-rel}.
\end{enumerate}
The derivation of the congruence result is now much analogous to the case of $\pBCK$. In the following we put $\Sigma=\Sigma_\pSKI$, $\Lambda=\Lambda_\pSKI$ and $\rho=\rho_\pSKI$. Moreover, we let $\approx^\pSKI$ denote the bisimilarity relation on the canonical model $\gamma=\gamma_\pSKI\colon \Lambda\to B(\Lambda,\Lambda)$, that is, the greatest $\barB(1_\Lambda,-)$-bisimulation (\Cref{ex:bisim}\ref{ex:bisim-bisim}). Our abstract congruence result for behavioural biconformances (\Cref{thm:app-struct-cong-bisim}) then instantiates to:
\begin{theorem}[Congruence of $\approx^\pSKI$]\label{thm:cong-pSKI}
For all $n$-ary operations $\f\in \Sigma_\pSKI$ and $t_i,s_i\in \Lambda$,
\[ t_1\approx^\pSKI s_1\wedge \ldots \wedge t_n\approx^\pSKI s_n \implies \f(t_1,\ldots,t_n)\approx^{\pSKI} \f(s_1,\ldots,s_n). \] 
\end{theorem}
Again, we only need to check the conditions of \Cref{thm:app-struct-cong-bisim}. 

\begin{lemma}
The lifting situation $\L=(p,\barSigma,\barB)$ is compatible with $\S_\pSKI$, the lifting $\barSigma$ is $\circ$-monoidal, the lifting $\barB$ is laxly monoidal, and the map $\barSigma\colon \FRel_{\jne,\Lambda}\to \FRel_{\jne,\Sigma \Lambda}$ preserves directed joins.
\end{lemma}

This is immediate from the corresponding statement for $\pBCK$ (\Cref{lem:lifting-compatible}) since the lifting situation for $\pSKI$ restricts that for $\pBCK$. Liftability of $\rho$ w.r.t.\ $\L$ amounts to the following:

\begin{lemma}\label{lem:rho-pSKI-non-exp}
For every relation $R$ on $\Lambda$, the rule map \eqref{eq:rho-pSKI-relpre} is relation-preserving:
\begin{equation}\label{eq:rho-pSKI-relpre} \rho=\rho_\pSKI \c \barSigma((\Lambda,R)\times \barB((\Lambda,R),(\Lambda,R))\times \barB((\Lambda,R),\barB((\Lambda,R),(\Lambda,R)))) \to
    \barB((\Lambda,R), \barSigma^{\star} (\Lambda,R)). \end{equation}
\end{lemma}
The proof is identical to that of \Cref{lem:rho-pBCK-non-exp}, restricted to relations (i.e.\ $\{0,1\}$-valued fuzzy relations). We only need to observe that in the proof case for combinators, the affinity of the term $T$ is not needed if $d$ is $\{0,1\}$-valued. Therefore, for such $d$ the argument also applies to $S$, $S'$, $S''$. 

Analogously to \Cref{cor:sound}, we get soundness of bisimulations for \emph{contextual equivalence}:
\begin{corollary}[Soundness of $\approx^\pSKI$] 
For all $t,s\in \Lambda_\pSKI$, if $t\approx^\pSKI s$ then $t\approx^\pSKI_\ctx s$.
\end{corollary}

The reader should have noted that our proofs of \Cref{thm:cong-bck,thm:cong-pSKI} are fairly short, modular, and almost mechanical. This simplicity is enabled by the clean separation and isolation of required conditions on the lifted functors and the operational rules of the language in our abstract congruence results (\Cref{thm:app-struct-cong,thm:app-struct-cong-bisim}). In particular, we mention that our above proofs are markedly simpler than existing congruence proofs for probabilistic $\lambda$-calculi, which rest on non-trivial calculations based on linear programming duality~\cite{cdl15} or combinatorial properties of sets of real numbers~\cite{dsa14}.

\subsection{Languages with Variables}
 While we have considered combinatory logics for simplicity, languages with variables can be modelled in the AHOS framework as well, by moving from sets to presheaves. This works much like in the higher-order abstract GSOS framework~\cite{gmstu23}. We sketch the idea for the call-by-name $\lambda$-calculus with its big-step operational semantics. Take the base category $\B=\Set^\fset$, where $\fset$ is the category of finite cardinals and functions. Following Fiore, Turi, and Plotkin~\cite{DBLP:conf/lics/FiorePT99}, the presheaf $\Lambda\in \Set^\fset$ of $\lambda$-terms sending $n\in\fset$ to the set of $\lambda$-terms (modulo $\alpha$-equivalence) with free variables from $n$, forms the initial algebra for the endofunctor $\Sigma$ on $\Set^\fset$ given by $(\Sigma X)(n) = n+X(n+1)+X(n)\times X(n)$. Its summands correspond to the constructors of $\lambda$-terms (variables, $\lambda$-abstraction, application). We take the behaviour bifunctor $B$ on $\Set^\fset$ given by $B(X,Y)(n)=(\{\bot\}+\Set^\fset(X^n,Y))\times (\{\bot\}+Y^X(n))$.
The first component captures substitution behaviour (e.g.\ every $\lambda$-term $t\in \Lambda(n)$ yields a natural transformation $(\phi_t)_m\colon \Lambda^n(m)\to \Lambda(m)$ given by $\vec{s} \mapsto t[\vec{s}]$), and the second component captures big-step transitions, where $Y^X$ is the exponential in $\Set^\fset$. (For probabilistic $\lambda$-calculi, one composes the second component with $\D$.) Analogous to ~\cite[Sec.~5]{gmstu23}, the operational semantics can be encoded into an AHOS $\S_\lambda=(\B,\Sigma,B,\rho)$ satisfying all the requirements of our abstract congruence results.

\section{Conclusion and Future Work}\label{sec:conclusion}
We have introduced \emph{abstract higher-order specifications} (\emph{AHOS}), a categorical framework for specifying operational semantics of higher-order languages that combines the strengths of \emph{higher-order abstract GSOS} and \emph{monotone biGSOS} and does away with some of their limitations. Our main application of the new framework is a theory of congruence of behavioural conformances, parametric in a choice of a fibration and based on a fibrational generalization of Howe's method.

While we have put the focus of this paper on the foundational aspects of the AHOS framework and restricted ourselves to relatively simple instantiations, we expect that advanced language features that have previously been modelled in the higher-order abstract GSOS framework, such as call-by-value and call-by-push-value evaluation~\cite{gtu25}, simple~\cite{gmstu24} and recursive~\cite{gmtu24lics} types, and higher-order store~\cite{gmstu25} are captured in much the same way in the AHOS framework.

As a challenging application of our results, we aim to study probabilistic higher-order languages over \emph{continuous} distributions in the AHOS framework, aiming for a theory of behavioural distance for such languages. So far, only a qualitative notion of bisimilarity has emerged in the literature~\cite{dg19}.

A further natural step is to develop a fibrational theory of \emph{logical relations}, which besides coinduction are the most popular operational technique for establishing contextual equivalence. We expect that the recent work on fibrational differential logical relations~\cite{dg24} will provide inspiration. If logical relations and their congruence properties can be developed at the generality of AHOS, they may apply to relational reasoning far beyond contextual equivalence, e.g.~to \emph{parametricity}~\cite{reynolds83}.

\clearpage

\section*{Acknowledgement}
The author wishes to thank Paul Wild and Jonas Forster for discussions on the Wasserstein metric.

\bibliographystyle{ACM-Reference-Format}
\bibliography{mainBiblio}

@PREAMBLE{ {\providecommand{\noopsort}[1]{}} }

@string{acm="ACM"}

@string{springer="Springer"}

@string{lipics="LIPIcs"}

@string{dagstuhl="Schloss Dagstuhl -- Leibniz-Zentrum f{\"u}r Informatik"}

@misc{u25_arxiv,
      title={Higher-Order Behavioural Conformances via Fibrations}, 
      author={Henning Urbat},
      year={2025},
      eprint={2507.18509},
      archivePrefix={arXiv},
      primaryClass={cs.PL},
      url={https://arxiv.org/abs/2507.18509},
}

@misc{gpt25,
      title={Big Steps in Higher-Order Mathematical Operational Semantics}, 
      author={Sergey Goncharov and Pouya Partow and Stelios Tsampas},
      year={2025},
      eprint={2506.01076},
      archivePrefix={arXiv},
      primaryClass={cs.PL},
      url={https://arxiv.org/abs/2506.01076}, 
      note={To appear in Proc.~ICFP 2025}
}

@inproceedings{bgkmfsw24,
  author       = {Harsh Beohar and
                  Sebastian Gurke and
                  Barbara K{\"{o}}nig and
                  Karla Messing and
                  Jonas Forster and
                  Lutz Schr{\"{o}}der and
                  Paul Wild},
  OPTeditor       = {Olaf Beyersdorff and
                  Mamadou Moustapha Kant{\'{e}} and
                  Orna Kupferman and
                  Daniel Lokshtanov},
  title        = {Expressive Quantale-Valued Logics for Coalgebras: An Adjunction-Based
                  Approach},
  booktitle    = {41st International Symposium on Theoretical Aspects of Computer Science,
                  {STACS} 2024, March 12-14, 2024, Clermont-Ferrand, France},
  series       = {LIPIcs},
  volume       = {289},
  pages        = {10:1--10:19},
  publisher    = {Schloss Dagstuhl - Leibniz-Zentrum f{\"{u}}r Informatik},
  year         = {2024},
  doi          = {10.4230/LIPICS.STACS.2024.10},
}

@inproceedings{reynolds83,
  author       = {John C. Reynolds},
  OPTeditor       = {R. E. A. Mason},
  title        = {Types, Abstraction and Parametric Polymorphism},
  booktitle    = {Information Processing 83, Proceedings of the {IFIP} 9th World Computer
                  Congress, Paris, France, September 19-23, 1983},
  pages        = {513--523},
  publisher    = {North-Holland/IFIP},
  year         = {1983},
}

@article{dg24,
  author       = {Francesco Dagnino and
                  Francesco Gavazzo},
  title        = {A Fibrational Tale of Operational Logical Relations: Pure, Effectful
                  and Differential},
  journal      = {Log. Methods Comput. Sci.},
  volume       = {20},
  number       = {2},
  year         = {2024},
  doi          = {10.46298/LMCS-20(2:1)2024},
}

@InProceedings{dgy19,
  author =	{Dal Lago, Ugo and Gavazzo, Francesco and Yoshimizu, Akira},
  title =	{{Differential Logical Relations, Part I: The Simply-Typed Case}},
  booktitle =	{46th International Colloquium on Automata, Languages, and Programming (ICALP 2019)},
  pages =	{111:1--111:14},
  series =	{Leibniz International Proceedings in Informatics (LIPIcs)},
  year =	{2019},
  volume =	{132},
  OPTeditor =	{Baier, Christel and Chatzigiannakis, Ioannis and Flocchini, Paola and Leonardi, Stefano},
  publisher =	{Schloss Dagstuhl -- Leibniz-Zentrum f{\"u}r Informatik},
  doi =		{10.4230/LIPIcs.ICALP.2019.111},
}

@inproceedings{gavazzo18,
  author       = {Francesco Gavazzo},
  OPTeditor       = {Anuj Dawar and
                  Erich Gr{\"{a}}del},
  title        = {Quantitative Behavioural Reasoning for Higher-order Effectful Programs:
                  Applicative Distances},
  booktitle    = {Proceedings of the 33rd Annual {ACM/IEEE} Symposium on Logic in Computer
                  Science, {LICS} 2018, Oxford, UK, July 09-12, 2018},
  pages        = {452--461},
  publisher    = {{ACM}},
  year         = {2018},
  doi          = {10.1145/3209108.3209149},
}

@misc{gmstu25,
      title={Bialgebraic Reasoning on Stateful Languages}, 
      author={Sergey Goncharov and Stefan Milius and Lutz Schröder and Stelios Tsampas and Henning Urbat},
      year={2025},
      eprint={2503.10955},
      archivePrefix={arXiv},
      primaryClass={cs.PL},
      url={https://arxiv.org/abs/2503.10955}, 
      note={To appear in Proc.~ICFP 2025}
}

@INPROCEEDINGS{jr99,
  author={Jeffrey, A. and Rathke, J.},
  booktitle={Proceedings. 14th Symposium on Logic in Computer Science (Cat. No. PR00158)}, 
  title={Towards a theory of bisimulation for local names}, 
  year={1999},
  volume={},
  number={},
  pages={56-66},
  doi={10.1109/LICS.1999.782586}}

@InProceedings{ls15,
  author =	{Lenglet, Serguei and Schmitt, Alan},
  title =	{{Howe's Method for Contextual Semantics}},
  booktitle =	{26th International Conference on Concurrency Theory (CONCUR 2015)},
  pages =	{212--225},
  series =	{Leibniz International Proceedings in Informatics (LIPIcs)},
  year =	{2015},
  volume =	{42},
  editor =	{Aceto, Luca and de Frutos Escrig, David},
  publisher =	{Schloss Dagstuhl -- Leibniz-Zentrum f{\"u}r Informatik},
  doi =		{10.4230/LIPIcs.CONCUR.2015.212},
}

@inproceedings{dgl17,
  author       = {Ugo Dal Lago and
                  Francesco Gavazzo and
                  Paul Blain Levy},
  title        = {Effectful applicative bisimilarity: Monads, relators, and Howe's method},
  booktitle    = {32nd Annual {ACM/IEEE} Symposium on Logic in Computer Science, {LICS}
                  2017, Reykjavik, Iceland, June 20-23, 2017},
  publisher    = {{IEEE} Computer Society},
  year         = {2017},
  doi          = {10.1109/LICS.2017.8005117},
}

@article{dg19,
title = {On Bisimilarity in Lambda Calculi with Continuous Probabilistic Choice},
journal = {Electronic Notes in Theoretical Computer Science},
volume = {347},
pages = {121-141},
year = {2019},
note = {Proceedings of the Thirty-Fifth Conference on the Mathematical Foundations of Programming Semantics},
doi = {https://doi.org/10.1016/j.entcs.2019.09.007},
author = {Ugo Dal Lago and Francesco Gavazzo},
}

@article{gtu25,
author = {Goncharov, Sergey and Tsampas, Stelios and Urbat, Henning},
title = {Abstract Operational Methods for Call-by-Push-Value},
year = {2025},
publisher = {Association for Computing Machinery},
volume = {9},
number = {POPL},
doi = {10.1145/3704871},
journal = {Proc. ACM Program. Lang.},
articleno = {35},
numpages = {27},
}

@book{sangiorgi11, 
place={Cambridge}, 
title={Introduction to Bisimulation and Coinduction},
 publisher={Cambridge University Press}, 
author={Sangiorgi, Davide}, 
year={2011}}

@inproceedings{dsa14,
  author       = {Ugo Dal Lago and
                  Davide Sangiorgi and
                  Michele Alberti},
  OPTeditor       = {Suresh Jagannathan and
                  Peter Sewell},
  title        = {On coinductive equivalences for higher-order probabilistic functional
                  programs},
  booktitle    = {The 41st Annual {ACM} {SIGPLAN-SIGACT} Symposium on Principles of
                  Programming Languages, {POPL} '14, San Diego, CA, USA, January 20-21,
                  2014},
  pages        = {297--308},
  publisher    = {{ACM}},
  year         = {2014},
  doi          = {10.1145/2535838.2535872},
}

@article{ls91,
  author       = {Kim Guldstrand Larsen and
                  Arne Skou},
  title        = {Bisimulation through Probabilistic Testing},
  journal      = {Inf. Comput.},
  volume       = {94},
  number       = {1},
  pages        = {1--28},
  year         = {1991},
  doi          = {10.1016/0890-5401(91)90030-6},
}

@article{rot19,
  author       = {Jurriaan Rot},
  title        = {Distributive laws for monotone specifications},
  journal      = {Acta Informatica},
  volume       = {56},
  number       = {7-8},
  pages        = {585--617},
  year         = {2019},
  doi          = {10.1007/S00236-019-00333-X},
}

@article{fsw24,
  author       = {Jonas Forster and
                  Lutz Schr{\"{o}}der and
                  Paul Wild},
  title        = {Conformance Games for Graded Semantics},
  journal      = {CoRR},
  volume       = {abs/2411.03069},
  year         = {2024},
  url          = {https://doi.org/10.48550/arXiv.2411.03069},
  doi          = {10.48550/ARXIV.2411.03069},
  eprinttype    = {arXiv},
  eprint       = {2411.03069},
  bibsource    = {dblp computer science bibliography, https://dblp.org}
}

@article{dallago_zorzi12,
  author       = {Ugo Dal Lago and
                  Margherita Zorzi},
  title        = {Probabilistic operational semantics for the lambda calculus},
  journal      = {{RAIRO} Theor. Informatics Appl.},
  volume       = {46},
  number       = {3},
  pages        = {413--450},
  year         = {2012},
  doi          = {10.1051/ITA/2012012},
}

@article{kortanek_yamasaki95,
title = {Discrete infinite transportation problems},
journal = {Discrete Applied Mathematics},
volume = {58},
number = {1},
pages = {19-33},
year = {1995},
doi = {https://doi.org/10.1016/0166-218X(93)E0139-P},
author = {Kenneth O. Kortanek and Maretsugu Yamasaki},
}

@InProceedings{GianantonioRPO,
author="Di Gianantonio, Pietro
and Honsell, Furio
and Lenisa, Marina",
editor="Amadio, Roberto",
title="RPO, Second-Order Contexts, and $\lambda$-Calculus",
booktitle="Foundations of Software Science and Computational Structures",
year="2008",
publisher="Springer Berlin Heidelberg",
address="Berlin, Heidelberg",
pages="334--349",
isbn="978-3-540-78499-9"
}

@INPROCEEDINGS{cdl15,
  author={Crubill\'e, Raphaëlle and Dal Lago, Ugo},
  booktitle={30th Annual ACM/IEEE Symposium on Logic in Computer Science (LICS 2015)}, 
  title={Metric reasoning about $\lambda$-terms: The affine case}, 
  year={2015},
  pages={633-644},
  publisher = {IEEE},
}

@INPROCEEDINGS{kkhkh19,
  author={Komorida, Yuichi and Katsumata, Shin-ya and Hu, Nick and Klin, Bartek and Hasuo, Ichiro},
  booktitle={34th Annual ACM/IEEE Symposium on Logic in Computer Science (LICS 2019)}, 
  title={Codensity Games for Bisimilarity}, 
  year={2019},
  pages={1-13},
  publisher={IEEE},
}

@BOOK{jacobs99, 
AUTHOR       = "B. Jacobs", 
TITLE        = "Categorical Logic and Type Theory",
PUBLISHER    = "North Holland", 
SERIES       = "Studies in Logic and the Foundations of Mathematics",
NUMBER       = "141",
YEAR         = "1999"}

@Book{ahs90,
  author = 	 {Ad\'amek, Ji\v{r}\'i and Herrlich, Horst and Strecker, George E.},
  ALTeditor = 	 {},
  title = 	 {Abstract and Concrete Categories},
  publisher = 	 {John Wiley and Sons},
  year = 	 {1990},
  OPTkey = 	 {},
  OPTvolume = 	 {},
  OPTnumber = 	 {},
  OPTseries = 	 {},
  OPTaddress = 	 {},
  OPTedition = 	 {},
  OPTmonth = 	 {},
  note = 	 {Free online version: \url{http://www.tac.mta.ca/tac/reprints/articles/17/tr17abs.html}},
  OPTannote = 	 {},
}

@article{hj98,
  title = {Structural {Induction} and {Coinduction} in a {Fibrational Setting}},
  author = {Hermida, Claudio and Jacobs, Bart},
  year = {1998},
  OPTmonth = sep,
  journal = {Information and Computation},
  volume = {145},
  number = {2},
  pages = {107--152},
  issn = {0890-5401},
  doi = {10.1006/inco.1998.2725}
}

@inproceedings{gmstu23,
author = {Goncharov, Sergey and Milius, Stefan and Schr\"{o}der, Lutz and Tsampas, Stelios and Urbat, Henning},
title = {Towards a Higher-Order Mathematical Operational Semantics},
booktitle = {50th ACM SIGPLAN Symposium on Principles of Programming Languages (POPL 2023)},
year = {2023},
issue_date = {January 2023},
publisher = acm,
OPTaddress = {New York, NY, USA},
volume = {7},
OPTnumber = {POPL},
doi = {10.1145/3571215},
series = {Proc. ACM Program. Lang.},
OPTmonth = {Jan},
articleno = {22},
numpages = {27}
}

@phdthesis{morris,
  title = {Lambda-Calculus Models of Programming Languages},
  school = {Massachusetts Institute of Technology},
  author = {Morris, James H.},
  year = {1968}
}

@inproceedings{DBLP:conf/lics/Howe89,
  author    = {Douglas J. Howe},
  title     = {Equality In Lazy Computation Systems},
  booktitle = {4th Annual Symposium on Logic in Computer Science (LICS 1989)},
  pages     = {198--203},
  publisher = {{IEEE} Computer Society},
  year      = {1989},
  doi       = {10.1109/LICS.1989.39174},
  bibsource = {dblp computer science bibliography, https://dblp.org}
}

@article{DBLP:journals/iandc/Howe96,
  author    = {Douglas J. Howe},
  title     = {Proving Congruence of Bisimulation in Functional Programming Languages},
  journal   = {Inf. Comput.},
  volume    = {124},
  number    = {2},
  pages     = {103--112},
  year      = {1996},
  doi       = {10.1006/inco.1996.0008},
  bibsource = {dblp computer science bibliography, https://dblp.org}
}

@inproceedings{DBLP:conf/lics/FiorePT99,
  author    = {Marcelo P. Fiore and
               Gordon D. Plotkin and
               Daniele Turi},
  title     = {Abstract Syntax and Variable Binding},
  booktitle = {14th Annual {IEEE} Symposium on Logic in Computer Science (LICS 1999)},
  pages     = {193--202},
  publisher = {{IEEE} Computer Society},
  year      = {1999},
  doi       = {10.1109/LICS.1999.782615},
  bibsource = {dblp computer science bibliography, https://dblp.org}
}

@article{bbkk18,
  author       = {Paolo Baldan and
                  Filippo Bonchi and
                  Henning Kerstan and
                  Barbara K{\"{o}}nig},
  title        = {Coalgebraic Behavioral Metrics},
  journal      = {Log.~Methods~Comput.~Sci.},
  volume       = {14},
  number       = {3},
  year         = {2018},
}

@article{pitts1997operationally,
  title={Operationally-based theories of program equivalence},
  author={Pitts, Andrew M},
  journal={Semantics and Logics of Computation},
  volume={14},
  pages={241},
  year={1997}
}

@book{mac2013categories,
  title = {Categories for the {{Working Mathematician}}},
  author = {Mac Lane, S.},
  year = {1978},
  series = {Graduate {{Texts}} in {{Mathematics}}},
  edition = {2},
  volume = {5},
  publisher = springer,
  OPTaddress = {{New York}},
  url = {http://link.springer.com/10.1007/978-1-4757-4721-8},
  isbn = {978-0-387-98403-2},
  langid = {english}
}

@InCollection{Abramsky:lazylambda,
  author =       "S.  Abramsky",
  booktitle =    "Research topics in Functional Programming",
  publisher =    "Addison Wesley",
  title =        "The lazy $\lambda$-calculus",
  year =         "1990",
  pages =        "65--117",
  document-size = "129. 7 kbytes",
}

@article{DBLP:journals/lmcs/HirschowitzL22,
  author    = {Tom Hirschowitz and
               Ambroise Lafont},
  title     = {A categorical framework for congruence of applicative bisimilarity
               in higher-order languages},
  journal   = {Log. Methods Comput. Sci.},
  volume    = {18},
  number    = {3},
  year      = {2022},
  url       = {https://doi.org/10.46298/lmcs-18(3:37)2022},
  doi       = {10.46298/lmcs-18(3:37)2022},
  bibsource = {dblp computer science bibliography, https://dblp.org}
}

@inbook{pitts_2011,
  place={Cambridge},
  series={Cambridge Tracts in Theoretical Computer Science},
  title={Howe's method for higher-order languages},
  DOI={10.1017/CBO9780511792588.006},
  booktitle={Advanced Topics in Bisimulation and Coinduction},
  publisher={Cambridge University Press},
  author={Pitts, Andrew},
  editor={Sangiorgi, Davide and Rutten, JanEditors},
  year={2011},
  pages={197232},
  collection={Cambridge Tracts in Theoretical Computer Science}
}

@inproceedings{DBLP:conf/lics/BorthelleHL20,
  author    = {Peio Borthelle and
               Tom Hirschowitz and
               Ambroise Lafont},
  editor    = {Holger Hermanns and
               Lijun Zhang and
               Naoki Kobayashi and
               Dale Miller},
  title     = {A Cellular {Howe} Theorem},
  booktitle = {35th Annual {ACM/IEEE} Symposium on Logic in Computer
               Science, LICS'20},
  pages     = {273--286},
  publisher = {{ACM}},
  year      = {2020},
  url       = {https://doi.org/10.1145/3373718.3394738},
  doi       = {10.1145/3373718.3394738},
  bibsource = {dblp computer science bibliography, https://dblp.org}
}

@inproceedings{DBLP:conf/lics/LagoGL17,
  author    = {Dal Lago, Ugo and
               Gavazzo, Francesco and
               Levy, Paul Blain},
  title     = {Effectful applicative bisimilarity: Monads, relators, and {H}owe's method},
  booktitle = {32nd Annual {ACM/IEEE} Symposium on Logic in Computer Science (LICS 2017)},
  pages     = {1--12},
  publisher = {{IEEE} Computer Society},
  year      = {2017},
  doi       = {10.1109/LICS.2017.8005117},
  bibsource = {blp computer science bibliography, https://dblp.org}
}

@InProceedings{UrbatTsampasEtAl23,
  author        = {Henning Urbat and Stelios Tsampas and Sergey Goncharov and
                  Stefan Milius and Lutz Schr{\"o}der},
  title         = {Weak Similarity in Higher-Order Mathematical Operational
                  Semantics},
  year          = "2023",
  publisher     = {IEEE Computer Society Press},
  booktitle     = {38th Annual ACM/IEEE Symposium on Logic in Computer Science (LICS 2023)},
  OPTeditor        = {Igor Walukiewicz},
  doi           = {10.1109/LICS56636.2023.10175706}
}

@book{hindley2008lambda,
  title = {Lambda-{{Calculus}} and {{Combinators}}: {{An Introduction}}},
  shorttitle = {Lambda-{{Calculus}} and {{Combinators}}},
  author = {Hindley, J. Roger and Seldin, Jonathan P.},
  year = {2008},
  edition = {2},
  publisher = {{Cambridge University Press}},
  OPTaddress = {{Cambridge}},
  doi = {10.1017/CBO9780511809835},
  isbn = {978-0-521-89885-0}
}

@article{gmstu24,
author = {Goncharov, Sergey and Santamaria, Alessio and Schr\"{o}der, Lutz and Tsampas, Stelios and , Henning},
title = {Logical Predicates in Higher-Order Mathematical Operational Semantics},
year = {2024},
editor = {Naoki Kobayashi and James Worrell},
booktitle = {Foundations of Software Science and Computation Structures - 27th
                  International Conference, FoSSaCS 2024}
}

@InProceedings{gmtu24lics,
  author        = {Sergey Goncharov and Stefan Milius and Stelios Tsampas and Henning Urbat},
  title     = {Bialgebraic Reasoning on Higher-Order Program Equivalence},
  year          = "2024",
  publisher     = {IEEE Computer Society Press},
  booktitle     = {39th Annual ACM/IEEE Symposium on Logic in Computer Science (LICS 2024)},
  doi           = {10.1145/3661814.3662099},
  note = {Preprint: \url{https://arxiv.org/abs/2402.00625}}
}

\clearpage
\appendix 

\section{Complete Lattice Fibrations and Topological Functors}\label{app:fibrations-vs-topological}
In this appendix we give a proof of the folklore result that $\CLat_\sqcap$-fibrations are exactly fiber-small topological functors (\Cref{thm:fib-vs-top}). Let us first recall the notion of topological functor~\cite{ahs90}:

\begin{defn}[Topological functor]
Let $p\colon \E\to\B$ be a functor.
\begin{enumerate}
\item A \emph{$p$-source} is a possible large family of morphisms $(f_i\colon X\to pY_i)_{i\in I}$ in $\B$.
\item An \emph{initial lift} of $(f_i)$ is a source $(\barf_i\colon \ol{X}\to Y_i)_{i\in I}$ in $\E$ such that $p\barf_i = f_i$ for each $i\in I$, and moreover for every source $(g_i\colon Z\to Y_i)$ and $m\colon pZ\to X$ such that $pg_i = f_i\circ m$ for all $i\in I$, there exists a unique $\ol{m}\colon Z\to Y$ such that $p\ol{m}=m$ and $g_i=\barf_i\circ \ol{m}$ for all $i\in I$. The dual notion is that of a \emph{final lift} of a \emph{$p$-sink} $(f_i\colon pX_i\to Y)_{i\in I}$. 
\item $p$ is a \emph{topological functor} if every $p$-source has a unique initial lift.
\item A topological functor is \emph{fiber-small} if each fiber $\E_X$, $X\in \B$, is a small category.
\end{enumerate}
\end{defn}

\begin{lemma}[\cite{ahs90}, Exercise 21C]\label{lem:top-char}
A functor $p\colon \E\to \C$ is topological iff the following conditions hold:
\begin{enumerate}
\item Each fiber is a (possibly large) complete lattice.
\item Each morphism $f\colon X\to pY$ (viewed as a singleton $p$-source) has a unique initial lift.
\item Each morphism $f\colon pX\to Y$ (viewed as singleton $p$-sink) has a unique final lift.
\end{enumerate} 
\end{lemma}

Recall the definition of a fibration $p\colon \E\to \B$  (\Cref{def:fibration}). We say that $p$ is a \emph{poset fibration} if each fiber $\E_X$ ($X\in B$) is a small poset. Thus every $\CLat_\sqcap$-fibration is a poset fibration.

\begin{lemma}\label{lem:pos-fib-faithful}
In a poset fibration, cartesian lifts are unique.
\end{lemma}

\begin{proof}
By its universal property, the cartesian lift of a morphism $f\colon X\to pY$ is unique up to isomorphism in the fiber $\E_X$. Hence it is unique as isomorphisms in the poset $\E_X$ are identities.
\end{proof}

\begin{lemma}[\cite{jacobs99}, Lemma 9.1.2]\label{lem:bifib-char}
A fibration $p\colon \E\to \B$ is a bifibration iff for each $f\colon X\to Y$ in $\B$ the reindexing functor $f^{*}\colon \E_Y\to \E_X$ has a left adjoint.
\end{lemma}

The desired correspondence result is now a straightforward consequence of the above lemmas:

\begin{theorem}\label{thm:fib-vs-top}
A functor $p\colon \E\to \B$ is a $\CLat_\sqcap$-fibration iff it is topological and fiber-small.
\end{theorem}

\begin{proof}
($\To$) Suppose that $p$ is a $\CLat_\sqcap$-fibration. Since each reindexing functor $f^{*}$ preserves meets, it has a left adjoint by the adjoint functor theorem, so $p$ is a bifibration (\Cref{lem:bifib-char}).  To prove that $p$ is topological, we verify the conditions of \Cref{lem:top-char}. Condition (1) holds by hypothesis. Condition (2) follows from $p$ being a fibration and by \Cref{lem:pos-fib-faithful}; note that an initial lift of a morphism is the same as a cartesian lift. Dually, condition (3) follows from $p$ being an opfibration. 

\medskip\noindent
($\Leftarrow$) Suppose that $p$ is a fiber-small topological functor. By condition (1) of \Cref{lem:top-char}, each fiber is a complete lattice. By conditions (2) and (3) of \Cref{lem:top-char}, the functor $p$ is a bifibration; in particular, it is a fibration, and each reindexing functor $f^{*}$ has a left adjoint (\Cref{lem:bifib-char}), which implies that it preserves meets. Therefore, $p$ is a $\CLat_\sqcap$-fibration.
\end{proof}

\section{Omitted Proofs}

\subsection*{Proof of \Cref{lem:clat-fib-props}}
The four statements of the lemma are immediate from the fact every $\CLat_\sqcap$-fibration $p\colon \E\to \B$ is a topological functor (\Cref{thm:fib-vs-top}).
\begin{enumerate}
\item $p$ is faithful by \cite[Thm.~21.3]{ahs90}.
\item $p$ is a bifibration by \Cref{lem:top-char}.
\item There is a Galois connection $f_{*} \dashv f^{*}$ by \Cref{lem:bifib-char}.
\item Since $f^{*}\circ g^{*}$ is a cartesian lift of $gf$, we have $f^{*}\circ g^{*} = (gf)^{*}$ by the uniqueness of cartesian lifts (\Cref{lem:pos-fib-faithful}). The other statement is dual. 
\end{enumerate}

\subsection*{Details for \Cref{ex:iqtl-fibration}}
We prove that $p_\jne\colon \FRel_\jne\to \Set^{(2)}$ is a heterogeneous $\InQtl$-fibration (the proof for $\Rel_\jrp$ is analogous).
Let us verify the three axioms \eqref{eq:het-qtl-fibration-axioms}. 

For the first axiom  of \eqref{eq:het-qtl-fibration-axioms}, let $d_{=}$ denote the fuzzy relation corresponding to the identity relation, i.e.\ $d_{=}(x,x')=0$ if $x=x'$ and $d_{=}(x,x')=1$ otherwise. Then
\[ (f,f)^{*} d_{=}(x,x') = d_{=}(f(x),f(x')) \leq d_{=}(x,x'). \]
For the last inequality one distinguishes two cases:  for $x\neq x'$ it holds trivially since $d_{=}(x,x')=1$, and for $x=x'$ it holds because both sides are equal to $0$.

For the second axiom of \eqref{eq:het-qtl-fibration-axioms}, let $d,e\colon Y\times Y\to [0,1]$ be fuzzy relations. Then
\begin{align*} (f,h)^{*}(d\cdot e)(x,x') &= (d\cdot e)(f(x),h(x')) \\
&= \inf_{y\in Y} d(f(x),y) + e(y,h(x')) \\
&\leq \inf_{x''\in X} d(f(x),g(x''))+e(g(x''),h(x'))\\
&= \inf_{x''\in X} (f,g)^{*} d(x,x'')+(g,h)^{*} e(x'',x')\\
&= ((f,g)^{*} d \cdot (g,h)^{*} e)(x,x').
\end{align*}
The third axiom of \eqref{eq:het-qtl-fibration-axioms} actually holds with equality: for every $d\colon Y\times Y\to [0,1]$, we have
\[ ((f,f)^{*} d)^\circ(x,x') = ((f,f)^{*} d)(x',x) = d(f(x'),f(x)) = d^\circ(f(x),f(x')) = ((f,f)^{*} d^\circ)(x,x').\]

\subsection*{Details for \Cref{ex:liftings}\ref{ex:liftings:hom-functor}}
We prove that $\barH$ is properly defined on morphisms, that is, for any two pairs of jointly non-expansive maps $(f,g)\colon (\ol{X},d_{\ol{X}})\to (X,d_X)$ and $(h,k)\colon (Y,d_Y)\to (\ol{Y},d_{\ol{Y}})$, the pair $H((f,g),(h,k))$ is jointly non-expansive from $(Y^X,d_{Y^X})$ to $(\ol{Y}^{\ol{X}},d_{\ol{Y}^{\ol{X}}})$:
\begin{equation}\label{eq:hom-functor-lift-proof} d_{\ol{Y}^{\ol{X}}}(h\circ u\circ f, k\circ v \circ g ) \leq d_{Y^X}(u,v) \qquad \text{for all $u,v\in Y^X$}.  \end{equation}
To see this, we compute for all $x,x'\in \ol{X}$:
\begin{align*}
d_{\ol{Y}}(h\circ u\circ f(x), k\circ v\circ g(x')) & \leq d_Y(u\circ f(x),v\circ g(x')) & \text{$(h,k)$ jointly non-exp.}\\
&\leq d_X(f(x),g(x'))+d_{Y^X}(u,v) & \text{def.\ $d_{Y^X}$}\\
&\leq d_{\ol{X}}(x,x')+d_{Y^X}(u,v) & \text{$(f,g)$ jointly non-exp.}
\end{align*}
which proves \eqref{eq:hom-functor-lift-proof} by definition of $d_{\ol{Y}^{\ol{X}}}$.

\subsection*{Details for \Cref{ex:liftings}\ref{ex:liftings:dist}}
Let us check that $\D$ is properly defined on morphisms, that is, for every jointly non-expansive pair $(f,g)\colon (X,d_X)\to (Y,d_Y)$ the pair $(\D f,\D g)\colon (\D X, \bard_X)\to (\D Y, \bard_Y)$ is jointly non-expansive:
\[ \bard_Y(\D f(\phi),\D g(\phi'))\leq \bard_X(\phi,\phi')\qquad\text{for all $\phi,\phi'\in \D X$}. \]
 This follows from the observation that every solution $(t_{i,j})$ for the transportation from $\phi=\sum_{i\in I} p_i\cdot x_i$ and $\phi'=\sum_{j\in J} p_j'\cdot x_j'$ in $\D X$ is also a solution for the transportation problem from $\D f(\phi)=\sum_{i\in I} p_i\cdot f(x_i)$ and $\D g(\phi')=\sum_{j\in J} p_j'\cdot g(x_j')$, and moreover the cost of the latter is 
\[\sum_{i,j} d_Y(f(x_i),g(x_j'))\cdot t_{i,j}\leq \sum_{i,j} d_X(x_i,x_{j'})\cdot t_{i,j},\]
i.e.\ it is bounded above by the cost the solution for the original transportation problem.

\subsection*{Proof of \Cref{lem:wasserstein-props}}
Given a fuzzy relation $(X,d)$ we put $\bard = \barD d$.
\begin{enumerate}
\item Let $p\in [0,1]$ and let $\phi,\phi',\psi,\psi'\in \D X$ be distributions
\[ \phi=\sum_{i\in I} p_i\cdot x_i,\quad \phi'=\sum_{j\in J} p_j'\cdot x_j',\quad \psi=\sum_{i\in I} q_i\cdot x_i,\quad \psi'=\sum_{j\in J} q_j'\cdot x_j', \]
so that
\[ p\cdot \phi+(1-p)\cdot\psi= \sum_{i\in I} (p\cdot p_i + (1-p)\cdot q_i)\cdot x_i \qqand  p\cdot \phi'+(1-p)\cdot \psi'= \sum_{j\in J} (p\cdot p_j'+(1-p)\cdot q_j')\cdot x_j'.\]
Let $(s_{i,j})_{i,j}$ and $(t_{i,j})_{i,j}$ be optimal solutions of the transportation problem from $\phi$ to $\phi'$ and from $\psi$ to $\psi'$, respectively. Thus
\[
 \displaystyle\sum\limits_{j} s_{i,j} = p_{i}\; (i\in I),\qquad \displaystyle\sum\limits_{i}   s_{i,j} = p_{j}'  \; (j\in J),\qquad \displaystyle\sum\limits_{i,j} {d}(x_i,x_j') \cdot s_{i,j} = \bard(\phi,\phi').
\]
\[
 \displaystyle\sum\limits_{j} t_{i,j} = q_{i}\; (i\in I),\qquad \displaystyle\sum\limits_{i}  t_{i,j} = q_{j}'  \; (j\in J),\qquad \displaystyle\sum\limits_{i,j} {d}(x_i,x_j') \cdot t_{i,j} = \bard(\psi,\psi').
\] 
Then $(r_{i,j})_{i,j}:=(p\cdot s_{i,j}+(1-p)\cdot t_{i,j})_{i,j}$ is a solution of the transportation problem from $p\cdot \phi+(1-p)\cdot\psi$ to $p\cdot \phi'+(1-p)\cdot \psi'$, that is,
\[
 \displaystyle\sum\limits_{j} r_{i,j} = p\cdot p_{i} + (1-p)\cdot q_{i} \; (i\in I),\qquad \displaystyle\sum\limits_{i}   r_{i,j} =  p\cdot p_{j}' + (1-p)\cdot q_j'  \; (j\in J).
\]
Indeed, the first equation follows from
\[  
\sum_{j} r_{i,j}= \sum_j p\cdot s_{i,j}+(1-p)\cdot t_{i,j} = p\cdot \sum_j s_{i,j} + (1-p)\cdot \sum_j t_{i,j} =  p\cdot p_i + (1-p)\cdot q_i.
\]
The proof of the second equation is analogous. The cost of the solution $(r_{i,j})_{i,j}$ is
\[ \sum_{i,j} d(x_i,x_j')\cdot r_{i,j} = p\cdot \sum_{i,j} d(x_i,x_j')\cdot s_{i,j} + (1-p)\cdot \sum_{i,j} d(x_i,x_j')\cdot t_{i,j} = p\cdot \bard(\phi,\psi)+(1-p)\cdot\bard(\phi',\psi'). \]
Thus $(r_{i,j})_{i,j}$ witnesses the inequality \eqref{eq:convex-wasserstein}.
\item Let $\phi,\phi'\in \D X$ and $y,y'\in Y$, where
\[ \phi=\sum_{i\in I} p_i\cdot x_i,\quad \phi'=\sum_{j\in J} p_j'\cdot x_j',\quad \st(\phi,y)=\sum_{i\in I} p_i\cdot (x_i,y),\quad \st(\phi',y')=\sum_{j\in J} p_j'\cdot (x_j',y'). \]
Letting $\bard_{+}$ denote the Wasserstein metric on $\barD((X,d)\ol{\times}\,(Y,d'))$, we need to prove
\begin{equation}\label{eq:st-non-exp} 
\bard_{+}(\st(\phi,y),\st(\phi',y')) \leq \bard(\phi,\phi')+d'(y,y').
\end{equation}
Let $(t_{i,j})_{i,j}$ be an optimal solution of the transportation problem from $\phi$ to $\phi'$; that is,
\[
 \displaystyle\sum\limits_{j} t_{i,j} = p_{i}\; (i\in I),\qquad \displaystyle\sum\limits_{i}   t_{i,j} = p_{j}'  \; (j\in J),\qquad \displaystyle\sum\limits_{i,j} {d}(x_i,x_j') \cdot t_{i,j} = \bard(\phi,\phi').
\]
Then $(t_{i,j})_{i,j}$ is also a solution of the transportation problem from $\st(\phi,y)$ to $\st(\phi',y')$ of cost
\begin{align*}
\sum_{i,j} d_{+}((x_i,y),(x_j',y'))\cdot t_{i,j} &= \sum_{i,j} (d(x_i,x_j')+d'(y,y'))\cdot t_{i,j}\\
&=\sum_{i,j} d(x_i,x_j')\cdot t_{i,j} + d'(y,y')\cdot \sum_{i,j}t_{i,j} \\
&= \bard(\phi,\phi')+d'(y,y').
\end{align*}
using that $\sum_{i,j} t_{i,j}=1$. This proves \eqref{eq:st-non-exp}.
\item For $x,x'\in X$ the unique solution of the transportation problem from $\eta(x)=1\cdot x$ to $\eta(x')=1\cdot x'$ is given by sending all available supply from $x$ to $x'$; the cost of this solution is $d(x,x')$. Therefore
\[ \bard(\eta(x),\eta(x')) = d(x,x'),\]
proving that $\eta$ is non-expansive (in fact, isometric).
\item To show non-expansivity of $\mu$, let $\Phi,\Phi'\in \D\D X$, say
\[ \Phi = \sum_{i\in I} p_i\cdot \phi_i\qquad\text{where}\qquad \phi_i=\sum_{k\in K_i} p_{i,k}\cdot x_{i,k}, \] 
\[ \Phi' = \sum_{j\in J} p_j'\cdot \phi_j'\qquad\text{where}\qquad \phi_j'=\sum_{l\in L_j} p_{j,l}'\cdot x_{j,l}'. \] 
Letting $\tilde{d}$ denote the Wasserstein distance on $\barD\barD (X,d)$, we need to prove
\begin{equation}\label{eq:mu-non-exp}
\bard(\mu(\Phi),\mu(\Phi'))=\bard(\sum_{i\in I}\sum_{k\in K_i} p_i\cdot p_{i,k}\cdot x_{i,k},\, \sum_{j\in J}\sum_{l\in L_j} p_{j}'\cdot p_{j,l}'\cdot x_{j,l}') \leq \tilde{d}(\Phi,\Phi').
\end{equation}
Let $(s_{i,j})_{i,j}$ be an optimal solution of the transportation problem from $\Phi$ to $\Phi'$; that is,
\[
 \displaystyle\sum\limits_{j} s_{i,j} = p_{i}\; (i\in I),\qquad \displaystyle\sum\limits_{i}   s_{i,j} = p_{j}'  \; (j\in J),\qquad \displaystyle\sum\limits_{i,j} \bard(\phi_i,\phi_j') \cdot s_{i,j} = \tilde{d}(\Phi,\Phi').
\]
Moreover, for each $i\in I$ and $j\in J$ let $(t_{k,l}^{i,j})_{k\in K_i,l\in L_j}$ be an optimal solution of the transportation problem from $\phi_i$ to $\phi_j'$; that is,
\[
 \displaystyle\sum\limits_{l\in L_j} t_{k,l}^{i,j} = p_{i,k}\; (k\in K_i),\qquad \displaystyle\sum\limits_{k\in K_i}   t_{k,l}^{i,j} = p_{j,l}'  \; (l\in L_j),\qquad \displaystyle\sum\limits_{k\in K_i,l\in L_j} d(x_{i,k},x_{j,l}') \cdot t_{k,l}^{i,j} = \bard(\phi_i,\phi_j').
\]
Then $(r_{(i,k),(j,i)})_{(i,k),(j,l)} := (s_{i,j}\cdot t_{k,l}^{i,j})_{(i,k),(j,l)}$ is a solution to the transportation problem from $\mu(\Phi)$ to $\mu(\Phi')$; that is,
\[ \sum_{j\in J}\sum_{l\in L_j} r_{(i,k),(j,l)} = p_i\cdot p_{i,k} \; (i\in I,\, k\in K_i),\qquad  \sum_{i\in I}\sum_{k\in K_i} r_{(i,k),(j,l)} = p_j'\cdot p_{j,l}' \; (j\in J,\, l\in L_j). \]
Indeed, the first equation follows from the computation
\[ \sum_{j}\sum_l r_{(i,k),(j,l)} = \sum_{j}\sum_l s_{i,j}\cdot t_{k,l}^{i,j} = \sum_j s_{i,j} \sum_l t_{k,l}^{i,j} = \sum_j s_{i,j}\cdot p_{i,k} = p_i\cdot p_{i,k};  \]
analogously for the second equation. The cost of the solution $(r_{(i,k),(j,i)})_{(i,k),(j,l)}$ is
\begin{align*}
\sum_{i,k,j,l} d(x_{i,k},x'_{j,l})\cdot r_{(i,k),(j,l)} &= \sum_{i,k,j,l} d(x_{i,k},x'_{j,l})\cdot s_{i,j}\cdot t_{k,l}^{i,j}\\
 &= \sum_{i,j} s_{i,j}\cdot \sum_{k,l} d(x_{i,k},x'_{j,l})\cdot  t_{k,l}^{i,j}\\
&= \sum_{i,j} s_{i,j}\cdot \bard(\phi_i,\phi_j') \\
&= \tilde{d}(\Phi,\Phi').
\end{align*}
This proves \eqref{eq:mu-non-exp}.
\end{enumerate}

\subsection*{Proof of \Cref{prop:initial-alg-cofree-coalg-lift}}
Let $F^\star X$ be the free $F$-algebra on $X$, with structure and unit 
\[\ini_X\colon FF^{\star}X \to F^{\star}X\qqand\eta_X\colon X\to F^{\star} X\]
For every $\barF$-algebra $(A,a)$, the pair $(pA,pa)$ is an $F$-algebra. Therefore, for every $f\colon P\to A$ in $\E$, the morphism $pf\colon X\to pA$ uniquely extends to an $F$-algebra morphism \[(pf)^\#\colon (F^{\star}X,\ini_X)\to (pA,pa)\qquad\text{such that}\qquad pf=(pf)^\#\circ \eta_X. \] 
We define the object $\barF^\star P\in \E_{F^{\star}X}$ by
\[\barF^{\star} P:=\bigsqcap_{(A,a)\in \Alg(\ol{F}),\, f\colon P\to A} ((pf)^{\#})^{*} A.\]
The proposition then immediately follows from the following statements:
\begin{enumerate}[label=(\alph*)]
\item\label{proof:a} There exists a (necessarily unique) morphism $\eta_P\colon P\to \barF^{\star}P$ above $\eta_X$.
\item\label{proof:b} There exists a (necessarily unique) morphism $\ini_P\colon \barF\barF^{\star}P\to \barF^{\star}P$ above $\ini_X$.
\item\label{proof:c} $(\barF^\star P,\ini_P)$ is the free $\barF$-algebra with unit $\eta_P$: for every $(A,a)\in \Alg(\barF)$ and $f\colon P\to A$ in $\E$, there exists unique $\barF$-algebra morphism $f^\#\colon (\barF^{\star}P,\ini_P) \to (A,a)$ such that $f=f^\#\circ \eta_P$. Moreover, the morphism $f^\#$ is above $(pf)^\#$.
\end{enumerate}
\medskip\noindent \emph{Proof of \ref{proof:a}}:
We need to prove $P\sqleq \eta_X^{*}\barF^{\star}P$. For every $(A,a)\in \Alg(\barF)$ and $f\colon P\to A$ we have
\[ P\sqleq (pf)^{*}A = ((pf)^\#\circ \eta_X)^{*}A = \eta_X^{*}((pf)^\#)^{*}A, \]
where the first step uses that $f\colon P\to A$ in $\E$ and the second step uses \Cref{lem:clat-fib-props}\ref{lem:clat-fib-probs-split}. Since $\eta_X^{*}$ preserves meets, it follows that 
\[ P \sqleq \bigsqcap_{(A,a)\in \Alg(\ol{F}),\, f\colon P\to A} \eta_X^{*}((pf)^{\#})^{*} A = \eta_X^{*} \bigsqcap_{(A,a)\in \Alg(\ol{F}),\, f\colon P\to A} ((pf)^{\#})^{*} A = \eta_X^{*}\barF^{\star}P.  \]

\medskip\noindent\emph{Proof of \ref{proof:b}}: We need to prove $\barF \barF^{\star}P \sqleq \ini_X^{*}\barF^{\star}P$. Fix $(A,a)\in \Alg(\barF)$ and $f\colon P\to A$, and 
denote the cartesian lift of $(pf)^\#\colon F^{\star}X\to pA$ by 
\[\ol{(pf)^\#}\colon ((pf)^\#)^{*}A\to A.\]
Then the diagram below commutes, where $i\colon \barF^{\star}P\to ((pf)^\#)^{*}A$ is the unique morphism in $\E_{F^{\star}X}$ witnessing that $\barF^{\star}P\sqleq ((pf)^\#)^{*}A$.
\[
\begin{tikzcd}
p\barF\barF^{\star}P \ar{r}{p\barF i}  \ar[equals]{d} & p\barF((pf)^\#)^{*}A \ar[equals]{d} \ar{rr}{p\bar F\ol{(pf)^\#}} && p\barF A \ar[equals]{dd} \ar{r}{pa} & pA \ar[equals]{dd} \\
Fp\barF^{\star}P \ar[equals]{d} \ar{r}{Fpi} & Fp((pf)^\#)^{*}A \ar{drr}{Fp\ol{(pf)^\#}}  &&& \\
FF^{\star}X \ar[equals]{d} \ar{rrr}{F(pf)^\#} &&& FpA \ar{r}{pa} & pA \ar[equals]{d} \\
FF^{\star}X \ar{rr}{\ini_X} && F^{\star}X \ar{rr}{(pf)^\#} && pA \\
\end{tikzcd}
\]
Indeed, the three upper parts commute because $Fp=p\barF$, the middle left part commutes because $pi=\id$ and $\ol{(pf)^\#}$ is above $(pf)^\#$ by definition, and the lower part commutes because $(pf)^\#\colon (F^{\star}X,\ini_X)\to (pA,pa)$ is an $F$-algebra morphism. The outside of the above diagram thus shows that the composite $a\circ \bar F\ol{(pf)^\#}\circ \barF i$ is above $(pf)^\#\circ \ini_X$. Therefore, the universal property of the cartesian lift $\ol{(pf)^\#}$ implies that there exists a unique dashed morphism above $\ini_X$ making the diagram below commute:
\[ 
\begin{tikzcd}
\barF\barF^{\star}P \ar{r}{\barF i} \ar[dashed]{dr} & \barF((pf)^\#)^{*}A \ar{r}{\bar F\ol{(pf)^\#}} & \barF A \ar{r}{a} & A \ar[equals]{d} \\
& ((pf)^\#)^{*}A \ar{rr}{\ol{(pf)^\#}} && A 
\end{tikzcd}
\]
This proves $\barF \barF^{\star}P \sqleq \ini_X^{*}((pf)^\#)^{*}A$ for every $(A,a)\in \Alg(F)$ and $f\colon P\to X$, whence $\barF \barF^{\star}P \sqleq \ini_X^{*}\barF^{\star}P$
because $\ini_X^{*}$ preserves meets.

\medskip\noindent\emph{Proof of \ref{proof:c}}: Given $(A,a)\in \Alg(\barF)$ and $f\colon P\to A$, the morphism $f^\#:=\ol{(pf)^\#}\circ i$ is an $\barF$-algebra morphism from $(\barF^\star P, \dot\ini_P)$ to $(A,a)$ satisfying $f= f^\# \circ\eta_P$, that is, the diagram below commutes: 
\[
\begin{tikzcd}
\barF \barF^{\star}P \ar{r}{\ini_P} \ar{d}[swap]{\barF i} & \barF^{\star}P \ar{d}{i} & P \ar{l}[swap]{\eta_P} \ar{dd}{f} \\
\barF ((pf)^\#)^{*}A \ar{d}[swap]{F\ol{(pf)^\#}} & ((pf)^\#)^{*}A \ar{d}{\ol{(pf)^\#}}\\
\barF A \ar{r}{a} & A \ar[equals]{r} & A 
\end{tikzcd}
\]
Indeed, it suffices to show that the diagram commutes with the faithful functor $p$ applied to it, and this yields precisely the commutative diagram stating that $(pf)^\#\colon (F^{\star}X,\ini_X)\to (pA,pa)$ is an $F$-algebra morphism satisfying $pf=(pf)^\#\circ \eta_X$. 

To show uniqueness of $f^\#$, suppose that $h\colon (\barF^{\star}P,\ini_P)\to (A,a)$ is an $\barF$-algebra morphism satisfying $f=h\circ \eta_P$; that is, the diagram below commutes: 
\[
\begin{tikzcd}
\barF \barF^{\star}P \ar{r}{\ini_P} \ar{d}[swap]{\barF h} & \barF^{\star}P \ar{d}{h} & P \ar{l}[swap]{\eta_P} \ar{d}{f} \\
\barF A \ar{r}{a} & A \ar[equals]{r} & A 
\end{tikzcd}
\] 
Applying $p$ to this diagram shows that $ph\colon (F^{\star}X,\ini_X)\to (pA,pa)$ is an $F$-algebra morphism satisfying $pf=ph\circ \eta_X$. Therefore $ph=(pf)^\#$, which uniquely determines $h$ because $p$ is faithful.

\subsection*{Proof of \Cref{lem:sim-props}}
The statements of \Cref{lem:sim-props} are subsumed by the following lemma, which establishes additional properties of behavioural conformances:

\begin{lemma}\label{lem:sim-props-app}
Let $F\colon \B\to \B$ be an endofunctor with a laxly monoidal lifting $\barF\colon \E\to \B$ along the $\Qtl$-fibration $p\colon \E\to \B$, and suppose that $(C,c)$ is an $F$-coalgebra.
\begin{enumerate}
\item\label{lem:sim-props-app-identity} The unit $1_C\in \E_C$ is a $\barF$-behavioural conformance on $(C,c)$.
\item\label{lem:sim-props-app-composition} If $P,Q\in \E_C$ are $\barF$-behavioural conformances on $(C,c)$, then so is $P\cdot Q$.
\item\label{lem:sim-props-app-join} If $P_i\in \E_C$ ($i\in I$) are $\barF$-behavioural conformances on $(C,c)$, then so is $\bigsqcup_{i\in I} P_i$. 
\item\label{lem:sim-props-app-transclos} If $P\in \E_C$ is a $\barF$-behavioural conformance on $(C,c)$, then so is its transitive closure $P^{+}$.
\item\label{lem:sim-props-app-greatest} The coalgebra $(C,c)$ has a greatest $\barF$-behavioural conformance $P_c^\to$, given by the join in $\E_C$ of all $\barF$-behavioural conformances.
\item \label{lem:sim-props-app-greatest-refl-trans} The greatest $\barF$-behavioural conformance $P_c^\to$ on $(C,c)$ is reflexive and transitive.
\end{enumerate}
\end{lemma}

\begin{proof}
\begin{enumerate}
\item $1_C$ is $\barF$-behavioural conformance because
\[ 1_C \sqleq c^{*} 1_{FC} \sqleq c^{*} \barF 1_C;\]
the first step uses that $p$ is a $\Qtl$-fibration and the second one that $\barF$ is laxly monoidal.
\item Given $\barF$-behavioural conformances $P,Q$, the composite $P\cdot Q$ is a $\barF$-behavioural conformance because
\[P\cdot Q \sqleq c^{*}\barF P \cdot c^{*}\barF Q\sqleq c^{*}(\barF P\cdot \barF Q) \sqleq c^{*}\barF(P\cdot Q);
\]
the first step uses that $P$ and $Q$ are $\barF$-behavioural conformances, the second one that $p$ is a $\Qtl$-fibration, and the third one that $\barF$ is laxly monoidal.
\item If $P_i$ ($i\in I$) are $\barF$-behavioural conformances, where $I\neq \emptyset$, then $P=\bigsqcup_i P_i$ is a $\barF$-behavioural conformance because 
\[ c_{*} P = c_{*}\bigsqcup_{i\in I} P_i = \bigsqcup_{i\in I} c_{*} P_i \sqleq \bigsqcup_{i\in I} \barF P_i \sqleq \bigsqcup_{i\in I} \barF P \sqleq \barF P.  \]
The second step uses that $p$ is a $\CLat_\sqcap$-fibration (so $c_{*}$ preserves all joins), the third one that each $P_i$ is a $\barF$-behavioural conformance, and the fourth one that $P_i\sqleq P$ and $\barF$ is a functor.
\item If $P$ is an $\barF$-behavioural conformance, then $P^{+}=\bigsqcup_{n\geq 0} P^n$ is a $\barF$-behavioural conformance by parts \ref{lem:sim-props-app-composition} and \ref{lem:sim-props-app-join}.
\item Immediate from part \ref{lem:sim-props-app-join}.
\item By parts \ref{lem:sim-props-app-identity} and \ref{lem:sim-props-app-composition}, we know that $1_C$ and $P_c^\to\cdot P_c^\to$ are $\barF$-behavioural conformances on $(C,c)$. Since $P_c^\to$ is the greatest $\barF$-behavioural conformance, it follows that $1_C\sqleq P_c^\to$ and $P_c^\to\cdot P_c^\to\sqleq P_c^\to$, so $P_c^\to$ is reflexive and transitive.\qedhere
\end{enumerate}
\end{proof}

\subsection*{Proof of \Cref{lem:bisim-props}}
We need a simple auxiliary result:
\begin{lemma}\label{lem:trans-clos}
Let $M$ be a quantale and $x\in M$.
\begin{enumerate}
\item\label{lem:trans-clos-clos} The element $x^{+}$ is the least transitive element of $M$ such that $x\sqleq x^{+}$. 
\item\label{lem:trans-clos-inv} If $M$ is involutive, then $(x^\circ)^{+}=(x^{+})^\circ$.
\end{enumerate}
\end{lemma}

\begin{proof}
\begin{enumerate}
\item The element $x^{+}$ is transitive because
\[ x^{+}\cdot x^{+} = (\bigsqcup_{n\geq 1} x^n)\cdot (\bigsqcup_{m\geq 1} x^m) = \bigsqcup_{n,m\geq 1} x^{n+m} \sqleq \bigsqcup_{n\geq 1} x^n = x^{+}.\]
Moreover, if $x\sqleq y$ and $y$ is transitive, then $x^n\sqleq y^n\sqleq y$ for all $n\geq 1$, hence $x^{+}\sqleq y$.
\item This follows from the computation
\[
(x^\circ)^{+} = \bigsqcup_{n\geq 1} (x^\circ)^n = \bigsqcup_{n\geq 1} (x^n)^\circ = (\bigsqcup_{n\geq 1} x^n)^\circ = (x^{+})^\circ, 
\]  
where the second and the third step use that $M$ is an involutive quantale.
\qedhere
\end{enumerate}
\end{proof}

The statement of \Cref{lem:bisim-props} is subsumed by the following lemma, which establishes additional properties of behavioural biconformances:

\begin{lemma}\label{lem:bisim-props-app}
Let $F\colon \B\to \B$ be an endofunctor with a laxly monoidal lifting $\barF\colon \E\to \B$ along the $\InQtl$-fibration $p\colon \E\to \B$, and suppose that $(C,c)$ is an $F$-coalgebra.
\begin{enumerate}
\item\label{lem:bisim-props-app-identity} The unit $1_C\in \E_C$ is a $\barF$-behavioural biconformance on $(C,c)$.
\item\label{lem:bisim-props-app-composition} If $P,Q\in \E_C$ are $\barF$-behavioural biconformances on $(C,c)$, then so is $P\cdot Q$.
\item\label{lem:bisim-props-app-converse} If $P\in \E_C$ is a $\barF$-behavioural biconformance on $(C,c)$, then so is $P^\circ$.
\item\label{lem:bisim-props-app-join} If $P_i\in \E_C$ ($i\in I$) are $\barF$-behavioural biconformances in $(C,c)$, then so is $\bigsqcup_{i\in I} P_i$. 
\item\label{lem:bisim-props-app-transclos} If $P\in \E_C$ is a $\barF$-behavioural biconformance on $(C,c)$, then so is its transitive closure $P^{+}$.
\item\label{lem:bisim-props-app-greatest} The coalgebra $(C,c)$ has a greatest $\barF$-behavioural biconformance $P_c^\leftrightarrow$, given by the join in $\E_C$ of all $\barF$-behavioural biconformances.
\item \label{lem:bisim-props-app-greatest-refl-trans} The greatest $\barF$-behavioural biconformance $P^\leftrightarrow_c$ on $(C,c)$ is reflexive, symmetric, and transitive.
\end{enumerate}
\end{lemma}

\begin{proof}
\begin{enumerate}
\item This follows from \Cref{lem:sim-props-app}\ref{lem:sim-props-app-identity} since $1_C^\circ = 1_C$.
\item Let $P,Q$ be behavioural biconformances. Then $P,P^\circ,Q,Q^\circ$ are behavioural conformances, so $P\cdot Q$ and $(P\cdot Q)^\circ = Q^\circ \cdot P^\circ$ are behavioural conformances by \Cref{lem:sim-props-app}\ref{lem:sim-props-app-composition}. This shows that $P\cdot Q$ is a behavioural biconformance.
\item Immediate from the definition of a behavioural biconformance, since $(P^{\circ})^\circ = P$.
\item Let $P_i$ ($i\in I$) be behavioural biconformances, where $I\neq\emptyset$. Then $P_i$ and $P_i^\circ$ are behavioural conformances, so $\bigsqcup_i P_i$ and $(\bigsqcup_i P_i)^\circ = \bigsqcup_i P_i^\circ$ are behavioural conformances by \Cref{lem:sim-props-app}\ref{lem:sim-props-app-join}, so $\bigsqcup_i P_i$ is a behavioural biconformance. 
\item Let $P$ be a behavioural biconformance. Then $P$ and $P^\circ$ are behavioural conformances, so $P^{+}$ and $(P^\circ)^{+}$ are behavioural conformances by \Cref{lem:sim-props-app}\ref{lem:sim-props-app-transclos}. Since
$(P^{+})^\circ = (P^\circ)^{+}$ by \Cref{lem:trans-clos}\ref{lem:trans-clos-inv}, we conclude that $P^{+}$ is a behavioural biconformance.
\item Immediate from part \ref{lem:sim-props-app-join}.
\item By parts \ref{lem:bisim-props-app-identity}, \ref{lem:bisim-props-app-composition}, \ref{lem:bisim-props-app-converse} we know that $1_C$, $P_c^\leftrightarrow\cdot P_c^\leftrightarrow$, and $(P_c^\leftrightarrow)^\circ$ are behavioural biconformances on $(C,c)$. Since $P_c^\leftrightarrow$ is the greatest behavioural biconformance, it follows that \[1_C\sqleq P_c^\leftrightarrow,\qquad P_c^\leftrightarrow\cdot P_c^\leftrightarrow\sqleq P_c^\leftrightarrow, \qqand (P_c^\leftrightarrow)^\circ\sqleq P_c^\leftrightarrow,\]
which proves that $P_c^\leftrightarrow$ is reflexive, symmetric, and transitive.\qedhere
\end{enumerate}
\end{proof}

\subsection*{Proof of \Cref{lem:cong-props}}

\begin{rem}
Recall that a $\Qtl$-fibration $p\colon \E\to \B$ satisfies
\[ f^{*} P \cdot f^{*} Q \sqleq f^{*} (P\cdot Q) \qquad\text{for all $f\colon X\to Y$ in $\B$}.  \] 
Using the adjunction $f_{*} \dashv f^{*}$, one readily verifies that this implies a corresponding property of $f_{*}$:
\[ f_{*} (P \cdot Q) \sqleq f_{*} P \cdot f_{*} Q \qquad\text{for all $f\colon X\to Y$ in $\B$}. \] 
\end{rem}
We now prove the two statements of \Cref{lem:cong-props}.

\begin{enumerate}
\item This follows from the computation
\[ a_{*}\barSigma (P\cdot Q) = a_{*}(\barSigma P\cdot \barSigma Q) \sqleq a_{*}\barSigma P \cdot a_{*}\barSigma Q \sqleq P\cdot Q.  \]
The first step uses that $\barSigma$ is monoidal, the second one that $p$ is $\Qtl$-fibration, and the last one that $P$ and $Q$ are congruences.
\item This follows from the computation
\begin{align*}
a_{*}\barSigma P^{+} & = a_{*}\barSigma (\bigsqcup_{n\geq 1}P^{n}) & \text{def.\ $(-)^{+}$} \\
& = a_{*}\bigsqcup_{n\geq 1}\barSigma P^{n} & \text{$\barSigma$ pres.\ dir.\ joins} \\
& = a_{*}\bigsqcup_{n\geq 1}(\barSigma P)^{n} & \text{$\barSigma$ monoidal} \\ 
& = \bigsqcup_{n\geq 1}a_{*} (\barSigma P)^{n} & \text{\Cref{lem:clat-fib-props}\ref{lem:clat-fib-props-adjunction}} \\ 
& \sqleq \bigsqcup_{n\geq 1}(a_{*} \barSigma P)^{n} & \text{$p$ is a $\Qtl$-fibration} \\
& \sqleq \bigsqcup_{n\geq 1}P^{n} & \text{$P$ congruence} \\
& = P^{+} & \text{def.\ $(-)^{+}$} \\
\end{align*}
For the second step, note that the join $\bigsqcup_{n\geq 1} P^n$ is indeed directed: Since $P$ is reflexive by assumption, 
\[P = P \cdot 1 \sqleq P \cdot P =  P \cdot P\cdot 1\sqleq  P \cdot P\cdot P = \cdots, \]
hence the objects $P^n$ ($n\geq 1$) form an ascending chain in $\E_A$.\qedhere
\end{enumerate}

\subsection*{Proof of \Cref{lem:howe-closure-props}}
To prepare the proof we need a simple lemma:

\begin{lemma}\label{lem:cong-props-reverse} Let $\Sigma\colon \B\to\B$ be an endofunctor with a $\circ$-monoidal lifting $\barSigma\colon \E\to\E$ along the $\InQtl$-fibration $p\colon \E\to \B$. If $P$ is a $\barSigma$-congruence on a $\Sigma$-algebra $(A,a)$, then so is $P^\circ$.
\end{lemma}

\begin{proof}
This follows from the computation
\[ a_{*}\barSigma P^\circ \sqleq a_{*} (\barSigma P)^\circ \sqleq (a_{*} \barSigma P)^\circ \sqleq P^\circ. \]
The first step uses that $\barSigma$ is $\circ$-monoidal, the second one that $p$ is an $\InQtl$-fibration, and third one that $P$ is a congruence.
\end{proof}

We now prove the three statements of \Cref{lem:howe-closure-props}.
\begin{enumerate}
\item Suppose that $P$ is reflexive. Then $\hat P$ is reflexive because
\[ 1 \sqleq P \sqleq P\sqcup (a_{*} \barSigma\hatP)\cdot P = \hatP;\]
the first step uses that $P$ is reflexive and the third one holds by definition of $\hatP$. To prove that $\hat P$ is a congruence on $(A,a)$, we compute
\[ a_{*} \barSigma\hatP = (a_{*} \barSigma\hatP)\cdot 1 \sqleq (a_{*} \barSigma\hatP)\cdot P \sqleq P\sqcup (a_{*} \barSigma\hatP)\cdot P = \hatP; \]
the first step uses that $\E_A$ is a monoid, and the second one that $P$ is reflexive. 
\item Suppose that $P$ is transitive (i.e.\ $P\cdot P \sqleq P$). Then
\[ \hat P \cdot P = (P\sqcup (a_{*} \barSigma\hatP)\cdot P)\cdot P = P\cdot P\sqcup (a_{*} \barSigma\hatP)\cdot P\cdot P \sqleq P\sqcup (a_{*} \barSigma\hatP)\cdot P = \hatP; \]
the first step uses the definition of $\hatP$, the second one that $\E_A$ is a quantale, the third one that $P$ is transitive, and the last one the definition of $\hatP$ again.
\item Suppose that $P$ is reflexive and symmetric. We only need to prove
\begin{equation}\label{eq:howes-trick-proof-goal} 
\hatP \sqleq (\hatP^{\circ})^{+},
\end{equation}
as this immediately implies that $\hatP^{+}$ is symmetric: we have
\[\hatP^{+}\sqleq (\hatP^\circ)^{++} = (\hatP^\circ)^{+} = (\hatP^{+})^\circ,\]
where the last step follows by \Cref{lem:trans-clos}\ref{lem:trans-clos-inv}. To prove \eqref{eq:howes-trick-proof-goal}, recall that $\hatP$ is the least fixed point (equivalently, the least pre-fixed point) of the map $Q\mapsto P\sqcup (a_{*}\bar\Sigma Q)\cdot P$ on $\E_A$. Therefore, it suffices to prove that $(\hatP^\circ)^{+}$ is a pre-fixed point, i.e.,
\[ P\sqcup (a_{*}\bar\Sigma (\hatP^\circ)^{+})\cdot P \sqleq (\hatP^\circ)^{+}, \]
or equivalently
\begin{align}
P &\sqleq (\hatP^\circ)^{+} \label{eq:howes-trick-proof-goal-2}\\
(a_{*}\bar\Sigma (\hatP^\circ)^{+})\cdot P & \sqleq (\hatP^\circ)^{+}. \label{eq:howes-trick-proof-goal-3}
\end{align}
The inequality \eqref{eq:howes-trick-proof-goal-2} follows by the computation
\[ P = P^\circ \sqleq \hatP^\circ \sqleq (\hatP^\circ)^{+}, \]
using that $P$ is symmetric in the first step, that $P\sqleq \hatP$ by definition of $\hatP$ in the second step, and that $(-)^{+}$ is the transitive closure in the third step. 

To prove \eqref{eq:howes-trick-proof-goal-3}, we compute
\[ (a_{*}\barSigma (\hatP^\circ)^{+})\cdot P \sqleq (\hatP^\circ)^{+} \cdot P \sqleq (\hatP^\circ)^{+} \cdot (\hatP^\circ)^{+} \sqleq (\hatP^\circ)^{+}. \] 
For the first step, we use that $\hatP$ is congruence on $(A,a)$ by part \ref{lem:howe-closure-props-reflexive}, so 
 $\hatP^\circ$ is congruence by \Cref{lem:cong-props-reverse}, so $(P^\circ)^{+}$ is a congruence by \Cref{lem:cong-props}\ref{lem:cong-props-transhull}. The second step follows from \eqref{eq:howes-trick-proof-goal-2}, and the third one from the definition of $(-)^{+}$.
\qedhere
\end{enumerate}

\subsection*{Proof of \Cref{prop:howe-closure-app-struc}}
\begin{notation}
\begin{enumerate}
\item 
For notational distinction, we write $\barSigma_0\colon \E_0\to \E_0$ and $\barB_0\colon \E_0^\opp\times \E_0\to \E_0$ for the liftings of $\Sigma$ and $B$ along $p_0$, that respectively, are restrictions of the liftings $\barSigma$ and $\barB$ of $\Sigma^{(2)}$ and $B^{(2)}$ along $p$ (see \Cref{not:restricted-liftings}). 
\item The liftability condition on $\S$ and $p$, which says that
\[ \rho\colon \barSigma_0 \barB_0^\infty(\Lambda,\Lambda)\xto{\thickcdot} \barB_0(\Lambda, \barSigma_0^{\star}\Lambda),  \]
is equivalent to
\[ (\rho,\rho)\colon \barSigma \barB^\infty(\Lambda,\Lambda)\xto{\thickcdot} \barB(\Lambda,\barSigma^{\star}\Lambda).  \]
\item We write $(f,g)^{*}$, $(f,g)_{*}$ for (op)reindexing w.r.t.~the lifting $p\colon \E\to \B^{(2)}$, and $f^{*}$, $f_{*}$ for (op)reindexing w.r.t.~$p_0\colon \E_0\to \B$. Note that $(f,f)^{*}=(f^{*},f^{*})$ and $(f,f)_{*}=(f_{*},f_{*})$for all $f$. 
\end{enumerate}
\end{notation}

Recall that the canonical model $\gamma$ of $\S$ is the join in $\HOCoalg_\Lambda(B)$ of the ordinal-indexed chain $(\gamma_\alpha)_\alpha$ given in \Cref{rem:complete-ahos}:
\[ \gamma = \bigvee_\alpha \gamma_\alpha. \]
To show that $\hatP$ is a $\barB_0(1_\Lambda,-)$-behavioural conformance on $(\Lambda,\gamma)$, we need to prove that
\[  \hatP \sqleq \gamma^{*}\barB_0(1_\Lambda,\hatP). \]
For this purpose, it suffices to prove that
\begin{equation}\label{eq:proof-goal-cong-app} \hatP \sqleq (\gamma_\alpha,\gamma)^{*}\barB(\hatP,\hatP) \qquad\text{for all ordinals $\alpha$}. \end{equation}
This statement immediately implies that $\hatP$ is a $\barB$-behavioural conformance on $(\Lambda,(\gamma,\gamma))$: we have
\[ \hatP\sqleq \bigsqcap_\alpha (\gamma_\alpha,\gamma)^{*}\barB(\hatP,\hatP) = (\gamma,\gamma)^{*}\barB(\hatP,\hatP) \sqleq (\gamma,\gamma)^{*}\barB(1_\Lambda,\hatP) = \gamma^{*} \barB(1_\Lambda,\hatP), \]
where the first step follows from \eqref{eq:proof-goal-cong-app}, the second one from our assumption that $\S$ is compatible with $\L$, and the third one from $\hatP$ being reflexive (\Cref{lem:howe-closure-props}\ref{lem:howe-closure-props-reflexive}) and $\barB$ being contravariant in the first component.

To prove \eqref{eq:proof-goal-cong-app}, we proceed by transfinite induction. The limit step is straightforward: if $\alpha$ is a limit ordinal, we have
\[ \hatP \sqleq \bigsqcap_{\beta<\alpha}(\gamma_{\beta},\gamma)^{*}\barB(\hatP,\hatP) = (\gamma_{\alpha},\gamma)^{*}\barB(\hatP,\hatP),\]
where the first step uses that $\hatP\sqleq (\gamma_{\beta},\gamma)^{*}\barB(\hatP,\hatP)$ for all $\beta<\alpha$ by induction, and the second one follows from out assumption $\S$ is compatible with $\L$, since $\gamma_\alpha=\bigvee_{\beta<\alpha} \gamma_\beta$.

For the successor step, suppose that \eqref{eq:proof-goal-cong-app} holds for some $\alpha$. We will prove the following five inequalities (the first three are auxiliary statements used for proving the last two):
\begin{align}
\barSigma \hatP &\sqleq (\ini,\ini)^{*}(\gamma_{\alpha+1},\gamma)^{*} \barB(\hatP,\hatP) \label{eq:howe-aux-1}\\
((\gamma_{\alpha+1},\gamma)^{*}\barB(\hatP,\hatP))\cdot P &\sqleq (\gamma_{\alpha+1},\gamma)^{*}\barB(\hatP,\hatP) \label{eq:howe-aux-2}\\
1_{\Lambda} & \sqleq (\gamma_{\alpha+1},\gamma)^{*} \barB(\hatP,\hatP) \label{eq:howe-3} \\
P &\sqleq (\gamma_{\alpha+1},\gamma)^{*}\barB(\hatP,\hatP) \label{eq:howe-1} \\
((\ini,\ini)_{*}\barSigma\hatP)\cdot P &\sqleq (\gamma_{\alpha+1},\gamma)^{*}\barB(\hatP,\hatP) \label{eq:howe-2} 
\end{align}
From
 \eqref{eq:howe-1} and \eqref{eq:howe-2} it then follows that 
\[\hatP = P\sqcup ((\ini,\ini)_{*}\barSigma \hat{P})\cdot P \sqleq (\gamma_{\alpha+1},\gamma)^{*}\barB(\hatP,\hatP),\] 
as required. It thus only remains to establish the above inequalities \eqref{eq:howe-aux-1}--\eqref{eq:howe-2}.

\medskip\noindent\textbf{Proof of \eqref{eq:howe-aux-1}:}
This follows from the computation
\begin{align*}
\bar\Sigma\hatP & \sqleq (\Sigma^{(2)}(\hat\gamma_\alpha,\hat\gamma))^{*} \barSigma\barB^\infty(\hatP,\hatP) & \text{see below} \\
& \sqleq (\Sigma^{(2)}(\hat\gamma_\alpha,\hat\gamma))^{*} (\rho,\rho)^{*} \barB(\hatP,\barSigma^{*} \hatP) & \text{$\S$ liftable w.r.t.\ $\L$} \\
& \sqleq (\Sigma^{(2)}(\hat\gamma_\alpha,\hat\gamma))^{*} (\rho,\rho)^{*} (B^{(2)}((\id,\id),(\hat\ini,\hat\ini)))^{*} \barB(\hatP,\hatP) & \text{\Cref{prop:initial-alg-cofree-coalg-lift}} \\
& = (\ini,\ini)^{*} (\gamma_{\alpha+1},\gamma)^{*} \barB(\hatP,\hatP) & \text{def. $\gamma_{\alpha+1}$}.
\end{align*}
For the first step, we argue as follows. Form the $\barB(\hatP,-)$ coalgebra
\[ k = (\, \hatP \xto{i} (\gamma_\alpha,\gamma)^{*}\barB(\hatP,\hatP) \xto{\ol{(\gamma_\alpha,\gamma)}} \barB(\hatP,\hatP)\,)   \]
where $i$ is the morphism in $\E_\Lambda$ witnessing the induction hypothesis, and $\ol{(\gamma_a,\gamma)}$ is the cartesian lift of $(\gamma_a,\gamma)\colon \Lambda\to B(\Lambda,\Lambda)=p\barB(\hatP,\hatP)$. Since $k$ is above $(\gamma_\alpha,\gamma)$, the cofree extension $\hat k\colon \hatP \to \barB^\infty(\hatP,\hatP)$ of $\id\colon \hatP\to \hatP$ is above the cofree extension $(\hat\gamma_\alpha,\hat\gamma)\colon \Lambda\to B^\infty(\Lambda,\Lambda)$ of $(\id,\id)\colon \Lambda\to\Lambda$ by the dual of \Cref{prop:initial-alg-cofree-coalg-lift}. Thus $\barSigma \hat k$ is above $\Sigma^{(2)}(\hat\gamma_\alpha,\hat\gamma)$, which proves  $\barSigma\hat P\sqleq (\Sigma^{(2)}(\hat\gamma_\alpha,\hat\gamma))^{*} \barSigma\barB^\infty(\hatP,\hatP)$.

\medskip\noindent\textbf{Proof of \eqref{eq:howe-aux-2}:}
This follows from the computation
\begin{align*}
((\gamma_{\alpha+1},\gamma)^{*}\barB(\hatP,\hatP))\cdot P &\sqleq (\gamma_{\alpha+1},\gamma)^{*}\barB(\hatP,\hatP)\cdot (\gamma,\gamma)^{*}\barB(1_\Lambda,P) & \text{$P$ a $\barB$-behavioural conformance} \\
&\sqleq (\gamma_{\alpha+1},\gamma)^{*}(\barB(\hatP,\hatP)\cdot \barB(1_\Lambda,P)) & \text{$p$ heterogeneous fib.} \\ 
&\sqleq (\gamma_{\alpha+1},\gamma)^{*}\barB(\hatP,\hatP\cdot P) & \text{$\barB$ laxly monoidal} \\
 &\sqleq (\gamma_{\alpha+1},\gamma)^{*}\barB(\hatP,\hatP) & \text{\Cref{lem:howe-closure-props}\ref{lem:howe-closure-props-weakly-transitive}}. \\
\end{align*}

\medskip\noindent\textbf{Proof of \eqref{eq:howe-3}:}
Since the structure $\ini\colon \Sigma\Lambda\to \Lambda$ of the initial $\Sigma$-algebra is an isomorphism, we have
\begin{align*}
1_\Lambda & \sqleq (\ini^{-1},\ini^{-1})^{*} 1_{\Sigma\Lambda} & \text{$p$ het.\ $\Qtl$-fibration} \\
& \sqleq (\ini^{-1},\ini^{-1})^{*} \barSigma 1_\Lambda & \text{$\barSigma$ laxly monoidal} \\
& \sqleq (\ini^{-1},\ini^{-1})^{*} \barSigma \hatP & \text{\Cref{lem:howe-closure-props}\ref{lem:howe-closure-props-reflexive}} \\
& \sqleq (\ini^{-1},\ini^{-1})^{*} (\ini,\ini)^{*}(\gamma_{\alpha+1},\gamma)^{*} \barB(\hatP,\hatP) & \text{by \eqref{eq:howe-aux-1}} \\
&= (\gamma_{\alpha+1},\gamma)^{*} \barB(\hatP,\hatP). &
\end{align*}

\medskip\noindent\textbf{Proof of \eqref{eq:howe-1}:}
This follows from the computation
\begin{align*}
P &= 1_\Lambda \cdot P  & \\
&\sqleq (\gamma_{\alpha+1},\gamma)^{*}\barB(\hatP,\hatP)\cdot P & \text{by \eqref{eq:howe-3}} \\
 &\sqleq (\gamma_{\alpha+1},\gamma)^{*}\barB(\hatP,\hatP) & \text{by \eqref{eq:howe-aux-2}}. 
\end{align*}

\medskip\noindent\textbf{Proof of \eqref{eq:howe-2}:}
This follows from the computation
\begin{align*}
((\ini,\ini)_{*}\bar\Sigma\hatP)\cdot P & \sqleq (\ini,\ini)_{*} (\ini,\ini)^{*}(\gamma_\alpha,\gamma)^{*}\barB(\hatP,\hatP)\cdot P & \text{by \eqref{eq:howe-aux-1}} \\
 & \sqleq (\gamma_\alpha,\gamma)^{*}\barB(\hatP,\hatP)\cdot P & \text{since $(-)_{*} \dashv (-)^{*}$} \\
 &\sqleq (\gamma_\alpha,\gamma)^{*}\barB(\hatP,\hatP) & \text{by \eqref{eq:howe-aux-2}}. 
\end{align*} 
This concludes the proof of \Cref{prop:howe-closure-app-struc}.

\subsection*{Proof of \Cref{lem:lifting-compatible}}
For the proof of compatibility of $\L$, we need a continuity property of the Wasserstein distance. Recall that  $\barD(\{\bot\}+X)$ carries a DCPO structure given by $\phi\leq \phi'$ iff $\phi(x)\leq \phi'(x)$ for all $x\in X$.

\begin{lemma}\label{lem:wasserstein-cont}
Let $(X,d)$ be a fuzzy relation and let $\bard = \barD(d_\bot)$ be the induced Wasserstein metric on $\D(\{\bot\}+X)$. Then for every directed family $\phi_k\in \D(\{\bot\}+X)$, $k\in K$, and every $\psi\in \D(\{\bot\}+X)$,
\[ \bard(\bigvee_{k\in K} \phi_k, \psi) = \bigvee_{k\in K} \bard(\phi_k,\psi)  \]
\end{lemma}

\begin{proof}
Let $\phi=\bigvee_{k\in K} \phi_k$, and put  $\phi= p\cdot \bot + \sum_{i\in I} p_i\cdot x_i$ and $\psi=q\cdot \bot + \sum_{j\in J} q_j\cdot y_j$. Recall that $\bard(\phi,\psi)$ is the optimal value of the transportation problem from $\phi$ to $\psi$, that is, of the linear program $\TP(\phi,\psi)$ shown below:
\begin{equation*}
\begin{array}{ll@{}ll}
\text{minimize}  & \displaystyle\sum\limits_{i,j}& \,d_\bot(x_i,y_j) \cdot t_{i,j}&\\
\text{subject to}& \displaystyle\sum\limits_{j}&  \, t_{i,j} = p_i  & (i\in I)\\
& \displaystyle\sum\limits_{i}&   \,t_{i,j} = q_j  & (j\in J) \\
&& \,t_{i,j}\geq 0 & (i\in I,\, j\in J)
\end{array}
\end{equation*}
Its dual linear program $\TP_{\mathsf{d}}(\phi,\psi)$ is given by
\begin{equation*}
\begin{array}{ll@{}ll}
\text{maximize}  & \displaystyle\sum\limits_{i}  &  a_i\cdot p_i +  \displaystyle\sum\limits_{j} b_j\cdot q_j &\\
\text{subject to}& & a_i+b_j \leq d_\bot(x_i,y_j) & (i\in I,\, j\in J)\\
&& a_i,b_j\in \mathbb{R} & (i\in I,\, j\in J) \\
&& (a_i)_{i\in I}, (b_j)_{j\in J} \text{ bounded} & 
\end{array}
\end{equation*}
Both linear programs $\TP(\phi,\psi)$ and $\TP_{\mathsf{d}}(\phi,\psi)$ have an optimal solution, and their optimal values coincide~\cite[Thm.~2.1, Thm.~2.2]{kortanek_yamasaki95}. Take an optimal solution $(a_i)_{i\in I}$, $(b_j)_{j\in J}$ of $\TP_{\mathsf{d}}(\phi,\psi)$. Since the solution is bounded, there exists $a> 0$ such that $|a_i|\leq a$ for all $i$. For every $\phi'=\sum_{i\in I} p_i'\cdot x_i$, we note that $(a_i)$, $(b_j)$ is also a (not necessarily optimal) solution of $\TP_{\mathsf{d}}(\phi',\psi)$, and the difference between the values of $\TP_{\mathsf{d}}(\phi,\psi)$ and $\TP_{\mathsf{d}}(\phi',\psi)$ under these solutions is bounded:
\begin{equation}\label{eq:diff} |\sum_i a_i\cdot p_i - \sum_i a_i\cdot p_i'| \leq a\cdot \sum_i |p_i-p_i'|. \end{equation}
To prove the statement of the lemma, suppose towards a contradiction that
\[ C:=\bigvee_{k\in K} \bard(\phi_k,\psi) < \bard(\phi, \psi). \] 
Since the join $\phi=\bigvee_k \phi_k$ is directed, there exists $k\in K$ such that \[\sum_i |p_i-p_i'|<(\bard(\phi,\psi)-C)/a,\qquad \text{where}\qquad  \phi_k = \sum_{i} p_i'\cdot x_i.\]
By \eqref{eq:diff} this implies 
\[ |\sum_i a_i\cdot p_i - \sum_i a_i\cdot p_i'| < \bard(\phi,\psi)-C. \]
In other words, $(a_i)$, $(b_j)$ is a solution of $\TP_{\mathsf{d}}(\phi_k,\phi)$ whose value is greater than $C$. By duality, this means that the optimal solution of $\TP(\phi_k,\psi)$ has a value greater than $C$, that is, $\bard(\phi_k,\psi)>C$, a contradiction.
\end{proof}

We split the statement of \Cref{lem:lifting-compatible} into four sublemmas.

\begin{lemma}
The lifting situation $\L=(p,\barSigma,\barB)$ is compatible with $\S_\pBCK$.
\end{lemma}

\begin{proof}
Let $f_i\colon \Lambda\to B(\Lambda,\Lambda)$ ($i\in I$) be a $\leq$-directed family, let $g\colon \Lambda\to B(\Lambda,\Lambda)$, and let $d,e\colon \Lambda\times\Lambda\to [0,1]$ be fuzzy relations. We need to prove
 \begin{equation}\label{eq:proof-goal-comp}
(\bigvee_{i\in I} f_i, g)^{*}\barB(d,e) = \bigsqcap_{i\in I}\, (f_i,g)^{*}\barB(d,e).\end{equation}
To follows from the computation below for all $s,t\in \Lambda$:
\begin{align*}
((\bigvee_{i\in I} f_i, g)^{*}\barB(d,e))(s,t) &= \barB(d,e)(\bigvee_i f_i(s),g(t)) & \\
&= \bigvee_i \barB(d,e)(f_i(s),g(t)) & \text{by \Cref{lem:wasserstein-cont}} \\
&= \bigvee_i ((f_i,g)^{*}\barB(d,e))(s,t). \qedhere
\end{align*}
\end{proof}

\begin{lemma}
The lifting $\barSigma$ is $\circ$-monoidal.
\end{lemma}

\begin{proof}
We need to prove $\barSigma(d\cdot e)=\barSigma d \cdot \barSigma e$ for all fuzzy relations $d,e\colon X\times X\to [0,1]$. It suffice to show $\barSigma(d\cdot e)(s,t)=(\barSigma d \cdot \barSigma e)(s,t)$ for elements $s,t\in \Sigma X$ of the same shape, for otherwise both sides are $1$ and the equality holds trivially. For all $n$-ary $\f\in \Sigma$ and $x_i,y_i\in X$ ($i=1,\ldots,n$), we have
\begin{align*}
&\barSigma(d\cdot e)(\f(x_1,\ldots,x_n),\f(y_1,\ldots,y_n)) \\
&= \sum_{i=1}^n (d\cdot e)(x_i,y_i) \\
&= \sum_{i=1}^n \inf_{z} \{ d(x_i,z) + e(z,y_i) \} \\
&=\inf_{(z_1,\ldots,z_n)} \{ \sum_{i=1}^n (d(x_i,z_i) + e(z_i,y_i)) \} \\
&=\inf_{(z_1,\ldots,z_n)} \{ \sum_{i=1}^n d(x_i,z_i) + \sum_{i=1}^n e(z_i,y_i) \} \\
&=\inf_{(z_1,\ldots,z_n)} \{ \barSigma d(\f(x_1,\ldots,x_n), \f(z_1,\ldots,z_n)) + \barSigma e(\f(z_1,\ldots,z_n),\f(y_1,\ldots,y_n)) \} \\
&=(\barSigma d\cdot \barSigma e)(\f(x_1,\ldots,x_n),\f(y_1,\ldots,y_n)). \qedhere
\end{align*}
\end{proof}

\begin{lemma}
The map $\barSigma\colon \FRel_{\jne,X}\to \FRel_{\jne,\Sigma X}$ preserves directed joins.
\end{lemma}

\begin{proof}
Let $d_i\colon X\times X\to [0,1]$ be a $\sqleq$-directed family of fuzzy relations with join $d$. We need to prove that 
$\Sigma d  = \bigsqcup_i \barSigma d_i$. It suffices to show $\Sigma d(s,t)  = (\bigsqcup_i \barSigma d_i)(s,t)$ for term $s,t\in \Sigma X$ of the same shape, for otherwise both sides are $1$. For all $n$-ary $\f\in \Sigma$ and $x_i,y_i\in X$ ($i=1,\ldots,n$), we have
 \begin{align*}
& \barSigma d(\f(x_1,\ldots,x_n),\f(y_1,\ldots,y_n)) \\
&=\sum_{k=1}^n d(x_k,y_k) \\
&=\sum_{k=1}^n \bigwedge_i d_i(x_k,y_k) \\
&=\bigwedge_i \sum_{k=1}^n d_i(x_k,y_k) \\
&= (\bigsqcup_i \barSigma d)(\f(x_1,\ldots,x_n),\f(y_1,\ldots,y_n))
\end{align*}
The third step uses that the meet is directed and the sum is finite.
\end{proof}

\begin{lemma}
The lifting $\barB$ is laxly monoidal.
\end{lemma}

\begin{proof}
$\barB$ is composed of $(-)_\bot$, $\barH$, and $\barD$, so it suffices to show that all three lifted functors are laxly monoidal. 

\begin{enumerate}
\item For $(-)_\bot$, we need to show for $d,e\colon X\times X\to [0,1]$ and $x,x'\in \{\bot\}+X$ that
\[ (d\cdot e)_\bot(x,x') \leq (d_\bot\cdot e_\bot)(x,x').\]
For $x=\bot$ this holds because the left-hand side is $0$. For $x'=\bot$ it holds because either both sides are $0$ or the right-hand side is $1$. For $x,x'\in X$ it also holds because both sides are equal to $(d\cdot e)(x,y)$.
\item For $H$, we need to prove for $d,d'\colon X\to [0,1]$, $e,e'\colon Y\times Y\to [0,1]$ and $f,g\in Y^X$ that
\begin{equation}\label{eq:H-lax-mon} (e\cdot e')^{d\cdot d'}(f,g) \leq (e^d\cdot e'^{d'})(f,g).  \end{equation}
For all $x,x',x''\in X$ and $h\in Y^X$ we have
\begin{align*}
(e\cdot e')(f(x),g(x'')) &\leq e(f(x),h(x'))+e'(h(x'),g(x'')) \\
& \leq d(x,x')+e^d(f,h) + d'(x',x'') + e'^{d'}(h,g) \\
& = d(x,x')+d'(x',x'')+e^d(f,h)+e'^{d'}(h,g).
\end{align*}
Taking the infimum over all $h\in Y^X$ and $x'\in X$, thus proves
\[ (e\cdot e')(f(x),g(x'')) \leq (d\cdot d')(x,x'') + (e^d\cdot e'^{d'})(f,g). \]
Since this inequality holds for all $x,x''\in X$, we conclude \eqref{eq:H-lax-mon}.

\item For $\D$, let $d,e\colon X\times X\to [0,1]$. For $\phi,\psi\in \D X$ we need to prove that
\[ \barD(d\cdot e)(\phi,\psi) \leq (\barD d\cdot \barD e)(\phi,\psi). \]
For this it suffices to prove that for every $\xi \in \D X$, every transportation plan from $\phi$ to $\xi$ along $d$ of cost $C$, and every transportation plan from $\xi$ to $\psi$ along $e$ of cost $C'$, there is a transportation plan from $\phi$ to $\psi$ along $d\cdot e$ of cost at most $C+C'$. But this is clear: Just compose the two transportation plans in the obvious way.
\end{enumerate}
\end{proof}

\subsection*{Proof of \Cref{prop:ctx-greatest}}
\begin{enumerate}
\item $d_\ctx^\pBCK$ is adequate: apply its definition to the trivial context $C=[\cdot]$.
\item $d_\ctx^\pBCK$ is a $\barSigma$-congruence: Let $\f$ in $\Sigma$ and let $C[\cdot]$ be a context of input type, and let $t_i,s_i\in \Lambda$ for $i=\ldots,n$. We define 
\[ C_i[\cdot] := C[f(s_1,\cdots,s_{i-1},\cdot\,, t_{i+1},\cdots,t_n)] \qquad (i=1,\ldots,n). \]
Note that
\[
C_1[t_1]=C[f(t_1,\ldots,t_n)], \quad C_n[s_n]=C[f(s_1,\ldots,s_n)],\quad C_i[s_i]=C_{i+1}[t_{i+1}]\;\; (i=1,\ldots,n-1). 
\]
It follows that
\begin{align*}
&~|\gamma(C[\f(t_1,\ldots,t_n)])(\bot)-\gamma(C[s_1,\ldots,s_n](\bot))| \\
=&~ |\sum_{i=1}^{n} \big(\,\gamma(C_i[t_i])(\bot) - \gamma(C_{i}[s_i])(\bot)\,\big)| \\ 
\leq&~ \sum_{i=1}^{n} |\gamma(C_i[t_i])(\bot) - \gamma(C_{i+1}[s_i])(\bot)|\\
\leq&~ \sum_{i=1}^n d_\ctx^\pBCK(t_i,s_i).
\end{align*}
Since the above inequality holds for all contexts $C[\cdot]$, we conclude
\[ d_\ctx^\pBCK(\f(t_1,\ldots,t_n),\f(s_1,\ldots,s_n))\leq \sum_{i=1}^n d_\ctx^\pBCK(t_i,s_i), \]
proving that $d_\ctx^\pBCK$ is $\bar\Sigma$-congruence.
\item Every adequate $\barSigma$-congruence $d$ on $(\Lambda,\ini)$ satisfies $d\sqleq d_\ctx^\pBCK$: Let $C[\cdot]$ be a context and $t,s\in \Lambda$. Let us first observe that $d(C[t],C[s])\leq d(t,s)$. To see this, we proceed by structural induction on the context $C[\cdot]$. For the empty context $C=[\cdot]$ the statement is trivial. Thus suppose that $C=\f(t_1,\ldots,t_{i-1},C'[\cdot],t_{i+1},\cdots,t_n)$ for some $\f$ in $\Sigma$ and some context $C'[\cdot]$. Then
\[ d(C[t],C[s])\leq d(C'[t],C'[s]) + \sum_{j\neq i} d(t_j,t_j) \leq d(t,s)\]
where the first uses that $d$ is a $\barSigma$-congruence and the second one follows by induction; note that $d(t_j,t_j)=0$ because $d$ is a $\barSigma$-congruence and thus a pseudometric.

By adequacy of $d$ is now follows that
\[ |\gamma(C[t])(\bot)-\gamma(C[s])(\bot)| \leq d(C[t],C[s])\leq d(t,s),  \]
hence $d_\ctx^\pBCK(t,s)\leq d(t,s)$ for all $t,s$, which proves $d\sqleq d_\ctx^\pBCK$.\qedhere
\end{enumerate}

\subsection*{Proof of \Cref{lem:rho-pSKI-non-exp}}
For ease of notation, we denote by $\barR$ the derived relations on
 \[(\Lambda,R)^{(\Lambda,R)},\quad
\barR(\{\bot\}+(\Lambda,R)^{(\Lambda,R)}),\quad
 \barB((\Lambda,R),(\Lambda,R)),\quad \barB((\Lambda,R), \barB((\Lambda,R),(\Lambda,R))), \quad \barSigmas(\Lambda,R),\]
etc. Non-expansivity of $\rho$ means that if $\f$ in $\Sigma$ and  $(t_i,\phi_i,\Phi_i),(t_i',\phi_i',\Phi_i')\in \Lambda\times {B}(\Lambda,\Lambda)\times  {B}(\Lambda, {B}(\Lambda,\Lambda))$,
\begin{align*}
& R(t_i,t_i) \text{ and }  \barR(\phi_i,\phi_i') \text{ and }  \barR(\Phi_i,\Phi_i')\quad \text{for all $i=1,\dots,n$}\\
\implies~&
 \barR(\rho(\f((t_1,\phi_1,\Phi_1),\ldots, (t_n,\phi_n,\Phi_n))), \rho(\f((t_1',\phi_1',\Phi_1'),\ldots, (t_n',\phi_n',\Phi_n')))).  \end{align*}
For all operations $\f$ except $S$, $S'$, $S''$ this follows from \Cref{lem:rho-pBCK-non-exp} restricted to relations ($\{0,1\}$-valued fuzzy relations).
For $\f\in \{S,S',S''\}$ the map $\rho$ is defined as
\[ \rho(\f((t_1,-,-),\ldots,(t_n,-,-))) = 1\cdot (t\mapsto T)    \]
where $T=T[t,t_1,\ldots,t_n]\in \Sigmas\Lambda$ is a term in the variables $t,t_1,\ldots,t_n\in \Lambda$.  Now let $t_i,t_i'\in \Lambda$ and suppose that $R(t_i,t_i')$ for $i=1,\ldots,n$. By definition of the relation $\barR$ of the term algebra $\barSigma^\star(\Lambda,R)$, we have for all $t,t'\in \Lambda$ that
\[ R(t,t')\implies \barR(T[t,t_1,\dots,t_n], T[t',t_1',\ldots,t_n']).\]
This means that the functions $t\mapsto T[t,t_1,\ldots,t_n]$ and $t\mapsto T[t,t_1',\ldots,t_n']$ are related in the hom-set relation from \Cref{ex:liftings-rel}\ref{ex:liftings:hom-functor}. Therefore, $\barR(1\cdot (t\mapsto T[t,t_1,\ldots,t_n]),1\cdot (t\mapsto T[t,t_1'\ldots,t_n']))$ and so $\barR(\rho(\f((t_1,-,-),\ldots,(t_n,-,-)),\rho(\f((t_1',-,-),\ldots,(t_n',-,-))))$, as required.

\end{document}